%% file: main.tex
\newif\ifTR
\TRtrue

\ifTR
  \documentclass[acmsmall,nonacm]{acmart}
  \makeatletter
  \let\@authorsaddresses\@empty
  \makeatother
\else
  \documentclass[acmsmall]{acmart}
\fi

\usepackage{braket}
\usepackage[noend, ruled,linesnumbered]{algorithm2e}
\usepackage[capitalize]{cleveref}
\usepackage{paralist}
\usepackage{thm-restate}
\usepackage{tikz}
\usepackage{todonotes}
\usepackage{xspace}
\usepackage{subcaption}
\usepackage{wrapfig}
\usetikzlibrary{arrows,quantikz2,calc,automata}
\pgfdeclarelayer{bg}
\pgfdeclarelayer{bg1}
\pgfsetlayers{bg1,bg,main}
\usepackage{adjustbox}
\usepackage{cleveref}

\input{paroshMacros}

\input{macros}
\input{tikzlibrarymyautomata.code.tex}

\usepackage[xcolor,leftbars]{changebar}
\cbcolor{red}
\nochangebars

\AtEndPreamble{%
\theoremstyle{acmdefinition}
\newtheorem{remark}[theorem]{Remark}}

\ifTR
\else
\setcopyright{cc}
\setcctype{by}
\acmDOI{10.1145/3776712}
\acmYear{2026}
\acmJournal{PACMPL}
\acmVolume{10}
\acmNumber{POPL}
\acmArticle{70}
\acmMonth{1}
\received{2025-07-10}
\received[accepted]{2025-11-06}

\fi


\title{Parameterized Verification of Quantum Circuits}

\author{Parosh Aziz Abdulla}
\orcid{0000-0001-6832-6611}
\affiliation{%
  \institution{Uppsala University}
  \city{Uppsala}
  \country{Sweden}
}
\affiliation{%
  \institution{Mälardalen University}
  \city{Västerås}
  \country{Sweden}
}
\email{parosh@it.uu.se}

\author{Yu-Fang Chen}
\orcid{0000-0003-2872-0336}
\affiliation{%
  \institution{Academia Sinica}
  \city{Taipei}
  \country{Taiwan}
}
\email{yfc@iis.sinica.edu.tw}

\author{Michal Hečko}
\orcid{0009-0003-2428-8547}
\affiliation{%
  \institution{Brno University of Technology}
  \city{Brno}
  \country{Czech Republic}
}
\email{ihecko@fit.vut.cz}

\author{Lukáš Holík}
\orcid{0000-0001-6957-1651}
\affiliation{%
  \institution{Brno University of Technology}
  \city{Brno}
  \country{Czech Republic}
}
\affiliation{%
  \institution{Aalborg University}
  \city{Aalborg}
  \country{Denmark}
}
\email{holik@fit.vutbr.cz}

\author{Ondřej Lengál}
\orcid{0000-0002-3038-5875}
\affiliation{%
  \institution{Brno University of Technology}
  \city{Brno}
  \country{Czech Republic}
}
\email{lengal@fit.vutbr.cz}

\author{Jyun-Ao Lin}
\orcid{0000-0001-8560-2147}
\affiliation{%
  \institution{National Taipei University of Technology}
  \city{Taipei}
  \country{Taiwan}
}
\email{jalin@ntut.edu.tw}

\author{Ramanathan S. Thinniyam}
\orcid{0000-0002-9926-0931}
\affiliation{%
  \institution{Uppsala University}
  \city{Uppsala}
  \country{Sweden}
}
\email{ramanathan.s.thinniyam@it.uu.se}

\renewcommand{\shortauthors}{P. A. Abdulla, Y. Chen, M. Hečko, L. Holík, O. Lengál, J. Lin, R. S. Thinniyam}

\begin{document}

\ifTR
\title[Parameterized Verification of Quantum Circuits (Technical Report)]
{Parameterized Verification of Quantum Circuits\\(Technical Report)}
\else
\title{Parameterized Verification of Quantum Circuits}
\fi

\begin{abstract}
We present the first fully automatic framework for verifying relational properties of \emph{parameterized quantum programs}, i.e., a program that, given an input size, generates a corresponding quantum circuit.
We focus on verifying input-output correctness as well as equivalence. At the core of our approach is a new automata model, synchronized weighted tree automata (SWTAs), which compactly and precisely captures the infinite families of quantum states produced by parameterized programs. We introduce a class of transducers to model quantum gate semantics and develop composition algorithms for constructing transducers of parameterized circuits. Verification is reduced to functional inclusion or equivalence checking between SWTAs, for which we provide decision procedures. Our implementation demonstrates both the expressiveness and practical efficiency of the framework by verifying a diverse set of representative parameterized quantum programs with verification times ranging from milliseconds to seconds.
\end{abstract}

\maketitle

\vspace{-0.0mm}
\section{Introduction}
\vspace{-0.0mm}

Quantum technology is advancing at an unprecedented pace and has the potential to reshape numerous sectors at both national and global levels. As quantum computers become increasingly complex and widely deployed, ensuring the correctness of programs that run on them becomes a~matter of critical importance. Errors in quantum programs can have serious consequences, including incorrect outcomes in cryptographic tasks, unnecessary consumption of quantum resources, and the misinterpretation of experimental data. 
Over the past decade, quantum program verification has emerged as a vibrant research area and has seen impressive advances. Notable developments include rigorous proof systems, automated reasoning tools, and foundational theories addressing various quantum paradigms; see, e.g.,~\cite{BauerMarquartLS23,chen2025verifying,fang2024symbolic,yu2021quantum,perdrix2008quantum,zhou2019applied,unruh2019quantum,feng2021quantum,ying2012floyd,liu2019formal,amy2018towards,Chareton2021,ChenCLLTY23,abdulla2025verifying}. These works have collectively laid the groundwork for understanding and verifying quantum programs in practical and theoretical settings. 

Most useful quantum algorithms (e.g., Shor's, Grover's, QFT) are inherently \emph{parameterized}\footnote{In quantum computing, ``parameterized circuits/programs'' often refers to circuits with parameters such as rotation angles. Here, we use the term to mean circuits parameterized by input size, a usage more familiar in the verification community.}, meaning they are designed to operate correctly for arbitrary input sizes. Writing a separate quantum program for each input size of an algorithm is both inefficient and fundamentally unscalable. A more effective approach is to design a \emph{parameterized program} that takes the input size $n$ and dynamically generates the corresponding quantum circuit.%

Yet, existing verification approaches are very limited in their ability to automatically \textbf{verify parameterized programs}. 
Even though notable progress in this direction has been made by the \qbricks~\cite{Chareton2021} and \autoq~\cite{ChenCLLTY23,abdulla2025verifying} projects, \qbricks is not fully automated, and \autoq is limited to a narrow class of parameterized circuits.
As of now, a fully automatic verification approach capable of handling general parameterized quantum programs remains an open problem. Achieving this goal requires overcoming fundamental challenges, and we address those in this paper.

We particularly focus on automated verification of \emph{relational properties}, aka \emph{functional verification}, that is particularly meaningful and broadly applicable in the context of parameterized quantum programs. Given a set of quantum state pairs $(p, q)$, the task is to determine whether the program, for each input size, correctly maps input state $p$ to the expected output state $q$. This captures \emph{input-output correctness} across all parameterized instances.

The main challenges can be narrowed down to delivering three key components of the functional verification framework: 
\begin{inparaenum}[(i)]
  \item  a formal model for specifying the property of interest---namely, an
    infinite set of quantum state pairs;
  \item  a language for describing the parameterized quantum program; and
  \item  an efficient algorithm for checking whether the program satisfies the given property.
\end{inparaenum}
We answer these challenges in this paper and present, for the first time, a complete and efficient verification framework tailored for parameterized functional verification.

\textbf{First, at the core of our approach is a novel class of automata, which we refer to as \emph{synchronized weighted tree automata (SWTAs)}}. 
This model extends standard tree automata~\cite{tata} with two key features: \emph{colors}, for synchronization
\begin{changebar}
(inspired by~\cite{LSTA2025}),
\end{changebar}
and \emph{weights}, to encode quantum \emph{amplitudes}\footnote{A quantum amplitude is a complex number that influences the chance of seeing a particular result when we observe the quantum state.}. This combination yields a powerful and compact formalism that can represent rich families of quantum state spaces---including those parameterized by input size and those with exponentially many distinct amplitudes.
Notably, unlike standard tree automata, which merely characterize a set of quantum states (trees), SWTAs also define a function from \emph{color sequences} to quantum states. This functional view is central to enabling relational verification.
\begin{changebar}
Since SWTAs are a~new formal model, we also explore its basic properties, such
as decidability and complexity of emptiness checking (\pspace-complete; \cref{sec:emptiness}),
inclusion and equivalence checking (both undecidable; \cref{sec:language_equivalence}), and
closure properties (closed under union, not closed under intersection and
complement; \cref{sec:boolean}).
\end{changebar}

\textbf{Second, to reason about circuit behaviors, we introduce a class of transducers that operate on SWTAs.}
These transducers serve as semantic models capable of describing the behavior of individual quantum gates, circuits composed of finitely many components, and---crucially---\textit{parameterized programs}.
Our transducers provide a uniform representation of standard quantum gates\footnote{Quantum gates are the fundamental operations in quantum programs, responsible for transforming quantum states from one form to another.}, including all single-qubit gates, controlled gates, and their compositions.
Notably, the encoding is efficient.
As a case in point, the Quantum Fourier Transform (QFT) gate can be captured using a transducer with a polynomial number of states, thereby avoiding the anticipated exponential blow-up.
Importantly, applying a transducer that encodes a parameterized quantum program to an SWTA preserves the color sequence associated with each state. This preserved sequence acts as a~bridge, linking states in two SWTAs that represent the quantum system before and after executing the parameterized program, hence allowing the capture of input-output relations. 

We also develop a comprehensive toolkit to support the construction of the transducer of quantum programs. 
Central to this is a \emph{composition operator} that enables the automatic assembly of complex transducers from simpler components. Given a quantum program and transducers specifying the behavior of individual gates, the composition operator automatically constructs a transducer representing the behavior of the entire program. More generally, it can compose the transducers for two sub-circuits to produce one for the entire circuit.
We further extend this operator to parameterized programs: starting from a transducer for a finite circuit, we automatically generate a~transducer for a parameterized version consisting of \mbox{an arbitrary number of repeated components.}

\textbf{Third, we develop decision procedures for checking equivalence and entailment of SWTAs}. That is, given SWTAs $\aut$ and $\autb$, we can check that they define the same function, denoted $\semof{\aut} = \semof{\autb}$, or that one is included in the other, denoted $\semof{\aut} \subseteq \semof{\autb}$. 

These algorithms complete the foundation of a framework for solving the relational verification problem. Specifically:
\begin{enumerate}
    \item The desired relational property is specified using a pair of SWTAs, $\aut$ and $\autb$, with related quantum states linked via shared color sequences;
    \item The parameterized quantum program is encoded as a transducer and applied to $\aut$ to produce a new SWTA, $\aut'$;
    \item Verification is then reduced to checking the inclusion $\semof{\aut'} \subseteq \semof{\autb}$.
\end{enumerate}
%
Our approach also supports \textbf{checking the \emph{equivalence} of two parameterized quantum programs}---that is, verifying whether they produce identical outputs on all valid inputs, regardless of the input size. This capability is particularly important for ensuring the correctness of program optimizations. The idea behind our method is straightforward:
\begin{enumerate}
    \item We begin by constructing an SWTA $\aut$ that encodes all \emph{computational basis states}, or more generally, any set of $2^n$ linearly independent quantum states, for every input size $n$.
    \item We then apply the transducers corresponding to the two parameterized programs to $\aut$, producing two output SWTAs, $\autb$ and $\autb'$.
    \item Finally, we check whether the two resulting state sets are equal, i.e., whether $\semof{\autb} = \semof{\autb'}$.
\end{enumerate}
This procedure is sound because quantum programs (or circuits) implement linear transformations. If two such programs produce the same outputs when applied to a complete set of linearly independent inputs, then they are equivalent.

We implemented our framework as a C++ tool named \tool~\cite{tool} and evaluated it on a~diverse set of parameterized quantum verification tasks. For functional verification, we applied \tool to the Bernstein-Vazirani algorithm, an arithmetic ripple-carry adder, and a syndrome extraction circuit from a quantum error-correcting code. For equivalence checking, we used \tool to verify the correctness of optimized implementations of Grover's search algorithm and Hamiltonian simulation by comparing them against their unoptimized, parameterized counterparts.

These benchmarks span a broad spectrum of parameterized quantum circuits,
including algorithmic, arithmetic, simulation, and error-correction components,
demonstrating the generality and expressiveness of our approach.
\tool successfully verified all examples efficiently, with runtimes ranging
from 14\,ms (Bernstein-Vazirani) to 11\,s (Adder).
Overall, these results showcase the practical effectiveness, versatility, and
scalability of our verification framework.
\begin{changebar}
To the best of our knowledge, we are the first to fully automatically verify
parameterized versions of the circuits.
\end{changebar}

\newcommand{\figQuantumStates}[0]{
\begin{wrapfigure}[16]{r}{0.5\textwidth}
\ifTR\else
\vspace*{-5mm}
\fi

\begin{center}
$
  \begin{bmatrix}
    \invsqrttwo  &0 & 0 & \smallfrac{\sqrt 3}{4} & 0 &\smallfrac 1 4 & \smallfrac 1 4 & \smallfrac{\sqrt 3}{4}
  \end{bmatrix} ^T
$

\vspace{3mm}

$
\invsqrttwo \ket{000} + \smallfrac{\sqrt 3}{4} \ket{011} + \smallfrac 1 4 \ket{101} + \smallfrac 1 4 \ket{110} + \smallfrac{\sqrt 3}{4} \ket{111}
$

\includegraphics[scale=1.2]{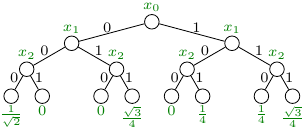}
\end{center}
\vspace{-4mm}
\caption{A quantum state involving three qubits. We provide the vector, Dirac, and tree-based representations.
The tree spans all basis states, with the corresponding amplitudes shown as leaf labels.
For instance, the left-most path represents $\ketof{000}$ with amplitude $\invsqrttwo$.
}
\label{triple:qbit:state:fig}
\end{wrapfigure}
}

\vspace{-0.0mm}
\section{Quantum States and Tree Automata}\label{sec:overview}
\vspace{-0.0mm}

We present a step-by-step overview of quantum states and our automata model.
We will illustrate the concepts through a series of simple examples.
We begin by introducing a tree-like representation of quantum states, followed by a tree automata-based framework for characterizing (potentially infinite) sets of such states.
Next, we introduce a formalism based on \emph{(tree) transducers} to capture the transition relation induced by a quantum circuit on sets of quantum states.
%

\vspace{-0.0mm}
\subsection{Quantum States}
\vspace{-0.0mm}

\ifTR\else
\figQuantumStates  
\fi
Classical and quantum systems differ fundamentally in how they represent and evolve states.
A classical state of a system with $n$ bits is a single element from the set $\{0,1\}^n$.
At any given time, the system is in exactly one state (e.g., $0101$ or $1110$).
Quantum computers use qubits
\ifTR
\figQuantumStates  
\fi
(quantum bits) to store and process information. 
In the quantum setting, classical states are referred to as \emph{(computational) basis states}.
An $n$-qubit system is described by a quantum state belonging
to a $2^n$-dimensional space, and  a state is a~\emph{superposition} of all $2^n$ basis states.
Each basis state, usually written as in the \emph{Dirac notation} $\ket x$, where $x\in\{0,1\}^n$, corresponds to a definite state of all $n$ qubits (with each qubit being either $\ket 0$ or $\ket 1$).
A general state is written as
$\ketof\psi=\sum_{\xvar\in\{0,1\}^n}\alpha_\xvar\ketof\xvar$, where 
each $\xvar$ is a basis state, and the coefficients $\alpha_\xvar$, called the \emph{amplitudes}, are complex numbers satisfying the normalization condition: $\sum_{\xvar\in\{0,1\}^n}|\alpha_\xvar|^2 =1$.
\cref{triple:qbit:state:fig} depicts a quantum state in a system with $3$ qubits.
In the quantum computing literature, states are sometimes written as column vectors of length $2^n$, with complex entries corresponding to the amplitudes of the respective basis states.
To simplify notation, we often use the transpose of this column vector.
For example, we write $\big[\invsqrttwo\hspace{2mm}\invsqrttwo\big]^T$ to represent the
state $\invsqrttwo \ket 0 + \invsqrttwo \ket 1$.
This results in a more compact row vector form.
In our verification framework, we will represent the state of $n$ qubits as a perfect binary tree of depth $n$, where each leaf represents the amplitude of the basis state we get by traversing the tree from the root to the leaves (\cref{triple:qbit:state:fig}).
The tree representation is closely related to the vector representation of quantum states.
A tree encodes the vector by concatenating the amplitudes of the leaves in which they occur from left to right.

\newcommand{\figPlainTreeAut}[0]{
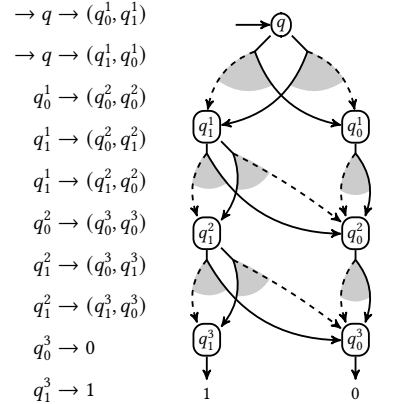
\begin{wrapfigure}[21]{r}{0.34\textwidth}
\ifTR\else
\vspace*{2mm}
\fi
\hspace{-5mm}
\scalebox{0.8}{
  \begin{minipage}{30mm}
  \vspace*{-8mm}
  \begin{align*}
    \initmark q & \to (q_0^1, q_1^1) \\[1mm]
    \initmark q & \to (q_1^1, q_0^1)
    \\[1mm]
              q_0^1 & \to (q_0^2, q_0^2) \\[1mm]
              q_1^1 & \to (q_0^2, q_1^2)
    \\[1mm]
              q_1^1 & \to (q_1^2, q_0^2)\\[1mm]
              q_0^2 & \to (q_0^3, q_0^3)
    \\[1mm]
              q_1^2 & \to (q_0^3, q_1^3)\\[1mm]
              q_1^2 & \to (q_1^3, q_0^3)
    \\[1mm]
              q_0^3 & \to 0\\[1mm]
              q_1^3 & \to 1
  \end{align*}
  \end{minipage}
}
\scalebox{0.7}{
  \begin{minipage}{20mm}
  \input{Figures/all-basis-ta.tikz}
  \end{minipage}
}
\vspace{-6mm}
\caption{A plain tree automaton that characterizes the set of all $3$-qubit
  basis states.
  In the figure, transitions are represened by hyperedges (represented by grey
  arcs) with dashed arrows going to the left-hand child and solid arrows going
  to the right-hand child.
  }
\label{all:aut:fig}
\end{wrapfigure}%
}

\vspace{-0.0mm}
\subsection{Tree Automata}
\vspace{-0.0mm}

\ifTR

\begin{wrapfigure}[21]{r}{0.34\textwidth}
\ifTR\else
\vspace*{2mm}
\fi
\hspace{-5mm}
\scalebox{0.8}{
  \begin{minipage}{30mm}
  \vspace*{-8mm}
  \begin{align*}
    \initmark q & \to (q_0^1, q_1^1) \\[1mm]
    \initmark q & \to (q_1^1, q_0^1)
    \\[1mm]
              q_0^1 & \to (q_0^2, q_0^2) \\[1mm]
              q_1^1 & \to (q_0^2, q_1^2)
    \\[1mm]
              q_1^1 & \to (q_1^2, q_0^2)\\[1mm]
              q_0^2 & \to (q_0^3, q_0^3)
    \\[1mm]
              q_1^2 & \to (q_0^3, q_1^3)\\[1mm]
              q_1^2 & \to (q_1^3, q_0^3)
    \\[1mm]
              q_0^3 & \to 0\\[1mm]
              q_1^3 & \to 1
  \end{align*}
  \end{minipage}
}
\scalebox{0.7}{
  \begin{minipage}{20mm}
  \input{Figures/all-basis-ta.tikz}
  \end{minipage}
}
\vspace{-6mm}
\caption{A plain tree automaton that characterizes the set of all $3$-qubit
  basis states.
  In the figure, transitions are represened by hyperedges (represented by grey
  arcs) with dashed arrows going to the left-hand child and solid arrows going
  to the right-hand child.
  }
\label{all:aut:fig}
\end{wrapfigure}%
\fi
We use (tree) automata to encode both finite and infinite sets of quantum states.
We will introduce our tree automata model in three steps, namely, plain (standard) automata~\cite{ChenCLLTY23}, level-synchronized (colored) tree automata (LSTAs)~\cite{abdulla2025verifying}, and weighted synchronized tree automata (SWTAs).
A plain automaton adheres to the classical definition, comprising a finite set of states and transition rules.
Notice that the classical definition attaches a label to each transition. To simplify the presentation of the overview section, we retain only the leaf transition labels and omit the non-leaf labels.
We use
\ifTR\else

\begin{wrapfigure}[21]{r}{0.34\textwidth}
\ifTR\else
\vspace*{2mm}
\fi
\hspace{-5mm}
\scalebox{0.8}{
  \begin{minipage}{30mm}
  \vspace*{-8mm}
  \begin{align*}
    \initmark q & \to (q_0^1, q_1^1) \\[1mm]
    \initmark q & \to (q_1^1, q_0^1)
    \\[1mm]
              q_0^1 & \to (q_0^2, q_0^2) \\[1mm]
              q_1^1 & \to (q_0^2, q_1^2)
    \\[1mm]
              q_1^1 & \to (q_1^2, q_0^2)\\[1mm]
              q_0^2 & \to (q_0^3, q_0^3)
    \\[1mm]
              q_1^2 & \to (q_0^3, q_1^3)\\[1mm]
              q_1^2 & \to (q_1^3, q_0^3)
    \\[1mm]
              q_0^3 & \to 0\\[1mm]
              q_1^3 & \to 1
  \end{align*}
  \end{minipage}
}
\scalebox{0.7}{
  \begin{minipage}{20mm}
  \input{Figures/all-basis-ta.tikz}
  \end{minipage}
}
\vspace{-6mm}
\caption{A plain tree automaton that characterizes the set of all $3$-qubit
  basis states.
  In the figure, transitions are represened by hyperedges (represented by grey
  arcs) with dashed arrows going to the left-hand child and solid arrows going
  to the right-hand child.
  }
\label{all:aut:fig}
\end{wrapfigure}%
\fi
the leaf transition labels to represent amplitudes.
The non-leaf labels are assumed to be identical and irrelevant at this stage.
These non-leaf symbols will later be used to represent variable names and, in some cases, to avoid non-deterministic transitions.
As usual, the automaton accepts a given tree if it is possible to label the nodes of the tree with states of the automaton such that the root%
is labeled by the root state and the children of each node are labeled in accordance with the transition rules. 
\cref{all:aut:fig} depicts an automaton accepting trees that characterize all basis quantum states over three qubits. 
Intuitively, the states labeled~$q_0^i$ accept trees whose leaves all have amplitude~$0$, while the states labeled~$q_1^i$ accept trees with exactly one leaf having amplitude~$1$, and all other leaves having amplitude~$0$.
Analogously, we can construct automata that handle any given number~$n$ of qubits.
A~key observation is that plain tree automata allow us to represent a set containing exponentially many quantum states using only a~linear number of automaton states.
%
On the other hand, plain automata cannot capture infinite sets of quantum states.
The reason for this is that we represent states as \emph{perfect} binary trees, and representing an infinite set of such trees requires a~synchronization mechanism: all leaf transitions must occur at the same level.

\begin{wrapfigure}[14]{r}{67mm}
\vspace*{-2mm}
\begin{subfigure}[b]{25mm}
\scalebox{0.8}{
\input{Figures/lsta-ex.tikz}
}
\caption{}
\label{label}
\end{subfigure}
\begin{subfigure}[b]{20mm}
\scalebox{0.8}{
  \begin{minipage}{20mm}
  \vspace*{-8mm}
  \begin{align*}
    \initmark q_0 & \to (q_1, q_1) \\[1mm]
              q_1 & \newlabeledto{\tacc 1} (q_4, q_4) \\[1mm]
              q_1 & \newlabeledto{\tacc 2} (q_2, q_3) \\[1mm]
              q_2 & \to 0 \\[1mm]
              q_3 & \to \invsqrttwo \\[1mm]
              q_4 & \to \invtwo
  \end{align*}
  \end{minipage}
}
\caption{}
\label{label}
\end{subfigure}
\begin{subfigure}[b]{20mm}
  \includegraphics[width=20mm,keepaspectratio]{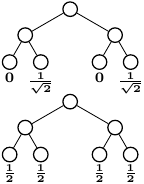}
\caption{}
\label{fig:lsta-two-trees}
\end{subfigure}
\vspace{-8mm}
\caption{An LSTA. (a)~A picture of the automaton. (b)~The set of transitions.
  (c)~The two trees accepted by the automaton.}
\label{color:aut:fig}
\end{wrapfigure}%
To handle such limitations, \cite{abdulla2025verifying}~introduces the LSTA model in which each transition is annotated with a color.
The set of colors is used to enforce level-synchronization: all transitions applied at a given level of the accepted tree must be labeled with the same color.
LSTAs can represent sets of trees more compactly than plain tree automata.
Let us consider the example shown in \cref{color:aut:fig}.
The automaton accepts the two trees depicted in \cref{fig:lsta-two-trees} as follows.
It begins with a~single transition from the root state~$q_0$; the color of this transition is immaterial.
The transition $\trntransof{q_0}{}{q_1}{q_1}$ indicates that, at the
next level, the state~$q_1$ is used to generate both the left and the right
subtree.
Note that $q_1$ has two transitions, one colored
\begin{changebar}
red (and decorated by~$\tacc 1$) and the other blue (decorated by $\tacc 2$).
\end{changebar}
However, due to
the color restriction, we must apply either the red transition to both children or the blue transition to both.
We cannot mix the colors, since all transitions applied at a~given level must share color (red or blue, in this case).

\begin{wrapfigure}[6]{r}{40mm}
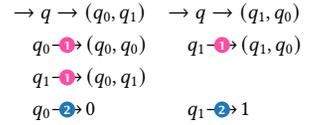

\vspace{-9mm}
\scalebox{0.8}{
\begin{minipage}{50mm}
\begin{align*}
  \initmark q & {}\to (q_0, q_1) &
  \initmark q & {}\to (q_1, q_0) \\
  q_0 & {}\newlabeledto{\tacc 1} (q_0, q_0) &
  q_1 & {}\newlabeledto{\tacc 1} (q_1, q_0) \\
  q_1 & {}\newlabeledto{\tacc 1} (q_0, q_1) \\
  q_0 & {}\newlabeledto{\tacc 2} 0 &
  q_1 & {}\newlabeledto{\tacc 2} 1
\end{align*}
\end{minipage}
}
\vspace{-3mm}
\caption{An LSTA recognizing the set of all basis states.}
\label{param:all:color:fig}
\end{wrapfigure}
The next example, depicted in \cref{param:all:color:fig},   demonstrates the expressive power of LSTAs compared to plain tree automata.
The LSTA accepts the set of all basis states $\{\ket{x}\mid x\in\{0,1\}^*\}$ for an arbitrary number $n$ of qubits.
Recall that the plain automaton in \cref{all:aut:fig} characterizes the set of all basis states for the fixed case $n = 3$.
To represent the basis states for arbitrary $n$, we would need infinitely many plain automata---one for each value of $n$.
In contrast, a single LSTA suffices to describe the
entire set.
We~use the state~$q_0$ to generate trees in which all leaves have amplitude~$0$, and the state $q_1$ to generate trees with exactly one leaf of amplitude~$1$, while all other leaves have amplitude~$0$.
The states~$q_0$ and~$q_1$ allow transitions that are either red or blue.
Intuitively, red transitions are applied to inner nodes, while blue transitions are applied at the leaf level.
Since all transitions applied at the same level must have the same color, we ensure that only perfect binary trees are generated: we apply red transitions repeatedly, level by level, to build the tree, and at some point, switch to blue transitions to generate the leaves.

\newcommand{\figExSWTA}[0]{
\begin{wrapfigure}[10]{r}{65mm}
\vspace{-3mm}
\begin{subfigure}[b]{24mm}
\scalebox{0.8}{
\begin{minipage}{25mm}
\begin{align*}
  \initmark q & {}\newlabeledto{\tacc 1} (\invsqrttwo q, \invsqrttwo q) \\
  \initmark q & {}\newlabeledto{\tacc 2} (\invsqrttwo \ell, \invsqrttwo \ell)
\end{align*}
\end{minipage}
}
\caption{}
\label{fig:weight-swta}
\end{subfigure}
\begin{subfigure}[b]{12mm}
\centering
\includegraphics[scale=1]{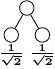}
\caption{}
\label{fig:weight-tree1}
\end{subfigure}
\begin{subfigure}[b]{22mm}
\centering
\includegraphics[scale=1]{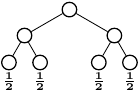}
\caption{}
\label{fig:weight-tree2}
\end{subfigure}
\vspace{-3mm}
\caption{(\subref{fig:weight-swta}) An SWTA recognizing the set of all quantum states with uniform superposition; $\ell$ is a leaf state. (\subref{fig:weight-tree1}) and (\subref{fig:weight-tree2}) show the trees with one qubit resp.\ two qubits in uniform superposition.}
\label{uniform:fig}
\end{wrapfigure}
}

\ifTR\else
\figExSWTA  
\fi
Finally, we introduce \emph{synchronized weighted tree automata}, and demonstrate their superiority over the LSTA model.
We begin by considering the set of all quantum basis states in \emph{uniform superposition}.
These are the states where all leaves have identical amplitudes.
\cref{fig:weight-tree1,fig:weight-tree2} depict the trees corresponding to the
cases where the number of qubits is $1$ and $2$, respectively.
For $n$ qubits, the amplitudes are equal to $2^{-\frac{n}{2}}$, i.e., they depend on the number of qubits.
This uniform superposition set cannot be captured by any finite set of LSTAs, since we would need one set of transitions to express the amplitudes for each given value of $n$.
To overcome this limitation, we introduce \emph{synchronized weighted tree automata (SWTAs)}, and show how weights can be used---together with colors---to generate all quantum states in uniform superposition.
Recall that in both plain and colored automata, a~transition of the form $\trntransof q {} {q_1}{q_2}$ defines the set of trees where the left and right subtrees are generated from the states $q_1$ and $q_2$, respectively.
In contrast, SWTA allows for a~more general form of transition, where the subtrees are generated using \emph{linear combinations} of states.
\ifTR
\figExSWTA  
\fi
These linear combinations are constructed by \emph{adding} states and \emph{multiplying} them by complex number \emph{constants}.
The SWTA depicted in \cref{uniform:fig} provides a simple example.
To construct the left and right subtrees, we (i)~recursively generate trees from the state~$q$, and (ii)~multiply the resulting trees by the scalar~$\invsqrttwo$.
This scaling means that each leaf of the subtree is multiplied by the constant~$\invsqrttwo$.
As~before, the construction can be repeated an arbitrary number of times using red transitions, after which we apply the blue transition to generate the leaf nodes labeled $\ell$.
A leaf node has the weight~$1$ by default.
This is in contrast to the standard notion of tree automata, where leaf rules explicitly assign values to leaves.
In our setting, the use of weights makes explicit leaf values unnecessary.
Notice that weights and colors operate \emph{in tandem} to generate the desired set of trees.
Without colors, we cannot enforce the level-wise synchronization required by the uniform superposition; and without weights, we cannot \mbox{generate the correct leaf amplitudes.}
%


\newcommand{\figTransHX}[0]{
\begin{wrapfigure}[14]{r}{42mm}
\ifTR
\else
\vspace{2mm}
\fi
\begin{subfigure}[b]{42mm}
\centering
\begin{quantikz}
  \lstick{$x_1$} & \gate{X} & \qw\\[-8mm]
\end{quantikz}
\scalebox{0.8}{
  \begin{minipage}{30mm}
    \begin{align*}
    \initmark \trntransof{q_0}{}{q_1(\rightT)}{q_1(\leftT)}
    \end{align*}
  \end{minipage}
}
\vspace{-1mm}
\caption{The $\notgate$ gate and its transducer}
\label{fig:X-gate-trans}
\end{subfigure}

\vspace{2mm}
\noindent
\begin{subfigure}[b]{42mm}
\centering
\begin{quantikz}
  \lstick{$x_0$} & \gate{H} & \qw\\[-8mm]
\end{quantikz}
\scalebox{0.8}{
  \begin{minipage}{30mm}
    \begin{align*}
      \initmark s_0 \to (&\invsqrttwo s_1(\leftT) + \invsqrttwo s_1(\rightT),\\
      &\invsqrttwo s_1(\leftT) - \invsqrttwo s_1(\rightT))
    \end{align*}
  \end{minipage}
}
\vspace{-1mm}
\caption{The $\hgate$ gate and its transducer}
\label{fig:H-gate-trans}
\end{subfigure}
\vspace{-8mm}
\caption{Basic gates and their transducers. Leaf states are~$q_1$ and~$s_1$.
}
\label{notH:transducer:fig}
\end{wrapfigure}
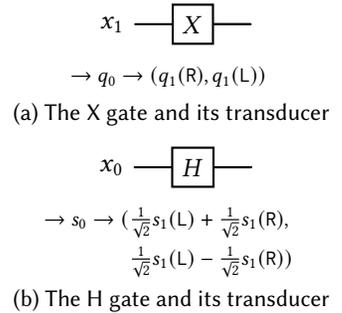
}
%
\vspace{-0.0mm}
\subsection{Transducers}\label{sec:label}
\vspace{-0.0mm}
%
We model the behavior of a quantum circuit as a~transition relation over a~set of
trees.
This transition relation is encoded as a transducer that transforms a tree representing the input quantum state into a tree representing the result of applying the circuit to that input. 
Since a~tree 
encodes a~quantum state by storing the sequence of amplitudes at its leaves, the transducer specifies how these amplitudes are modified to produce the output tree.
A transducer consists of a finite set of states and a collection of transition rules.
To initiate the transformation, the root of the input tree is labeled with a root state.
Each state, assigned to a node in the input tree, prescribes how to construct the left and right subtrees of the output tree based on the structure of the corresponding input subtrees.
The output tree is constructed by applying linear combinations---defined by the transducer's transition rules---to the subtrees at each node.
The right-hand side of a transition rule

\figTransHX
\noindent
is a pair that specifies how to generate the left and right subtrees of the output tree.
This is done using a set of \emph{ground terms}, each of the form $q(\leftT)$ or $q(\rightT)$.
Intuitively, $q(\leftT)$ means applying the state $q$ to the left subtree of the input, thereby producing a corresponding output subtree.
Similarly, $q(\rightT)$ refers to applying $q$ to the right subtree.
A transition defines the new left and right subtrees as linear combinations of such ground terms.
We begin by describing transducers that implement individual quantum gates.
The $\notgate$ gate (quantum NOT) inverts a single qubit by swapping its amplitudes~(\cref{fig:X-gate-trans}).
The corresponding transducer consists of two states, $q_0$ and $q_1$, together with one transition.
The transformation begins by labeling the root of the input tree with the root state $q_0$.
To construct the left subtree of the output, we apply state $q_1$ to the right subtree of the input; conversely, we apply $q_1$ to the left subtree to construct the right subtree.
The leaf state does not alter the amplitudes---it simply multiplies them by $1$.
As a result, the transducer swaps the amplitudes of the two leaves, which captures the behavior of the $\notgate$ gate.

The $\hgate$ gate (Hadamard) creates superpositions, enabling quantum algorithms to explore multiple states in parallel.
The transducer, depicted in~\cref{fig:H-gate-trans} has two states: the root state $q_0$ and the leaf state $q_1$.
To construct the left subtree of the output, we apply $q_1$ to both the left and right subtrees of the input, multiply each resulting subtree by $\invsqrttwo$, and then sum them.
To construct the right subtree, we again apply $q_1$ to both subtrees, multiply each by $\invsqrttwo$, but this time subtract the right subtree from the left.
This transformation faithfully models the Hadamard gate, which maps the basis states $\ket{0}$ and $\ket{1}$ to superpositions $\frac{\ket{0} + \ket{1}}{\sqrt{2}}$ and $\frac{\ket{0} - \ket{1}}{\sqrt{2}}$, respectively.

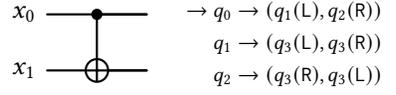
\begin{wrapfigure}[7]{r}{52mm}
\vspace{-2mm}
\begin{quantikz}
  \lstick{$x_0$} & \ctrl 1 & \qw\\
  \lstick{$x_1$} & \targ{} & \qw\\[-8mm]
\end{quantikz}
\quad
\scalebox{0.8}{
  \begin{minipage}{30mm}
    \begin{align*}
      \initmark q_0 & {}\to (q_1(\leftT), q_2(\rightT))\\
      q_1 & {}\to (q_3(\leftT), q_3(\rightT))\\
      q_2 & {}\to (q_3(\rightT), q_3(\leftT))
    \end{align*}
  \end{minipage}
}
\vspace{-2mm}
\caption{ The $\cnotgate$ gate and its transducer;
$q_3$ is the leaf state.
}
\label{cnot:transducer:fig}
\end{wrapfigure}
The $\cnotgate$ gate is a two-qubit gate that flips (i.e., applies the NOT operation to) the target qubit if and only if the control qubit is in the quantum basis state $\ket{1}$.
If the control qubit is in the quantum basis state $\ket{0}$, the target qubit remains unchanged.
In terms of tree representation, the transducer corresponding to the $\cnotgate$ gate preserves the left subtree (representing the control qubit), while it swaps the two subtrees of the right subtree (representing the target qubit).
This conditional swap encodes the semantics of the $\cnotgate$ gate: the output is \mbox{modified only when the control is in the quantum basis state $\ket{1}$.}

\vspace{-0.0mm}
\subsection{Composition}\label{sec:label}
\vspace{-0.0mm}

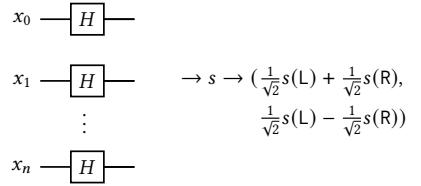
\begin{wrapfigure}[9]{r}{55mm}
\newcommand{\makedots}[0]{\lstick[label style={xshift=4mm,yshift=-6mm}]{\vdots}}
\vspace*{-4mm}
\scalebox{0.8}{
\begin{quantikz}
  \lstick{$x_0$} & \gate H & \qw\\
  \lstick{$x_1$} & \gate H\makedots & \qw\\[4mm]
  \lstick{$x_n$} & \gate H & \qw
\end{quantikz}
}
\quad\!\!
\scalebox{0.8}{
  \begin{minipage}{30mm}
    \begin{align*}
      \initmark s \to (&\invsqrttwo s(\leftT) + \invsqrttwo s(\rightT),\\
      &\invsqrttwo s(\leftT) - \invsqrttwo s(\rightT))
    \end{align*}
  \end{minipage}
}
\vspace*{-3mm}
\caption{The parameterized $\hgate$ gate transducer.
The state $s$ is a leaf state.
}
\label{paramH:transducer:fig}
\end{wrapfigure}
In the paper, we will introduce a set of composition operators on SWTAs and transducers.
The first operator applies a~transducer~$T$ to an SWTA~$\aut_1$ to produce a~new SWTA~$\aut_2$.
Here, $\aut_1$~describes a~\emph{precondition}, i.e., a set of input trees (quantum states), and~$T$ specifies how the circuit transforms these input trees to obtain the \emph{postcondition}, i.e., the resulting set of output trees.
The second operator composes a~sequence of transducers---each representing an individual gate or a~subcircuit---into a single transducer that captures the behavior of the entire circuit.
The third operator applies to a given circuit and captures the \emph{parameterized} circuit obtained by repeating the circuit an unbounded number of times.
For instance, given the transducer for a single $\hgate$ gate (as shown in~\cref{fig:H-gate-trans}), this operator produces a parameterized transducer that represents the repeated application of an arbitrary number of $\hgate$ gates (\cref{paramH:transducer:fig}).

Finally, \textbf{functional and equivalence verification} can be performed using all of the components introduced above, together with the SWTA functional equivalence and inclusion procedures described in \cref{sec:colored_equivalence_checking}. At this point, we should have a preliminary understanding of the overall framework. In the next step, we will provide more rigorous exposition.
%

\vspace{-0.0mm}
\section{Formal Definitions}\label{sec:label}
\vspace{-0.0mm}

\paragraph{Basics.}
We use $\naturals$ to denote $\{0, 1, \ldots\}$, $\integers$ to denote
integers, and $\complex$ to denote complex numbers.
Given a~partial function $f\colon A \partialto B$, we use $f(a) = \bot$ to
denote that the value of~$f$ for $a \in A$ is undefined
\begin{changebar}
and $\domof f$ to denote the set $\{a \in A \mid f(a) \neq \bot\}$.
\end{changebar}
A~\emph{linear form} $\linof v$ over $Q= \{q_1,\ldots, q_n\}$ is a non-empty partial map
from $Q$ to $\complex$ usually given by a~formal sum $\linof v =\sum_{q_i \in
P} a_{q_i}\cdot q_i$ with $P \subseteq Q$ and each $a_{q_i} \in \complex$ (we
often omit the operator~`$\cdot$' and the coefficient `1' if present, e.g., we
may write ``$3q_1 + q_2$''; the linear form~``0'' represents an empty
map).
We emphasize that we distinguish the cases when $q_i \notin P$ and $q_i \in P$
with $a_{q_i} =0$, e.g., ``$3 q_1$'' and ``$3 q_1 + 0 q_2$'' are two different things.
We will write~$\linindexof v q$ to denote~$a_q$ (if $q \notin P$, then
$\linindexof v q = \bot$) and use~$\linformsof Q$ to denote the set of all
linear forms over~$Q$.
The \emph{state support} of~$\linof v$ is defined as $\suppof{\linof v} = P$.

\paragraph{Words.}
An \emph{alphabet} is a~finite non-empty set of \emph{symbols}.
Given an alphabet~$A$, a \emph{word} over~$A$ is a finite sequence of symbols
$w = a_1 \ldots a_n$ with $a_i \in A$ for all $1 \leq i \leq n$.
For two words~$u = a_1 \ldots a_n$ and~$v=b_1 \ldots b_m$ over~$A$, we use $u
\concat v$ (or just the juxtaposition $uv$) to denote the word $a_1 \ldots a_n
b_1 \ldots b_m$, with the neutral element~$\epsilon$ (the empty string).
Let~$B$, $C$ be sets of words over~$A$.
Then $B \concat C$ (or just the juxtaposition $BC$) is defined as the set $\{u
\concat v \mid u\in B, v \in C\}$.
We use $B^0$ to denote the set $\{\epsilon\}$ and
for $i \in \naturals$, we use $B^{i+1}$ to denote the set $B \concat B^i$.
The notation
$B^{< n}$ is used for
$\bigcup_{0 \leq i < n} B^i$ and $B^*$ denotes~$\bigcup_{0 \leq i} B^i$.
Let us fix a~(finite non-empty) alphabet~$\abc$ such that $\abc \cap
\complex = \emptyset$.



\paragraph{Trees.}
A \emph{(perfect finite) binary tree}%
\footnote{We emphasize that all trees considered in this
paper are \emph{perfect}, i.e., all branches have the same height. This is in
contrast to the definition of trees in, e.g., \cite{tata}}
over $\abc$ is a partial function 
$t\colon \{0,1\}^* \partialto (\abc \cup \complex)$
such that the following conditions hold:
\begin{enumerate}
    \cbstart
  \item $t$ is a~finite, non-empty, and prefix-closed partial function and
    \cbend
  \item there exists a unique $h \in \mathbb{N}$, called the \emph{height} of
    $t$ and denoted as $h(t)$, such that for all positions $p \in \domof{t}$:
    \begin{enumerate}
        \item if $|p| = h$, then $t(p) \in \complex$ (leaf nodes),
        \item if $|p| < h$, then $t(p) \in \abc$ (internal nodes), and
        \item if $|p| > h$, then $t(p) = \bot$ (undefined).
    \end{enumerate}
\end{enumerate}
%
We use~$\abctrees$ to denote the set of all trees over~$\abc$.
For a word $u \in \{0,1\}^*$ representing a~\emph{branch} of~$t$, the
\emph{subtree} rooted at $u$ is given as $\subtreeof t u = \{v \mapsto t(uv)
\mid uv \in \domof t\}$.
For trees $t_0, t_1 \in \abctrees$ of the same height and a~label $a \in \abc$,
we define the tree constructor
\begin{equation}
\treecons a {t_0}{t_1} = \{ \epsilon \mapsto a \} \cup \{ 0u \mapsto t_0(u)
  \mid u \in \domof{t_0} \} \cup \{ 1u \mapsto t_1(u) \mid u \in \domof{t_1} \}.
\end{equation}
%
Trees $t, t' \in \abctrees$ are \emph{compatible}  iff 
$\height{t} = \height{t'}$ and for all $p \in
\{0,1\}^{<\height{t}}$ it holds that $t(p) = t'(p)$.
%
Given compatible trees $t, t' \in \abctrees$ and $a\in\complex$, we define the
following operations on trees:
\begin{align}
t + t' &= \{ p \mapsto t(p) \mid p \in \{0,1\}^{<\height{t}} \} \cup \{ p \mapsto t(p) + t'(p) \mid p \in \{0,1\}^{\height{t}} \} \quad\text{and}\\ 
a \cdot t &= \{ p \mapsto t(p) \mid p \in \{0,1\}^{<\height{t}} \} \cup \{ p \mapsto a \cdot t(p) \mid p \in \{0,1\}^{\height{t}} \}.
\end{align}
Intuitively, $t+t'$ puts the trees over each other and sums up the values in
the leaves and $a\cdot t$ multiplies all values in the leaves of~$t$ by~$a$.
The result of $t+t'$ for incompatible trees is~$\bot$ and if any of the
operands is~$\bot$, the result is also~$\bot$.

\newdef{Exepriment:\\
\underline{Operations over linear forms:} 
Given a linear form $\linof \ell = \sum_{1\leq i \leq n} c_i.x_i$ and an assignment $\nu:\{x_1,\ldots,x_n\}\rightarrow D$, 
%
%
$\linof{\ell} (\nu)$ is the linear form obtained from $\ell$ by substituting all $x_i$ by their values $\nu(x_i)$.
$\sum \linof \ell$ is the value obtained by evaluating the additions and multiplications in the formal sum, assuming, that the needed multiplication and addition is defined (semering, comutative, ..?).  
Given a second linear form $\linof r = \sum_{1\leq j \leq m} d_j.y_j$, we define their product as the linear form $\linof \ell . \linof r = \sum_{1\leq i\leq n}\sum_{1\leq j\leq m} c_i.d_j.x_i.y_j$.
}

\vspace{-0.0mm}
\subsection{Synchronized Weighted Tree Automata}\label{sec:label}
\vspace{-0.0mm}

A (finite binary) \emph{synchronized weighted tree automaton} (SWTA) over
$\abc$ is a~tuple
$\aut = \tuple{Q, \delta, \colors, \rootstate, E}$ where
$Q$ is a~finite set of \emph{states},
$\colors = \{\tacc 1, \ldots, \tcacc 5 n\}$ is a~(finite non-empty) set of \emph{colors},
$\rootstate \in Q$ is the \emph{root state},
$E \subseteq Q$ is a~set of \emph{leaf states}, and $\delta\colon Q \times \abc
\times \colors \partialto (\linformsof Q \times \linformsof Q)$ is a (top-down)
partial \emph{transition function}.
We also use
$\transof q a {\tcacc 7 c}{\linof \ell}{\linof r}$
to denote that $\delta(q, a, \tcacc 7 c) = (\linof \ell, \linof r)$, where
$\linof \ell, \linof r \in \linformsof Q$.
An example of an SWTA transition is $\transof {q_1} a {\tacc 1} {\invsqrttwo q_1 +
\invsqrttwo q_2}{\invsqrttwo q_2 + q_3}$.

The \emph{tree function of a~state~$q \in Q$ of~$\aut$} is a~(partial) function $\semof{\aut, q} \colon
(\abc \times \colors)^* \partialto \abctrees$ defined inductively as follows:
%
\begin{enumerate}
  \item  (base case)
    $\semofof{\aut, q} \epsilon$ is the leaf~1 if~$q$ is a leaf state, otherwise
    the value is undefined. Formally,
    \begin{equation}
      \semofof{\aut,q} \epsilon = \begin{cases}
        \{\epsilon \mapsto 1\} & \text{if } q \in E, \\
        \bot & \text{otherwise.}
      \end{cases}
    \end{equation}

  \item (inductive case)
    $\semofof{\aut,q}{\tuple{a, \tcacc 7 c}\concat u} = t$
    for $a \in \abc$, $\tcacc 7 c \in \colors$, and $u \in (\abc \times
    \colors)^*$, where~$t$ is defined in the following way.
    \newdef{Experiment: Let $\transof q a {\tcacc 7 c}{\linof \ell}{\linof r} \in \delta$ with and let $\nu:Q\rightarrow \abctrees$ maps every state $q\in Q$ to $\semofof{\aut, q} u$. Then 
      $t_{\linof \ell} = \sumlfwrt {\linof \ell} {\nu}$
   and
      $t_{\linof r}  = \sumlfwrt {\linof r} {\nu}$.}
    Let $\transof q a {\tcacc 7 c}{\linof \ell}{\linof r} \in \delta$ with
    \begin{equation}
      \supp(\linof \ell) = \{q^\ell_1, \dots, q^\ell_m\}\quad\text{and}\quad
      \supp(\linof r) = \{q^r_1, \dots, q^r_n\}.
    \end{equation}
    Then $t = \treecons a {t_\ell} {t_r}$ where
    \begin{equation}
      t_\ell = \sum_{i=1}^m \linof{\ell}[q^\ell_i] \cdot \semofof{\aut, q^\ell_i} u
      \quad \text{and} \quad
      t_r = \sum_{j=1}^n \linof{r}[q^r_j] \cdot \semofof{\aut, q^r_j} u.
    \end{equation}
    Intuitively, $t$ connects the $a$-labeled root to the sub-trees $t_\ell$ and
    $t_r$, constructed from the linear forms $\linof \ell$ and $\linof r$,
    respectively. In each case, every state $q$ occurring in the linear form is substituted by the tree $\semofof{\aut,
    q} u$, resulting in a linear form over trees, \mbox{which is then summed up.}
\end{enumerate}
The \emph{tree function of~$\aut$}, written $\semof \aut$, is a~function~$\semof
\aut\colon (\abc \times \colors)^* \partialto \abctrees$ given as $\semof
\aut = \semof{\aut, \rootstate}$.
Note that this definition differs slightly from the simplified version used in the introduction, which describes a tree function as a mapping from color sequences to quantum states. The formal definition maps \emph{symbol-color} sequences to quantum states.

\begin{example}\label{ex:semantics}
Let~$\autex$ be an SWTA defined by the following transitions (with the set
of leaf states~$\{u,v\}$):
\begin{equation}
\begin{aligned}
      \initmark & \transcof q a 1 {r+s}{r-s} \quad
    && \transcof r a 1 {2\finstof u}{0 \finstof u} \quad
    && \transcof s a 1 {\finstof u+ \finstof v}{0\finstof v } \\
    \initmark & \transcof q a 2 {r-s}{r+s}
    &&
    \transcof r a 2 {0\finstof u}{\invtwo \finstof u}
    &&
    \transcof s a 2 {\finstof u - \finstof v}{\finstof u - \smallfrac 3 2 \finstof v}
\end{aligned}
\end{equation}
Consider a~symbol-color sequence $w = \tuple{a, \tacc 1}\tuple{a, \tacc 2}$
and let us compute the tree  $t = \semofof \autex {w}$.
For the first symbol-color pair $\tuple{a, \tacc 1}$ and the root state~$q$,
the relevant transition is $\transcof q a 1 {r+s}{r-s}$.
Let us therefore compute the sub-trees $t_r = \semofof{\autex, r}{\tuple{a, \tacc
2}}$ and $t_s = \semofof{\autex, s}{\tuple{a, \tacc 2}}$.
For~$t_r$, the relevant transition is
$\transcof r a 2 {0\finstof u}{\invtwo \finstof u}$.
Since~$u$ is a~leaf state, we can conclude that
$t_r = \treecons a {\{\epsilon \mapsto 0\}} {\{\epsilon \mapsto \invtwo\}}$, or
\smalltreeof a 0 \invtwo.
On the other hand, for~$t_s$, the relevant transition is
$\transcof s a 2 {\finstof u - \finstof v}{\finstof u - \smallfrac 3 2 \finstof v}$.
Since~$u$ and~$v$ are both leaf states, we have that
$t_s = \treecons a {\{\epsilon \mapsto 1 - 1\}}{\{\epsilon \mapsto 1 - \smallfrac 3 2\}} =
\treecons a {\{\epsilon \mapsto 0\}}{\{\epsilon \mapsto -\invtwo\}}$, or
\smalltreeof a 0 {\text{-}\invtwo}.
Having computed~$t_r$ and~$t_s$, we can now construct the final tree~$t$ (using the $q$-transition above) as
$\treecons a {t_r + t_s}{t_r - t_s} = \bigtreeof a a 0 0 0 1$.
In a~similar manner,
one can easily compute the remaining values of~$\semof \autex$, summarized below:
\begin{equation}
\begin{aligned}
  \semofof \autex {\tuple{a,\tacc 1}\tuple{a,\tacc 1}} ={}& \bigtreeof a a 4 0 0 0 \qquad&
  \semofof \autex {\tuple{a,\tacc 1}\tuple{a,\tacc 2}} ={}& \bigtreeof a a 0 0 0 1 \\
  \semofof \autex {\tuple{a,\tacc 2}\tuple{a,\tacc 1}} ={}& \bigtreeof a a 0 0 4 0 &
  \semofof \autex {\tuple{a,\tacc 2}\tuple{a,\tacc 2}} ={}& \bigtreeof a a 0 1 0 0\\[-5mm]
\end{aligned}
\end{equation}
\qed
\end{example}

We say that \emph{$t$ is accepted from~$q$ in $\aut$}, denoted as $t \in
\langof{\aut, q}$, if there exists a~sequence $w \in (\abc \times \colors)^*$ such
that $\semofof{\aut,q} w = t$, and say that \emph{$t$ is accepted by~$\aut$}
if $t \in \langof{\aut, \rootstate}$.
The \emph{language} of~$\aut$, denoted as $\langof \aut \subseteq \abctrees$, is
the set of all trees accepted by~$\aut$.
We say that two SWTAs~$\aut$ and~$\but$ are \emph{functionally equivalent} if $\semof
\aut = \semof \but$, \emph{functionally included} if $\semof
\aut \subseteq \semof \but$ (we see $\semof \aut$ and $\semof \but$ as the
sets representing the functions), and that they are \emph{language equivalent} if $\langof
\aut = \langof \but$.
We note that every two functionally equivalent SWTAs are also language
equivalent, but not necessarily the other way round.

\begin{example} 
Consider the SWTA~$\autex$ from \cref{ex:semantics} and the following
SWTA~$\butex$ with the set of leaf states~$\{h\}$:
\begin{equation}
\begin{aligned}
 \initmark & \transcof f a 1 {4g}{0h} \quad
          && \transcof g a 1 {0\finstof h}{\smallfrac 1 4 \finstof h} \quad
          && \transcof k a 1 {4\finstof h}{0\finstof h} \quad
          && \transcof h a 1 h h \\
 \initmark & \transcof f a 2 {0h}{k}
          && \transcof g a 2 {\finstof h}{0 \finstof h}
          && \transcof k a 2 {0\finstof h}{\finstof h}
          && \transcof h a 2 h h 
\end{aligned}
\end{equation}
The tree function of~$\butex$ can be computed to be the following:
\begin{equation}
\begin{aligned}
  \semofof \butex {\tuple{a,\tacc 1}\tuple{a,\tacc 1}} ={}& \bigtreeof a a 0 1 0 0 \qquad&
  \semofof \butex {\tuple{a,\tacc 1}\tuple{a,\tacc 2}} ={}& \bigtreeof a a 4 0 0 0 \\
  \semofof \butex {\tuple{a,\tacc 2}\tuple{a,\tacc 1}} ={}& \bigtreeof a a 0 0 4 0 &
  \semofof \butex {\tuple{a,\tacc 2}\tuple{a,\tacc 2}} ={}& \bigtreeof a a 0 0 0 1
\end{aligned}
\end{equation}
We see that $\langof \autex = \langof \butex = \left\{
\bigtreeof a a 0 1 0 0,
\bigtreeof a a 4 0 0 0,
\bigtreeof a a 0 0 4 0,
\bigtreeof a a 0 0 0 1
\right\}$, but $\semof \autex \neq \semof \butex$, since, e.g., 
$\semofof \autex {\tuple{a,\tacc 1}\tuple{a,\tacc 1}} \neq
\semofof \butex {\tuple{a,\tacc 1}\tuple{a,\tacc 1}}$.
\qed
\end{example}

\subsection{Weighted Tree Transducers}
\vspace{-0.0mm}
A (finite binary two-tape) \emph{weighted tree transducer} (WTT) over alphabet~$\abc$ is
a~tuple $\trn = \tuple{Q, \delta, \rootstate, E}$ where
$Q$ is a~finite set of \emph{states},
$\rootstate \in Q$ is the \emph{root state},
$E \subseteq Q$ is a~set of \emph{leaf states}, and
$\delta\colon Q \times \abc \partialto (\linformsof {Q(\leftT,\rightT)} \times \linformsof {Q(\leftT,\rightT)})$ is a~\emph{transition function}
where $Q(\leftT,\rightT)$ is the set of ground terms of the form $q(\leftT)$ and $q(\rightT)$, for $q\in Q$ (the symbols~$\leftT$ and~$\rightT$ used here denote the $\underline{\leftT}$eft and the
$\underline{\rightT}$ight subtree of the input tree respectively).
An example of a~WTT transition is
$\trntransof q a {\invsqrttwo q(\leftT) + p(\rightT)}{-\invsqrttwo q(\leftT) - 3q(\rightT)}$.

For every state~$q \in Q$, the transducer defines a~(partial) unary function
over trees $\trn_q\colon \abctrees \partialto \abctrees$ defined
below.
First, for a~linear form $\linof x \in
\linformsof {Q(\leftT,\rightT)}$ and a~pair of trees $t_\leftT$
and~$t_\rightT$, we use $\linof x(t_\leftT,t_\rightT)$ to denote the tree
obtained by replacing each~$q(\leftT)$ and~$q(\rightT)$ in~$\linof x$ by the
tree~$\trn_q(t_\leftT)$ and~$\trn_q(t_\rightT)$, respectively, and summing up the
resulting linear form over trees (the result can be undefined due to
incompatibility).
%
%
%
Then~$\trn_q$ is defined inductively in the following way:
\begin{enumerate}
  \item  (base case) If $t = \{\epsilon\mapsto a\}$ (i.e., $t$~is
    a~leaf) then if~$q \in E$, then $\trn_q (t) = t$, else 
    $\trn_q(t) = \bot$ (i.e., it is undefined).
  \item  (inductive case) If $t \supset \{\epsilon\mapsto a\}$ (i.e., $t$ is not a~leaf) then
  if $\delta(q,a) = \bot$ (i.e., there is no transition for $q$ and $a$) then $\trn_q (t) = \bot$, else if $\delta(q,a) =  (\linof{\ell},\linof{r})$ then 
 $\trn_q (t) = \treecons {a} {\linof\ell(\subtreeof t 0,\subtreeof t 1)} {\linof r(\subtreeof t 0,\subtreeof t 1)}$ (and it is undefined one of the linear forms evaluates to $\bot$). 
That is, the two sub-trees of the resulting tree are both transducer images of
    the sub-trees
    \begin{changebar}
    $\subtreeof t 0$ and $\subtreeof t 1$
    \end{changebar}
    of the original tree w.r.t.\ the linear forms.
\end{enumerate}
The transducer~$\trn$ then defines the function given as~$\trn_\rootstate$.

\newdef{Experiment:
For every state~$q \in Q$, the transducer defines a~(partial) unary function
over trees $\trn_q\colon \abctrees \partialto \abctrees$ defined
below.
Then~$\trn_q$ is defined inductively in the following way:
\begin{enumerate}
  \item  (base case) If $t = \{\epsilon\mapsto a\}$ (i.e., $t$~is
    a~leaf) then if~$q \in E$, then $\trn_q (t) = t$, else 
    $\trn_q(t) = \bot$ (i.e., it is undefined).
  \item  (inductive case) If $t \supset \{\epsilon\mapsto a\}$ (i.e., $t$ is not a~leaf) then there are two cases.
  First, if $\delta(q,a) = \bot$ (i.e., there is no transition for $q$ and $a$) then $\trn_q (t) = \bot$. Otherwise, let $\delta(q,a) =  (\linof{\ell},\linof{r})$ and let $\nu = \{q(\leftT)\mapsto \trn_q(\subtreeof t 0),q(\rightT)\mapsto \trn_q(\subtreeof t 1)\mid q\in Q\}$. Then 
 $\trn_q (t) = \treecons {a} {\sumlfwrt {\linof\ell} {\nu}} {\sumlfwrt {\linof r} {\nu}}$. 
That is, the two sub-trees of the resulting tree are both transducer images of
    the sub-trees of the original tree w.r.t.\ the linear forms.
\end{enumerate}
The transducer~$\trn$ then defines the function given as~$\trn_\rootstate$.
}

\begin{example}
Let $\trnex$ be a~WTT defined as follows with the set of leaf states~$\{p\}$:
\begin{equation}
\begin{aligned}
\initmark & \trntransof p a {\invsqrttwo z(\leftT) + \invsqrttwo z(\rightT)}{\invsqrttwo z(\leftT) - \invsqrttwo z(\rightT)} \qquad
  &&
  \trntransof z a {p(\leftT)}{-p(\rightT)}
\end{aligned}
\end{equation}
Further, let $t$ be the tree \bigtreeof a a 0 0 0 1.
Let us now compute the tree $t' = \trnex(t)$.
We start with the transition
$\trntransof p a {\invsqrttwo z(\leftT) + \invsqrttwo z(\rightT)}{\invsqrttwo
z(\leftT) - \invsqrttwo z(\rightT)}$.
Here, $t_\ell = \smalltreeof a 0 0$ and $t_r = \smalltreeof a 0 1$.
Let us now compute the values of $\trn_z(t_\ell)$ and $\trn_z(t_r)$ using the
transition $\trntransof z a {p(\leftT)}{-p(\rightT)}$.
Since~$p$ is a~leaf state, the values are 
$t_\ell' = \trn_z(t_\ell) = \smalltreeof a 0 0$ and
$t_r' = \trn_z(t_r) = \smalltreeof a 0 {\text{-}1}$.
The result of the linear form $\invsqrttwo z(\leftT) + \invsqrttwo z(\rightT)$
with $t_\ell'$ substituted for $z(\leftT)$ and $t_r'$ substituted for
$z(\rightT)$ is then the tree \smalltreeof a 0 {\text{-}\invsqrttwo} and,
similarly, the result of $\invsqrttwo z(\leftT) - \invsqrttwo z(\rightT)$ is
the tree \smalltreeof a 0 {\invsqrttwo}.
Finally, $t' = \bigtreeof a a 0 {\text{-}\invsqrttwo} 0 {\invsqrttwo}$.
\qed
\end{example}

%
For a~tree language $\lang \subseteq \abctrees$, the \emph{image
of~$\lang$ in~$\trn$} is the tree language $\trnimageof \trn \lang =
\{\trnimageof \trn t \mid t \in \lang\}$.

\vspace{-0.0mm}
\section{Parameterized Verification of Quantum Circuits}\label{sec:use-cases}
\vspace{-0.0mm}

Let us now introduce our verification framework.
As mentioned before, we use SWTAs to encode sets of quantum states and WTTs to
encode quantum gates.
Let us now describe the following two approaches to verification of
parameterized circuits:
\begin{enumerate}
  \item  \emph{Relational verification}:
    Here, we are given two SWTAs, $\aut_\pre$ and $\aut_\post$, which represent
    the quantum states in the precondition and in the postcondition
    respectively, such that the relation (which input quantum state should be
    mapped to which output quantum state) is captured by the corresponding trees
    having the same color sequence.
    The circuit is given by a~sequence of WTTs $\trn_1, \ldots, \trn_k$.
    In the verification, we compute the SWTA $\autRes =
    \trn_k(\trn_{k-1}(\ldots(\trn_1(\aut_\pre))\ldots))$ using a~sequence of
    transducer image computation steps (\cref{sec:trn_image}) and then test
    $\semof{\autRes} \subseteq \semof{\aut_\post}$ using the algorithm in
    \cref{sec:colored_equivalence_checking}.

  \item  \emph{Equivalence checking}:
    We are given two sequences of transducers $\trn_1, \ldots,
    \trn_k$ and $\trn'_1, \ldots, \trn'_l$.
    In the verification, we start with the SWTA~$\autbases$ that represents
    all~$2^n$ computational bases for every size~$n$.
    We compute $\autRes = \trn_k(\trn_{k-1}(\ldots(\trn_1(\autbases))\ldots))$
    and
    $\autRes' = \trn'_k(\trn'_{k-1}(\ldots(\trn'_1(\autbases))\ldots))$
    and test~$\semof{\autRes} = \semof{\autRes'}$.
    Correctness follows from linearity of the operations and orthogonality of
    the vectors for the computational bases (from linear algebra, if two unitary operators behave the
    same on~$2^n$ orthogonal vectors, then they are equal).
\end{enumerate}
Below, we show how this approach can be applied on selected case studies of
parameterized circuits.

\newcommand{\figBV}[0]{
\begin{wrapfigure}[12]{r}{0.44\textwidth}
    \centering
    \vspace*{-9mm}
    \hspace*{-6mm}
    \begin{minipage}{\linewidth}
    \scalebox{0.74}{
      \input{Figures/bv-circuit.tikz}
    }
    \end{minipage}
    \vspace{-3mm}
    \caption{Circuit implementing the Bernstein-Vazirani algorithm for the
    secret $(\texttt{10})^* (\texttt{1} + \varepsilon)$}
    \label{fig:BV}
\end{wrapfigure}
}

\vspace{-0.0mm}
\subsection{The Bernstein-Vazirani Algorithm}\label{sec:BV}
\vspace{-0.0mm}

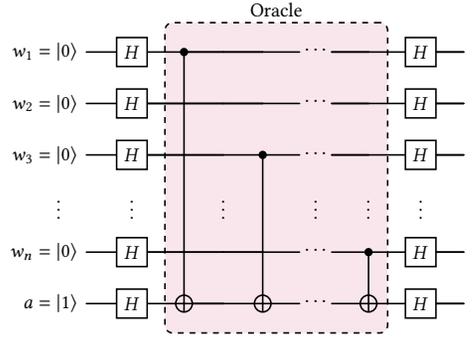
\begin{wrapfigure}[12]{r}{0.44\textwidth}
    \centering
    \vspace*{-9mm}
    \hspace*{-6mm}
    \begin{minipage}{\linewidth}
    \scalebox{0.74}{
      \input{Figures/bv-circuit.tikz}
    }
    \end{minipage}
    \vspace{-3mm}
    \caption{Circuit implementing the Bernstein-Vazirani algorithm for the
    secret $(\texttt{10})^* (\texttt{1} + \varepsilon)$}
    \label{fig:BV}
\end{wrapfigure}
In the first example, we show how to use our framework to model a~family of
circuits implementing the Bernstein-Vazirani (BV) algorithm~\cite{BernsteinV93} and
verify their functional correctness (the verification will be described in~\cref{sec:colored_equivalence_checking}).
The family of circuits that we consider is parameterized by~$n \in \naturals$ such that
for every such~$n$, there will be a~circuit with $n+1$ qubits (the one additional
qubit is an ancilla).
Each of the circuits implements the BV algorithm where the secret key is of the
form $(\texttt{10})^* (\texttt{1} + \varepsilon) = \{\varepsilon, \texttt{1},
\texttt{10}, \texttt{101}, \ldots \}$ such that its length is~$n$ (there is
exactly one such a~secret for any~$n$).
The schema for the circuits is given in \cref{fig:BV}.
We note that the oracle is specialized for the given family of secret keys.
The circuit starts, for $n$~qubits, with the quantum state $\ket{w_1 \dots w_n a} = \ket{0 \dots
0 1}$ where $w_i$ are working qubits and $a$ is an ancilla, and
the secret key will be $s_1 \ldots s_n = 101\ldots$
The expected output of the circuit is the quantum state $\ket{s_1 \ldots s_n
1}$, i.e., the working qubits contain the~secret key.

We model the problem by providing the following components:
\begin{inparaenum}[(i)]
  \item  SWTAs $\aut_\pre$ and $\aut_\post$ that encode the pre- and
    post-conditions,
  \item  transducer $\trn_{\hadtensor}$ that represents applying
    the $\had$ gate on all qubits, and
  \item  transducer $\transfCX$ representing the oracle.
\end{inparaenum}
They are defined as follows:
\begin{enumerate}
  \item  $\aut_\pre$ is an SWTA that accepts perfect trees where the branch
    $0\ldots 01$ has the value 1 and all other branches have the value~0, which
    can be achieved using the following two transitions (with the root
    state~$s_1$ and the set of leaf states~$\{s_2\}$):
    $\{\transcof{s_1} w 1 {s_1} {0s_1}, \transcof{s_1} a 1 {0s_2} {s_2}\}$.
    We emphasize that the symbol~$w$ is used to match \emph{any} of the working
    qubits~$w_1, \ldots, w_n$.

  \item  $\aut_\post$ is an SWTA that accepts all perfect trees where the branch
    $1010\ldots x 1$ (with~$x = 0$ for an even~$n$ and $x=1$ for an odd~$n$) has
    the value~1 and all other branches have the value zero, which is defined as
    follows (with the root state~$g$ and the set of leaf states~$\{c\}$):
    \begin{align}
         \initmark & \transcof{g} w 1 {0h}{h} \quad  &
         \initmark & \transcof{g} a 1 {0c}{c} \quad &
                  & \transcof{h} w 1 {g}{0g} \quad &&
                    \transcof{h} a 1 {0c}{c}
    \end{align}

  \item  $\trn_{\hadtensor}$ is a~WTT with the two transitions (and with~$u$
    being both a root and a~leaf state)
    \begin{equation}
    \begin{aligned}
      \initmark & \trntransof u w {\invsqrttwo u(\leftT) + \invsqrttwo u(\rightT)}{\invsqrttwo u(\leftT) - \invsqrttwo u(\rightT)}\\
      \initmark & \trntransof u a {\invsqrttwo u(\leftT) + \invsqrttwo u(\rightT)}{\invsqrttwo u(\leftT) - \invsqrttwo u(\rightT)}.
    \end{aligned}\label{eq:transH}
    \end{equation}
    We refer the reader to \cref{sec:quantumWTT} for details on how it can be
    obtained.

  \item  $\transfCX$ is a~WTT that models applying a~series of $\cnot$ gates
    with controls on odd working qubits~$w_{2i+1}$ and the ancilla~$a$ being the target.
    Our construction of $\transfCX$ is based on the key observation that the
    ancilla qubit should be flipped only when there is an odd number of
    $1$'s in the secret.
    Therefore, the $\transfCX$'s states have the form $r^i$ (going to read a
    working qubit with secret value $1$ or the ancilla qubit) and $s^i$ (going
    to read a working qubit with secret value~$0$ or the ancilla qubit), for $i \in \{0, 1\}$,
    which represents whether the automaton has already seen even (for $i = 0$)
    or odd (for $i = 1$) number of ones along the tree branch. The state~$r^0$
    is the root state and state~$l$ is the only leaf state.
    \noindent
    \hspace*{-15mm}
    \scalebox{0.97}{
    \begin{minipage}{1.05\textwidth}
    \begin{equation}
    \begin{aligned}
         \initmark& \trntransof{r^0} w {s^0(\leftT)}{s^1(\rightT)} &
                  & \trntransof{s^0} w {r^0(\leftT)}{r^0(\rightT)} &
                  & \trntransof{s^1} w {r^1(\leftT)}{r^1(\rightT)} &
                  & \trntransof{r^1} w {s^1(\leftT)}{s^0(\rightT)} \\
         \initmark& \trntransof{r^0} a {l(\leftT)}{l(\rightT)} &
                  & \trntransof{s^0} a {l(\leftT)}{l(\rightT)} &
                  & \trntransof{s^1} a {l(\rightT)}{l(\leftT)} &
                  & \trntransof{r^1} a {l(\rightT)}{l(\leftT)}
    \end{aligned}
    \end{equation}
    \end{minipage}
    }
\end{enumerate}

\medskip
We can now establish correctness of the model by computing
$\autRes = \trn_{\hadtensor}(\transfCX(\trn_{\hadtensor}(\aut_\pre)))$
and testing whether $\semof \autRes = \semof{\aut_\post}$.
An illustration of how $\trn_{\hadtensor}(\transfCX(\trn_{\hadtensor}(\aut_\pre)))$
is constructed is given in \trOrAppendix{appendix:bv}.

We note that our choice of a fixed secret family is not the only way of
modeling the circuit. Instead, it is possible to model a BV circuit where the
secret is arbitrary. In such a case, the
\begin{changebar}
secret is a~part
\end{changebar}
of the input in which the working qubits are interleaved with with secret
qubits.

\newcommand{\figAdder}[0]{
\begin{figure}[t]
  \begin{minipage}[b]{0.23\textwidth}
    \begin{subfigure}[t]{\textwidth}
      \centering
      \scalebox{0.9}{
        \begin{quantikz}
          & \qw          & \targ{}  & \ctrl{2} & \qw      \\
          & \targ{}      & \qw      & \ctrl{1}  & \qw      \\
          & \ctrl{-1}    &\ctrl{-2} & \targ{}   &\qw    
        \end{quantikz}
      }
      \caption{The $\maj$ gate}
      \label{fig:maj}
    \end{subfigure}
    \vspace{1.3cm} 

    \begin{subfigure}[b]{\textwidth}
      \scalebox{0.9}{
      \begin{quantikz}
        & \ctrl{2} & \targ{}  & \ctrl{1} & \qw     \\
        & \ctrl{1} & \qw      & \targ{}  & \qw    \\
          & \targ{}  &\ctrl{-2} & \qw    & \qw  
      \end{quantikz}      
      }
      \caption{The $\uma$ gate}
      \label{fig:uma}
    \end{subfigure}
  \end{minipage}
  \hfill
  \begin{minipage}[b]{0.76\textwidth}
    \begin{subfigure}[b]{\textwidth}
\scalebox{0.67}{
  \input{Figures/adder-circuit.tikz}
}
      \caption{The circuit with the expected postcondition where $c_{i+1}=a_ib_i \oplus a_ic_i \oplus b_ic_i$}
      \label{fig:rca}
    \end{subfigure}
  \end{minipage}
  \caption{A~circuit implementing a~ripple-carry adder.}
  \label{fig:Adder}
\end{figure}
}

\vspace{-0.0mm}
\subsection{Arithmetic Circuits}\label{sec:arithmetic}
\vspace{-0.0mm}

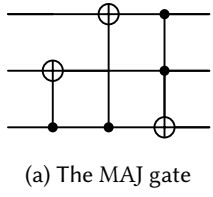
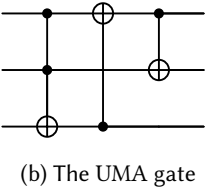
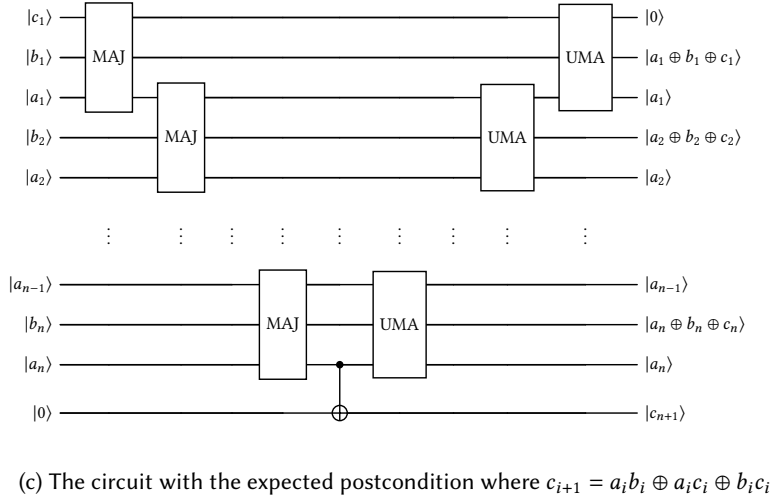
\begin{figure}[t]
  \begin{minipage}[b]{0.23\textwidth}
    \begin{subfigure}[t]{\textwidth}
      \centering
      \scalebox{0.9}{
        \begin{quantikz}
          & \qw          & \targ{}  & \ctrl{2} & \qw      \\
          & \targ{}      & \qw      & \ctrl{1}  & \qw      \\
          & \ctrl{-1}    &\ctrl{-2} & \targ{}   &\qw    
        \end{quantikz}
      }
      \caption{The $\maj$ gate}
      \label{fig:maj}
    \end{subfigure}
    \vspace{1.3cm} 

    \begin{subfigure}[b]{\textwidth}
      \scalebox{0.9}{
      \begin{quantikz}
        & \ctrl{2} & \targ{}  & \ctrl{1} & \qw     \\
        & \ctrl{1} & \qw      & \targ{}  & \qw    \\
          & \targ{}  &\ctrl{-2} & \qw    & \qw  
      \end{quantikz}      
      }
      \caption{The $\uma$ gate}
      \label{fig:uma}
    \end{subfigure}
  \end{minipage}
  \hfill
  \begin{minipage}[b]{0.76\textwidth}
    \begin{subfigure}[b]{\textwidth}
\scalebox{0.67}{
  \input{Figures/adder-circuit.tikz}
}
      \caption{The circuit with the expected postcondition where $c_{i+1}=a_ib_i \oplus a_ic_i \oplus b_ic_i$}
      \label{fig:rca}
    \end{subfigure}
  \end{minipage}
  \caption{A~circuit implementing a~ripple-carry adder.}
  \label{fig:Adder}
\end{figure}

Another class of examples to which our approach applies is that of arithmetic
circuits, such as adders and comparators.
These circuits frequently serve as basic components of quantum algorithms, e.g., within
the state preparation stage.
In this section, we focus on the verification of the functional correctness of 
a~carry-ripple adder~\cite{cuccaro2004new} given in \cref{fig:rca}.
The adder expects two binary numbers, $a_1\ldots a_n$ and $b_1 \ldots b_n$, as
well as an initial \emph{carry} bit~$c_1$ and an ancilla initialized to~$\ket
0$, and stores the sum of the numbers and the carry bit to the qubits $b_1\ldots
b_n$, with the output carry stored to the ancilla (the contents of the $a_1
\ldots a_n$ qubits is unchanged and the output value of~$c_1$ is~$0$).
The circuit is composed of three parts:
\begin{inparaenum}[(i)]
  \item  a~``\emph{downward staircase}'' sequence of Majority ($\maj$,
    \cref{fig:maj}) gates, denoted as~$\maj^\searrow$,
  \item  a~$\cnot$ between~$a_n$ and the ancilla, and
  \item  an~``\emph{upward staircase}'' sequence of UnMajority and Add ($\uma$,
    \cref{fig:uma}) gates, denoted as~$\uma^\nearrow$.
\end{inparaenum}
The precondition can be modelled by an SWTA accepting trees representing all
bases with the ancilla~$0$ and the postcondition can be expressed by an
SWTA accepting the expected results \mbox{(the construction is straightforward
but quite tedious).}

To verify the functional correctness of the adder using our framework, we need
to provide the transducers~$\trn_\maj^\searrow$ and $\trn_\uma^\nearrow$ for the
staircase sequences~$\maj^\searrow$ and~$\uma^\nearrow$ respectively.
We start with constructing the transducer~$\trn_\maj$ for the~$\maj$ gate
(\cref{fig:maj}), which can be easily done by composing the transducers for the
$\cnot$ and $\ccnot$ gates using the transducer composition algorithm from
\cref{sec:trn_compose}.
Then we need to construct $\trn_\maj^\searrow$ expressing a~parameterized
sequence of~$\maj$ gates in the given staircase pattern.
For this, we use the algorithm from \cref{sec:composition_par}, which takes
a~(fixed-input size) transducer and transforms it into a~transducer that
represents the given regular pattern for an arbitrary input size.
The transducer~$\trn_\uma^\nearrow$ is prepared similarly and
the verification of the adder then proceeds in the same way as in~\cref{sec:BV}.

%

\newcommand{\eccfig}[1]{
\begin{wrapfigure}[13]{r}{0.4\textwidth}
  \vspace{-5mm}
  \begin{minipage}{0.4\textwidth}
    \scalebox{0.78}{
    \input{Figures/qecc-circuit.tikz}
  }
  \end{minipage}
  \vspace{-4mm}
  \caption{Syndrome extraction circuit}
  \label{fig:repetition_code}
\end{wrapfigure}
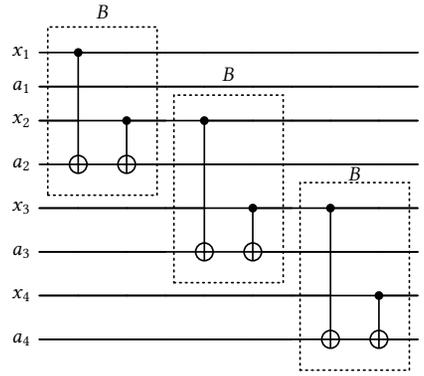
}

\vspace{-0.0mm}
\subsection{Quantum Error Correction Code}\label{sec:qecc}
\vspace{-0.0mm}

\eccfig  

As the next case study, we focus on verification of \emph{quantum error
correction codes} (QECC) and consider the \emph{syndrome extraction} circuit of
a~\emph{repetition code} used to correct bit-flip errors~\cite{Peres85}.
The code works such that a~qubit~$x_1$, whose value $\alpha \ket 0 + \beta \ket
1$ we aim to protect, is
entangled using~$\cnot$ gates with qubits~$x_2, \ldots, x_n$ (in a~similar way
as done by the GHZ circuit~\cite{GreenbergerHZ89}), obtaining the state
$\varphi = \alpha \ket{0^n} + \beta \ket{1^n}$.
Then~$\varphi$ enters a~noisy channel, which may perform some bit flips, so that
we obtain a~state~$\varphi' = \alpha \ket{w} + \beta \ket{\overline{w}}$ where
$w, \overline w$ are binary strings of the length~$n$ and $\overline{w}$ is the
binary complement of~$w$.
The next step is the syndrome extraction, implemented by the circuit in
\cref{fig:repetition_code} (for~$n=4$).
The circuit uses~$n-1$ ancillas (we use~$n$ in the figure to keep the structure
regular; $a_1$~can be removed) to detect bit-flips such that 
\begin{inparaenum}[(i)]
  \item  if a~single bit-flip error occurs at $x_1$, the ancilla $a_2$ flips from
    $\ket{0}$ to $\ket{1}$,
  \item  if the error occurs at the qubit $x_i$ for $2 \le i \le n-1$, both $a_i$ and
$a_{i+1}$ flip, and
  \item  if the error is at $x_n$, then only $a_n$ is affected.
\end{inparaenum}
We model the verification of the syndrome extraction circuit (for all sizes) as
follows:
For the precondition, we consider an SWTA that accepts all trees representing
quantum states $\invsqrttwo \ket{w} + \invsqrttwo \ket{\overline{w}}$, where $w$ contains at most one 1 at an odd position (corresponding to at most one error),
and for the postcondition, we construct an SWTA that represents the requirement
on the values of the ancillas as described above.
The transducer for the circuit is constructed using the algorithm from
\cref{sec:composition_par} and the verification is done similarly as in
\cref{sec:BV}.

%

\newcommand{\figGrover}[0]{
\begin{figure}[t]
  \begin{subfigure}[b]{0.37\textwidth}
    \scalebox{0.51}{
      \input{Figures/grover-circuit1.tikz}
    }
    \caption{Basic circuit~$C_1$ for Grover's algorithm}
    \label{fig:Grover_iter_b}
  \end{subfigure}
  \begin{subfigure}[b]{.62\textwidth}
    \scalebox{0.46}{
      \input{Figures/grover-circuit2.tikz}
    }
    \caption{Circuit~$C_2$ for Grover's algorithm without multi-control gates}
    \label{fig:Grover_iter_c}
  \end{subfigure}  
  \caption{Two implementations of a~single iteration of
  Grover's algorithm with the solution $\ket{0^n}$.}
  \label{fig:Grover_iter_subfig}
\end{figure}
}

\vspace{-0.0mm}
\subsection{Amplitude Amplification Circuits}\label{sec:grover}
\vspace{-0.0mm}

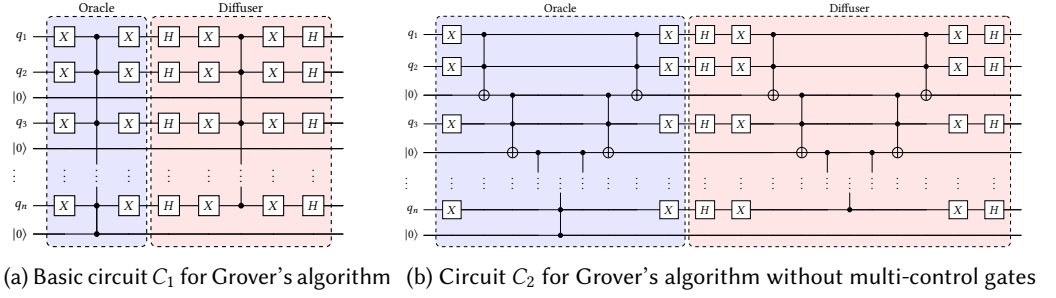
\begin{figure}[t]
  \begin{subfigure}[b]{0.37\textwidth}
    \scalebox{0.51}{
      \input{Figures/grover-circuit1.tikz}
    }
    \caption{Basic circuit~$C_1$ for Grover's algorithm}
    \label{fig:Grover_iter_b}
  \end{subfigure}
  \begin{subfigure}[b]{.62\textwidth}
    \scalebox{0.46}{
      \input{Figures/grover-circuit2.tikz}
    }
    \caption{Circuit~$C_2$ for Grover's algorithm without multi-control gates}
    \label{fig:Grover_iter_c}
  \end{subfigure}  
  \caption{Two implementations of a~single iteration of
  Grover's algorithm with the solution $\ket{0^n}$.}
  \label{fig:Grover_iter_subfig}
\end{figure}

Amplitude amplification algorithms form a fundamental class of quantum
search algorithms that may provide speedups (though mostly polynomial) over
classical algorithms~\cite{Grover96,BrassardHMT00,BrassardHT98}.
Among them, the best-known example is \emph{Grover's algorithm}~\cite{Grover96},
which operates iteratively: with each iteration, the probability of successfully
finding the correct answer increases (until a~certain bound).
The basic (parameterized) circuit~$C_1$ for one iteration of Grover's algorithm is
shown in~\cref{fig:Grover_iter_b}.

In practice, the circuit we will use may deviate from the basic form.
For instance, we might replace all multi-controlled gates with those supported
by a~concrete hardware, e.g., with a~cascade of~$\ccnot$ gates, in the circuit~$C_2$ in
\cref{fig:Grover_iter_c}.
In order to verify correctness of such a~modification for any circuits in the
families, we will leverage the SWTA-based framework in the following way:
We will construct the SWTA~$\autbases$ that accepts trees encoding all basis
states where ancillas are set to zero (the construction is simple).
Then, we perform two computations:
\begin{inparaenum}[(i)]
  \item  $\but_1 = C_1(\autbases)$ and
  \item  $\but_2 = C_2(\autbases)$.
\end{inparaenum}
One can easily see that all of the required parameterized (and standard) gates
can be implemented using WTTs.
In order to establish correctness
\begin{changebar}
of the second circuit,
\end{changebar}
it is enough to check whether~$\but_1$
and~$\but_2$ are functionally equivalent, $\semof{\but_1} = \semof{\but_2}$.

\newcommand{\figUnitHeisOne}[0]{
\begin{figure}[t]
    \centering
    \begin{subfigure}[t]{.17\textwidth}
        \scalebox{.55}{
            \begin{quantikz}[row sep=0.4cm]
                \qw &\ctrl{1} & \qw & \ctrl{1} &\qw  \\
                \qw &\targ{}    &\gate{R_z(2\delta)}   &\targ{} &\qw \\
            \end{quantikz}
        }
        \caption{$\rzz$ box}
        \label{fig:eZZ-gate}
    \end{subfigure}
    \hfill
    \begin{subfigure}[t]{.20\textwidth}
        \scalebox{.55}{
            \begin{quantikz}[row sep=0.1cm]
                \qw  &\gate{H} & \gate[2]{\rzz} &\gate{H} &\qw  \\
                \qw  &\gate{H}    &   &\gate{H} &\qw \\
            \end{quantikz}
        }
        \caption{$\rxx$ box}
        \label{fig:eXX-gate}
    \end{subfigure}
    \hfill
    \begin{subfigure}[t]{.28\textwidth}
        \scalebox{.55}{
            \begin{quantikz}[row sep=0.1cm]
                \qw    &\gate{S} &\gate{H} & \gate[2]{\rzz} &\gate{H}&\gate{S} &\qw  \\
                \qw &\gate{S}&\gate{H}    &   &\gate{H}&\gate{S} &\qw \\
            \end{quantikz}
        }
        \caption{$\ryy$ box}
        \label{fig:eYY-gate}
    \end{subfigure}
    \hfill
\begin{subfigure}[t]{.23\textwidth}
        \scalebox{.55}{
    \begin{quantikz}[row sep=0.4cm]
        \qw &\ctrl{1} & \qw & \ctrl{1} &\qw &\qw  \\
        \qw &\targ{}    &\gate{R_z(2\delta)}   &\targ{} &\gate{X} &\qw \\
    \end{quantikz}
}
\caption{$U_{zz}(2\delta)$ box}
\label{fig:uZZ-gate}
\end{subfigure}

\bigskip

    \begin{subfigure}{\textwidth}
        \centering
        \scalebox{.6}{
          \input{Figures/unitary-heis-circuit1.tikz}
        }
        \caption{The $\Uheis$ circuit composed of $\rxx$, $\ryy$, and $\rzz$ boxes.}
        \label{fig:heis-circuit}
    \end{subfigure}
    \vspace{-6mm}
    \caption{The circuit implementing the first-order approximation of
    $e^{-i\Hheis \delta}$, where $\delta = \smallfrac t r$ for some large
    integer~$r$ that sets the precision of simulation,
    $\rz = \textstyle \left[\begin{smallmatrix}
    e^{-i\delta} & 0 \\ 0 & e^{i\delta}
\end{smallmatrix}\right]$.
  We pick $t= \pi$ and $r=4$ in the experiment.}
\end{figure}
}

\newcommand{\figUnitHeisTwo}[0]{
\begin{figure}[t]
  \resizebox{\textwidth}{!}{
    \input{Figures/unitary-heis-circuit2.tikz}
  }
  \vspace{-2mm}
\caption{An optimized version of~$\Uheis$ obtained using the
  equalities $\had \cdot \had = \idgate, \had \cdot \sgate \cdot \sgate \cdot
  \had = \xgate$ and $\had \cdot \sgate \cdot \had = \sqrt{\xgate}$.}
  \label{fig:heis-circuit-opt}
\end{figure}
}

\vspace{-0.0mm}
\subsection{Circuit for Hamiltonian Simulation}\label{sec:hamiltonian}
\vspace{-0.0mm}

\begin{figure}[t]
    \centering
    \begin{subfigure}[t]{.17\textwidth}
        \scalebox{.55}{
            \begin{quantikz}[row sep=0.4cm]
                \qw &\ctrl{1} & \qw & \ctrl{1} &\qw  \\
                \qw &\targ{}    &\gate{R_z(2\delta)}   &\targ{} &\qw \\
            \end{quantikz}
        }
        \caption{$\rzz$ box}
        \label{fig:eZZ-gate}
    \end{subfigure}
    \hfill
    \begin{subfigure}[t]{.20\textwidth}
        \scalebox{.55}{
            \begin{quantikz}[row sep=0.1cm]
                \qw  &\gate{H} & \gate[2]{\rzz} &\gate{H} &\qw  \\
                \qw  &\gate{H}    &   &\gate{H} &\qw \\
            \end{quantikz}
        }
        \caption{$\rxx$ box}
        \label{fig:eXX-gate}
    \end{subfigure}
    \hfill
    \begin{subfigure}[t]{.28\textwidth}
        \scalebox{.55}{
            \begin{quantikz}[row sep=0.1cm]
                \qw    &\gate{S} &\gate{H} & \gate[2]{\rzz} &\gate{H}&\gate{S} &\qw  \\
                \qw &\gate{S}&\gate{H}    &   &\gate{H}&\gate{S} &\qw \\
            \end{quantikz}
        }
        \caption{$\ryy$ box}
        \label{fig:eYY-gate}
    \end{subfigure}
    \hfill
\begin{subfigure}[t]{.23\textwidth}
        \scalebox{.55}{
    \begin{quantikz}[row sep=0.4cm]
        \qw &\ctrl{1} & \qw & \ctrl{1} &\qw &\qw  \\
        \qw &\targ{}    &\gate{R_z(2\delta)}   &\targ{} &\gate{X} &\qw \\
    \end{quantikz}
}
\caption{$U_{zz}(2\delta)$ box}
\label{fig:uZZ-gate}
\end{subfigure}

\bigskip

    \begin{subfigure}{\textwidth}
        \centering
        \scalebox{.6}{
          \input{Figures/unitary-heis-circuit1.tikz}
        }
        \caption{The $\Uheis$ circuit composed of $\rxx$, $\ryy$, and $\rzz$ boxes.}
        \label{fig:heis-circuit}
    \end{subfigure}
    \vspace{-6mm}
    \caption{The circuit implementing the first-order approximation of
    $e^{-i\Hheis \delta}$, where $\delta = \smallfrac t r$ for some large
    integer~$r$ that sets the precision of simulation,
    $\rz = \textstyle \left[\begin{smallmatrix}
    e^{-i\delta} & 0 \\ 0 & e^{i\delta}
\end{smallmatrix}\right]$.
  We pick $t= \pi$ and $r=4$ in the experiment.}
\end{figure}
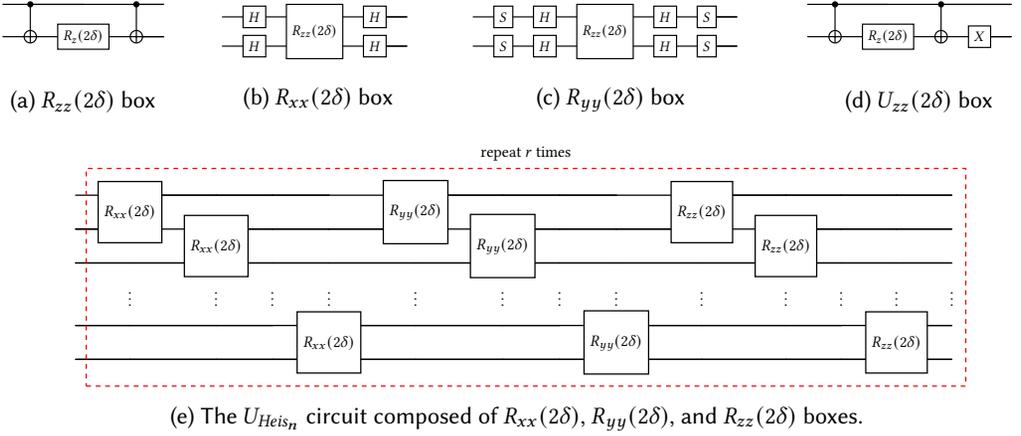

\emph{Hamiltonian simulation} is a fundamental computational task in quantum
computing, underpinning a broad class of quantum algorithms with applications in
physics, chemistry, and material science. At the core is the approximation of
the unitary evolution $e^{-i\hamilton t}$, where $\hamilton$ is a~Hermitian
matrix (the Hamiltonian) describing the dynamics of a quantum system, and~$t$ is
a~real-valued time parameter.
Efficient circuit implementations of this operator are critical for simulating
quantum systems. 

As an example, we focus on the simulation of a \emph{1D Heisenberg spin
chain}~\cite{1DHeis}, a well-studied model whose Hamiltonian is $\Hheis =
\sum_{j=1}^{n-1} \left( X_j X_{j+1} + Y_j Y_{j+1} + Z_j Z_{j+1} \right)$.
The evolution operator $e^{-i{\Hheis}t}$ can be approximated using the
\begin{changebar}
parameterized
\end{changebar}
circuit~$\Uheis$ obtained via a \emph{first-order Trotter-Suzuki
decomposition}~\cite[Sec.~4.7.2]{NielsenC16} shown in \cref{fig:heis-circuit}.

For efficiency, quantum circuits are often further optimized before being
executed on hardware.
In~\cref{fig:heis-circuit-opt}, we show an optimized version of the original
circuit.
Our task is to verify that the optimized versions of the circuits are equivalent
to the original ones.
We test the equivalence in the same was as in~\cref{sec:grover}.
For constructing the necessary transducers for the staircase sequences of boxes,
we use the algorithm from \cref{sec:composition_par}.

\newcommand{\tabResults}[0]{
\begin{wraptable}[9]{r}{0.30\textwidth}
  \vspace*{-4mm}
\caption{Results of verification of our case studies.}
\label{tab:results}
\vspace{-3mm}
\scalebox{0.8}{
\begin{tabular}{lr}
  \toprule
  \multicolumn{1}{c}{\bf circuit} & \multicolumn{1}{c}{\bf time}\\
  \midrule
  BV & 0.014\,s \\
  Grover & 0.088\,s \\
  Adder & 11.007\,s \\
  QECC & 0.314\,s \\
  Hamiltonian simulation &  0.663\,s \\
  \bottomrule
\end{tabular}
}
\end{wraptable}
}

\vspace{-0.0mm}
\subsection{Experimental Evaluation}\label{sec:experiments}
\vspace{-0.0mm}

\tabResults    
We implemented the techniques described in the paper in a~prototype tool
named \tool~\cite{tool}\footnote{To ensure accuracy, we restrict our attention
to the subset of complex numbers that admit an algebraic representation.
See~\trOrAppendix{sec:complex} for details.}
and used it to evaluate the case studies described in
\cref{sec:use-cases}.
The experiments were evaluated on a~computer with the
Intel i7-1365U CPU and 32\,GiB of RAM running Fedora Linux~42.
The results are given in \cref{tab:results}.
We are not aware of any other tool supporting fully automated
parameterized verification of quantum circuits that we could compare with.
From the table, you can see that most of the circuits were verified
quickly, with the longest time taken by the adder circuit
(\cref{sec:arithmetic}).
The adder took the longest due to the complexity of the circuit (compared to
other circuits) and the corresponding
postcondition, which made the constructed SWTAs and WTTs large.
Our implementation is an early prototype and there are many
opportunities to make it faster.

\begin{figure}[t]
  \resizebox{\textwidth}{!}{
    \input{Figures/unitary-heis-circuit2.tikz}
  }
  \vspace{-2mm}
\caption{An optimized version of~$\Uheis$ obtained using the
  equalities $\had \cdot \had = \idgate, \had \cdot \sgate \cdot \sgate \cdot
  \had = \xgate$ and $\had \cdot \sgate \cdot \had = \sqrt{\xgate}$.}
  \label{fig:heis-circuit-opt}
\end{figure}

\vspace{-0.0mm}
\section{Properties and Operations over SWTAs and WTTs}
\vspace{-0.0mm}

In this section, we lay out the operations and decision problems for SWTAs and
WTTs.

\vspace{-0.0mm}
\subsection{Domain DFA}\label{sec:label}
\vspace{-0.0mm}

Let~$\aut = \tuple{Q, \delta, \colors, \rootstate, E}$ be an SWTA over~$\abc$.
We define the \emph{domain DFA of~$\aut$}, denoted $\coloraut \aut$, to be
a~\emph{deterministic finite automaton over finite words}
(DFA)~\cite{EsparzaB23} $\coloraut \aut = \tuple{P, \abc\times \colors, \Delta,
\{\rootstate\}, P_f}$ with the set of states $P = 2^Q$, the initial state $\{\rootstate\}$, final states $P_f=\{S \in P\mid S\subseteq E \land S\neq \emptyset\}$, and the transition function $\Delta(S,\tuple{ a,\tcacc{7}{c}}) = \bigcup_{s\in S}\{p \in Q\mid
\transof{s}{a}{\tcacc 7 c}{\linof \ell}{\linof r}\in \delta \land p\in \suppof{\linof \ell}\cup \suppof{\linof r} \}$.
The following theorem establishes the connection between $\coloraut \aut$ and
the domain of~$\semof \aut$.

\begin{theorem}\label{thm:domain_dfa}
Let~$\aut$ be an SWTA.
Then $\langof{\coloraut \aut} = \domof{\semof \aut}$.
\end{theorem}

\vspace{-0.0mm}
\subsection{SWTA Language Emptiness}\label{sec:emptiness}
\vspace{-0.0mm}

The language emptiness problem for SWTAs asks whether a~given SWTA accepts at
least one tree.

\begin{theorem}\label{thm:swta-emptiness}
Given an SWTA~$\aut$, deciding whether $\langof \aut = \emptyset$ is
\pspace-complete.
\end{theorem}

\begin{proof}
  (\pspace-hardness) The hardness follows by reduction from universality of
  a~\emph{nondeterministic finite automaton} (NFA), which is a~known \pspace-complete
  problem~\cite{EsparzaB23}.
  In particular, consider an~NFA~$\nut = \tuple{Q, \Sigma, \Delta, q_I, F}$,
  where $Q$ is a~set of states, $\Sigma$ is the input (finite) alphabet, $q_I
  \in Q$ is the initial state (w.l.o.g.\ we consider exactly one initial
  state), $F \subseteq Q$ is the set of final states, and $\Delta\colon Q \times
  \Sigma \to 2^Q$ is the NFA transition function.
  We assume that~$\nut$ is complete, i.e., $\Delta(P,a) \neq \emptyset$ for
  any~$P \in Q, a \in \Sigma$.
  We construct the SWTA~$\aut_\nut = \tuple{Q, \delta,\Sigma, q_I, Q
  \setminus F}$ over the alphabet $\{a\}$, i.e., with the same states as~$\nut$ with a~single input
  symbol~$a$ and the set of colours corresponding to the alphabet~$\Sigma$, and
  with a~state being a~leaf state iff it was not a~final state of~$\nut$, where $\delta$ is defined as
  $\delta = \{\transof q a {\tcacc 4 c}{\linof x}{\linof x} \mid q \in Q, c \in
  \Sigma, P = \Delta(q,c), \linof x = \sum_{p \in P} p\}$.
  Then it holds that if there is a~word~$w = b_1 \ldots b_n$ not in $\langof
  \nut$---which means that the set of reachable states in~$\nut$ over~$w$ does
  not contain any state from~$F$---then there will be a~tree $t_w$ such that
  $\semof{\aut_\nut}(w') = t_w$ (since all reached states will be leaf states)
  where $w' = \tuple{a, b_1} \ldots \tuple{a,b_n}$.
  Therefore, $\nut$~is universal iff $\langof{\aut_\nut}$ is empty.

  (\pspace-membership) We reduce the
  \begin{changebar}
  SWTA language emptiness
  \end{changebar}
  problem to checking emptiness of the
  language of the domain DFA $\coloraut \aut$ (obviously, $\langof \aut =
  \emptyset$ iff $\langof{\coloraut \aut} = \emptyset$).
  While $\domof \aut$ can be exponentially larger than~$\aut$ (its construction
  is based on the subset construction), we can explore it on the fly during
  the emptiness check in the standard way,
  \begin{changebar}
  similarly as in NFA universality checking,
  \end{changebar}
  which can be done in nondeterministic polynomial space.
  \pspace membership then follows from Savitch's theorem.
\end{proof}



\vspace{-0.0mm}
\subsection{Functional Equivalence and Inclusion Checking}\label{sec:colored_equivalence_checking}
\vspace{-0.0mm}

\begin{changebar}
Our framework is based on testing functional equivalence and functional
inclusion of two
SWTAs~$\aut$ and~$\but$, i.e., on checking whether $\semof \aut = \semof \but$
and $\semof \aut \subseteq \semof \but$ respectively.
\end{changebar}
Unlike language equivalence and inclusion (i.e., $\langof \aut = \langof \but$
and $\langof \aut \subseteq \langof \but$)---which are undecidable for SWTAs,
cf.\ \cref{sec:language_equivalence}---their functional counterparts can be
decided efficiently and are usable for verification of parameterized quantum
circuits.

Checking functional equivalence and inclusion are done via a~reduction to
solving the \emph{zero-invariant problem of a~linear transition system}, which
can be done using the so-called Karr's
algorithm~\cite{Karr76,muller2004note,chen2017register}.
Let us begin with defining the necessary concepts.
A \emph{linear transition system} (LTS) with $k$~complex-valued variables arranged in
a~column vector $\overline{x}=(x_1, \ldots, x_k)^T$ is a~tuple $P = \tuple{S, s_0, \overline{v_0},
\ltstr}$, where $S$ is a finite set of states,
$s_0\in S$ is the initial state, $\overline{v_0} \in \complex^k$ is the (column)
vector of initial variable values, and $\ltstr$ is a finite set of transitions
of the form $s_1 \ltr A s_2$, where $s_1, s_2 \in S$ and $A \in
\complex^{k\times k}$ is a~\emph{linear transformation} represented by a $k
\times k$ complex matrix.
A~transition $s_1 \ltr A s_2\in \ltstr$ means that~$P$ can move from
state~$s_1$ to state~$s_2$ while reassigning the contents of variables as
$\overline{x}:=A\overline{x}$.
The \emph{zero-invariant problem} for~$P$ and~$S_t \subseteq S$ asks whether for
all sequences of transitions $s_0 \ltr{T_1} s_1 \ltr{T_2} \ldots
\ltr{T_m} s_m \in \ltstr^*$ such that~$s_m \in S_t$, it holds that $T_{m} \ldots T_2 T_1
\overline{v_0} = \mathbf 0$ where $\mathbf 0$ denotes the $k$-dimensional
zero vector.
Informally, we are asking whether there is a~sequence of transitions from the
initial configuration~$\tuple{s_0, \overline{v_0}}$ that reaches a~state
in~$S_t$ with at least one variable having a~non-zero value.

\begin{theorem}\label{thm:zero-invariant}
The zero-invariant problem for an LTS with~$n$ states and~$k$ variables can be
solved in time $\bigO(nk^3)$ considering a~unit cost of arithmetic operations.
\end{theorem}

\begin{proof}
The problem can be solved by the so-called
Karr's algorithm~\cite{Karr76,muller2004note}
(we note that Karr's algorithm solves a~more general problem of computing
\emph{affine relationships} in \emph{affine programs}, i.e., programs where
transformations are of the form $\overline{x}:=A\overline{x}+\overline{b}$;
here we only use $\overline{b}=\mathbf{0}$).
Intuitively, Karr's algorithm works by starting in the initial state with the
initial vector and propagating the vector along the transitions.
At each state, we collect the (at most~$k$) linearly independent vectors.
After the system saturates, we check that all states in~$S_t$ are only reachable
with the vector~$\mathbf 0$.
Karr's algorithm works in the time $\bigO(nk^3)$ if arithmetic operations have
a~unit cost.
\end{proof}

Let us now show how to build an~LTS that corresponds to the given SWTA
functional equivalence problem for SWTAs $\aut = \tuple{Q_a, \delta_a, \colors,
\rootstate_a, E_a}$ and~$\but = \tuple{Q_b, \delta_b, \colors, \rootstate_b,
E_b}$ over~$\abc$ (w.l.o.g.\ we assume they use the same set of
colors~$\colors$ such that $\tacc 1 \in \colors$ and that $Q_a \cap Q_b =
\emptyset$).
Let us first construct an SWTA~$\aut_{\minus}$ such that for all $w \in (\abc \times
\colors)^*$, if~$\semof \aut (w) = \semof \but (w)$, then $\aut_{\minus}$ will
\begin{changebar}
generate a~perfect tree of the height~$|w|$ with all leaves being~zero, otherwise
it will generate a~tree of the same height but with at least one non-zero leaf.
\end{changebar}
Concretely,
we define $\aut_{\minus} = \tuple{Q, \delta, \colors, \rootstate, E}$
to be an SWTA over~$\abc$ where
\begin{itemize}
  \item  $Q = Q_a \cup Q_b \cup \{\rootstate\}$ where
         $\rootstate \notin Q_a \cup Q_b$,
  \item  $E = E_a \cup E_b$, and
  \item  $\delta = \delta_a \cup \delta_b \cup \{\transcof \rootstate \alpha 1
    {\rootstate_a - \rootstate_b}{\rootstate_a - \rootstate_b}\}$ where
    $\alpha$ is any symbol from~$\abc$.
\end{itemize}

\begin{theorem}\label{thm:aut_minus}
For all $w \in \domof{\semof \aut} \cap \domof{\semof \but}$, it holds that
$\semof{\aut_{\minus}}(\tuple{\alpha, \tacc 1} w)$ is the tree defined as
$\treecons \alpha {t_{\minus}}{t_{\minus}}$ where $t_{\minus} = \semof \aut
(w) - \semof \but (w)$.
Moreover, $\domof{\aut_{\minus}} =
\{\tuple{\alpha, \tacc 1}\} \concat (\domof{\semof \aut} \cap \domof{\semof \but})$.
\end{theorem}


\begin{corollary}
Assume that $\domof{\semof \aut} = \domof{\semof \but}$.
Then it holds that $\semof \aut = \semof \but$ iff 
all trees in~$\langof{\aut_{\minus}}$ have only zero-valued leaves.
\end{corollary}

To check whether $\semof \aut = \semof \but$, it is now enough to proceed as
follows:
\begin{enumerate}
  \item  Check that $\domof{\semof \aut} = \domof{\semof \but}$ by testing
    $\langof{\coloraut \aut} = \langof{\coloraut \but}$ (alternatively, for
    testing functional inclusion, we test $\langof{\coloraut \aut} \subseteq
    \langof{\coloraut \but}$).
    This can be done in \pspace since we can perform an on-the-fly
    construction of the product of the domain DFAs~$\coloraut \aut$
    and~$\coloraut \but$, which are both exponential-sized, while looking for
    a~state that is accepting in exactly one of the DFAs).
    If the equivalence does not hold, return false.

  \item  Construct~$\aut_{\minus}$ and test whether all trees it accepts have
    only zero-valued leaves.
\end{enumerate}
To test that all trees accepted by~$\aut_{\minus}$ have only zeros in their
leaves, we reduce the problem to the zero-invariant problem of the LTS~$P_{\minus}$ defined below.
First, we define an auxiliary mapping~$\matrixof{\aut_{\minus}}\colon \abc
\times \colors \times \{\mathit{left}, \mathit{right}\} \to \complex^{k\times
k}$ with~$k = |Q|$ being the number of states in~$\aut_{\minus}$ (we assume $Q =
\{q_1, \ldots, q_k\}$).
The mapping represents the transitions of~$\aut_{\minus}$ as complex matrices in
a~similar way as in weighted automata~\cite{DrosteKV09}.
Let us begin by defining, for a~linear form $\linof x \in \linformsof Q$, the
vector~$\vectorof {\linof x} = (u_1, \ldots, u_k)$ such that
$u_i = \linindexof x {q_i}$ if $q_i \in \suppof{\linof x}$ and $u_i = 0$
otherwise.
Then, for all~$q \in Q$, we define the function $\rowof{q}\colon \abc
\times \colors \times \{\mathit{left}, \mathit{right}\} \to \complex^k$
such that
\begin{equation}
\rowof q(a, \tcacc 7 c, D) = \begin{cases}
  \vectorof{\linof \ell}  & \text{if } \transof q a {\tcacc 7 c} {\linof \ell}{\linof
  r} \in \delta \text{ and } D = \mathit{left},\\
  \vectorof{\linof r}  & \text{if } \transof q a {\tcacc 7 c} {\linof \ell}{\linof
  r} \in \delta \text{ and } D = \mathit{right}\text{, and}\\
  \mathbf{0} & \text{otherwise}.
\end{cases}
\end{equation}
We then define~$\matrixof{\aut_{\minus}}$ by composing the vectors for~$\rowof
q$ for all states~$q \in Q$ as follows:
\begin{equation}
\matrixof{\aut_{\minus}}(a, \tcacc 7 c, D) =
\left[
\begin{matrix}
  \rowof{q_1}(a, \tcacc 7 c, D) \\
  \vdots\\
  \rowof{q_k}(a, \tcacc 7 c, D)
\end{matrix}
\right].
\end{equation}
Further, consider the domain DFA $\coloraut{\aut_{\minus}} = \tuple{G, \abc
\times \colors, \Delta, g_0, G_f}$.
We construct
the LTS~$P_{\minus} = \tuple{S, s_0, \overline{v_0}, \ltstr}$ in the following way:
\begin{itemize}
  \item  $S = 2^Q \times G$, 
  \item  $s_0 = \tuple{\{\rootstate\}, g_0}$,
  \item  $\overline{v_0} = (1, 0, \ldots, 0)$, while assuming that $q_1 =
    \rootstate$, and
  \item  $\tuple{U,g} \ltr A \tuple{U', g'} \in \ltstr$ iff there exist $a \in
    \abc$ and $\tcacc 7 c \in \colors$ such that
      $g' = \Delta(g, \tuple{a, \tcacc 7 c})$ and
      one of the following holds:
      \begin{itemize}
        \item
          $U' = \bigcup\{\suppof{\linof \ell} \mid u \in U, \transof u a
          {\tcacc 7 c} {\linof \ell}{\linof r} \in \delta\}\}$ and
          $A = \matrixof{\aut_{\minus}}(a, \tcacc 7 c, \mathit{left})$ or
        \item
          $U' = \bigcup\{\suppof{\linof r} \mid u \in U, \transof u a
          {\tcacc 7 c} {\linof \ell}{\linof r} \in \delta\}\}$ and
          $A = \matrixof{\aut_{\minus}}(a, \tcacc 7 c, \mathit{right})$.
      \end{itemize}
\end{itemize}
Intuitively, in each $P_{\minus}$'s state $\tuple{U, g}$, the~$U$-component
denotes the set of $\aut_{\minus}$'s states ``\emph{active}'' at a~particular
branch of the input tree while~$g$ keeps track of $\aut_{\minus}$'s
active states 
\begin{changebar}
on \emph{all branches}
\end{changebar}
(since in order to accept a tree, all branches need to reach
a~leaf state in all of their active states).
\begin{changebar}
We also construct $S_t = \{\tuple{U, g} \in S \mid g \in G_f\}$, which
is a~set of states of the LTS representing computations
of~$\aut_{\minus}$ where all branches reach a~leaf node.
\end{changebar}

\begin{theorem}\label{thm:func_eq}
The functional equivalence $\semof \aut = \semof \but$ holds iff $\domof{\semof
\aut} = \domof{\semof \but}$ and the zero-invariant problem for~$P_{\minus}$
and~$S_t$ holds.
Similarly, the functional inclusion $\semof \aut \subseteq \semof \but$ holds iff $\domof{\semof
\aut} \subseteq \domof{\semof \but}$ and the zero-invariant problem for~$P_{\minus}$
and~$S_t$ holds.
Moreover, the two problems can be decided in \expspace.
\end{theorem}

\begin{proof}[Proof sketch.]
The correctness can be established using the reasoning outlined above.
Focusing on the complexity part,
we note that the size of~$P_{\minus}$ is exponential to the size of the input.
The length of a~sequence of transitions that would lead to a non-zero vector in
some of the target states is then also exponential.
On the first sight, this might permit a~\pspace algorithm, which would guess the
sequence on the fly and locally generate successors of vectors until the
requested counterexample is reached.
The problem is, however, that the sizes of the bit representations of the
numbers in the vector might grow exponentially, so they do not fit into
polynomial space.
They are, however, bounded by \expspace.
\end{proof}

\begin{theorem}\label{thm:}
The functional equivalence and inclusion problems for SWTAs are \pspace-hard.
\end{theorem}

\begin{proof}[Proof sketch.]
We prove \pspace-hardness of the problems by a~reduction from SWTA
language emptiness, which is \pspace-complete (cf. \cref{thm:swta-emptiness}).
Concretely, we can test language emptiness of an SWTA~$\aut$ by constructing an
SWTA~$\but$ with an empty language (so, also $\semof \but = \emptyset$) and
checking whether $\semof \aut = \semof \but$ or $\semof \aut \subseteq \semof \but$.
\end{proof}

\vspace{-0.0mm}
\subsection{Language Intersection Emptiness}
\vspace{-0.0mm}

\begin{theorem}\label{thm:isect_empt_undec}
    Given two SWTAs $\aut$ and $\but$, it is undecidable whether $\langof \aut
    \cap \langof \but = \emptyset$.
\end{theorem}

\begin{proof}
\begin{changebar}
By reduction from \emph{Post's correspondence problem} (PCP)~\cite{Post46}. 
\end{changebar}
Assume an instance~$I$ of PCP over the alphabet~$\{1, \ldots, 9\}$ with the set
of~$k$ pairs $I = \{\tuple{\alpha_1, \beta_1}, \ldots, \tuple{\alpha_k, \beta_k}\}$.
For a~word~$w$, let~$\revof w$ denote the reverse of~$w$, $\numof w$ denote the
  numerical value of~$w$, with $\numof \epsilon = 0$, and~$|w|$ denote the
  length of~$w$.
Let us now construct the SWTA~$\aut = \tuple{Q, \delta, \colors, q, E}$
over the alphabet~$\{a\}$ as follows:
  \begin{itemize}
    \item  $Q = \{q, p, r, s, u\}$, $E = \{p,r, s, u\}$, $\colors = \{\tacc 1, \ldots, \tcacc 7 k\}$, and
    \item  $\delta = \delta_1 \cup \delta_2 \cup \delta_\alpha \cup
      \delta_\beta$ where
      \begin{itemize}
        \item  $\delta_1 = \{\transcof q a 1 {p-r}{0s}\}$ ,
        \item  $\delta_2 = \{\transof s a {\tcacc
      7 c} s s , \transof u a {\tcacc 7 c} {u}{0s} \mid \tcacc 7 c \in
          \colors\}$,
        \item  $\delta_\alpha = \{\transof p a {\tcacc 7 c} {10^{|\alpha_c|}\cdot p +
          \numof{\revof{\alpha_c}} \cdot u}{0s} \mid \tuple{\alpha_c, \beta_c} \in I
          \}$, and
        \item  $\delta_\beta = \{\transof r a {\tcacc 7 c} {10^{|\beta_c|}\cdot r +
          \numof{\revof{\beta_c}} \cdot u}{0s} \mid \tuple{\alpha_c, \beta_c} \in I \}$.
      \end{itemize}
  \end{itemize}
  Intuitively, we are trying to construct a tree where the leftmost branch leads
  to the leaf with value~0 (by construction, all other branches always go to
  zero-valued leaves) iff the PCP has a solution.
  The concatenation in the PCP is simulated by addition and digit-shifting
  implemented by multiplication by a~power of~10.
  Therefore, if the PCP has a solution, then (and only then) the language
  of~$\aut$ contains a~tree whose all leaves are valued 0.

  Moreover, let $\but$ be an SWTA accepting all perfect binary trees of the
  height $\geq 1$ with leaves labeled by~0 (such an SWTA needs just two
  transitions $\transcof v a 1 {0z}{0z}$ and $\transcof z a 1 z z$ with the root
  state~$v$ and the set of leaf states $\{z\}$).
  It holds that~$\langof \aut \cap \langof \but \neq \emptyset$ iff~$I$
  has a~solution.
\end{proof}

\vspace{-0.0mm}
\subsection{Language Equivalence and Inclusion}\label{sec:language_equivalence}
\vspace{-0.0mm}

In this section, we show that questions about language inclusion or equivalence
of a~pair of SWTAs are undecidable.

\begin{restatable}{theorem}{thmInclusionUndecidable}\label{thm:inclusionUndecidable}
    Given SWTAs $\aut$ and $\but$, the following questions are undecidable:
    \begin{enumerate}
      \item  $\langof \aut \subseteq \langof \but$ and
      \item  $\langof \aut = \langof \but$.
    \end{enumerate}
\end{restatable}

\begin{proof}[Proof sketch.]
  Our proof of the undecidability of $\langof \aut \subseteq \langof \but$ is
  done by a~reduction from the emptiness problem of a~Turing machine
  (TM)~$M$, inspired by the proof of undecidability of universality of timed
  automata~\cite[Theorem~5.2]{AlurD94}.
  In particular, we construct SWTAs~$\aut$ and~$\but$ over the alphabet~$\{a\}$
  such that $\langof M \neq \emptyset \;\Leftrightarrow\; \langof \aut
  \nsubseteq \langof \but$. 
  $\but$ will be constructed as an SWTA that accepts all perfect trees with
  exactly one 1-valued leaf node and 0-valued
  leaves elsewhere (the construction is simple).

  On the other hand, $\aut$~will be constructed in the following way.
  Let~$\rho$ be a~run (i.e., a~sequence of configurations) of~$M$, then we use
  $\encof \rho$ to denote an encoding of~$\rho$ into a binary string (done in
  the standard way).
  Our goal is to construct~$\aut$ such that every tree accepted by~$\aut$ has
  all leaves labelled by~0 except the leaf at the end of one significant branch,
  which will have the value~$1$ iff the branch is \emph{not} an encoding of an
  accepting run of~$M$ (and it will have any other value otherwise).
  In other words, if there exists an accepting run~$\rho_{\mathit{acc}}$ of~$M$,
  then $\langof \aut$ contains a~tree whose branch~$\encof{\rho_{\mathit{acc}}}$
  does not end with~1 and, therefore, $\langof \aut \nsubseteq \langof \but$.

  The construction of~$\aut$ is technical, but the main idea is that we will
  start a~parallel computation from two states (using a~linear form $p-r$)
  where~$p$ will try to detect a~simple error in the given input, such as problems
  with the encoding or a~wrong format of the configurations (which can be
  described using a~regular language, so it is easy to implement them in an SWTA).
  For detecting harder errors, such as an inconsistency in the encoding of two
  consecutive configurations, we use the run from~$p$ to guess the position of
  the error and then use color-based synchronization to communicate with the run
  from~$r$ the nature of the error; $r$~will then make sure that the error was
  guessed correctly.
  Marking the position can be done by, e.g., multiplying the children states
  with~$2$.

  The proof of undecidability of $\langof \aut = \langof \but$ can be done by
  reduction from the inclusion: $\langof \aut \subseteq \langof \but \iff
  \langof \aut \cup \langof \but = \langof \but$.
\end{proof}


\vspace{-0.0mm}
\subsection{Boolean Operations}\label{sec:boolean}
\vspace{-0.0mm}

We will show that SWTAs are closed under language union but are not closed under
language intersection and language complement.

\begin{theorem}\label{thm:union}
Let $\aut$ and $\but$ be SWTAs.
Then one can effectively construct an SWTA $\cut$ such that $\langof \cut =
  \langof \aut \cup \langof \but$.
\end{theorem}

\begin{proof}
Let $\aut = \tuple{Q_a, \delta_a, \colors_a, \rootstate_a, E_a}$ and 
$\but = \tuple{Q_b, \delta_b, \colors_b, \rootstate_b, E_b}$ be SWTAs
over~$\abc$ such that, w.l.o.g., $Q_a \cap Q_b = \emptyset$ and $\colors_a \cap
  \colors_b = \emptyset$.
Then let $\cut = \tuple{Q_c, \delta_c, \colors_c, \rootstate_c, E_c}$ be an SWTA
over~$\abc$ constructed as follows:
\begin{itemize}
  \item  $Q_c = Q_a \cup Q_b \cup \{\rootstate_c\}$ with $\rootstate_c \notin
    Q_a \cup Q_b$,
  \item  $\colors_c = \colors_a \cup \colors_b$,
  \item  $E_c = E_a \cup E_b \cup G$ where $G = \{\rootstate_c\}$ if
    $\{\rootstate_a, \rootstate_b\} \cap (E_a \cup E_b) \neq \emptyset$, else $G
    = \emptyset$, and
  \item  $\delta_c = \delta_a \cup \delta_b \cup \{\transof{\rootstate_c} \alpha
    {\tcacc 7 e} {\linof \ell}{\linof r} \mid \transof{\rootstate_a} \alpha
    {\tcacc 7 e} {\linof \ell}{\linof r} \in \delta_a \lor
    \transof{\rootstate_b} \alpha
    {\tcacc 7 e} {\linof \ell}{\linof r} \in \delta_b\}$.
\end{itemize}
One can prove $\langof \cut = \langof \aut \cup \langof \but$ in the standard
way.
\end{proof}

\begin{theorem}\label{thm:isect_compl}
SWTAs are not closed under language intersection and language complement.
\end{theorem}

\begin{proof}
The non-closure under language intersection is a~direct corollary of
\cref{thm:swta-emptiness} (language intersection emptiness is undecidable) and
\cref{thm:swta-emptiness} (language emptiness is decidable).

For the non-closure under complement, it is easy to notice that the languages of
SWTAs are all of a~countable size, while the set of all perfect trees with
complex-valued leaves is uncountable.
\begin{changebar}
An alternative proof can be obtained by assuming (for the sake of
contradiction) closure under complement and then, from \cref{thm:union} and
\cref{thm:swta-emptiness} and using De Morgan's laws, one would obtain
decidability of the intersection emptiness, which contradicts
\cref{thm:isect_empt_undec}.
\end{changebar}
\end{proof}

\vspace{-0.0mm}
\subsection{Transducer Image}\label{sec:trn_image}
\vspace{-0.0mm}

In the section, we will give an algorithm for computing an SWTA
$\trn_2(\aut_1)$ representing the language $\trn_2(\langof{\aut_1})$ where
$\aut_1 = \tuple{Q_1, \delta_1, \colors, \rootstate_1, E_1}$ is an SWTA
over~$\abc$ and
$\trn_2 = \tuple{Q_2, \delta_2, \rootstate_2, E_2}$ is a~WTT over~$\abc$.
First, we introduce some useful notation.
Let $q(d) \in Q_2(\leftT, \rightT)$ and $\linof{\ell}, \linof{r} \in
\linformsof{Q_1}$
\begin{changebar}
(we recall that $Q_2(\leftT, \rightT)$ represents the set of
ground terms of the form $q(\leftT, \rightT)$ for $q \in Q_2$),
\end{changebar}
then we use $q(d)(\linof \ell, \linof r)$ to denote the linear form over $Q_1
\times Q_2$ defined as
\begin{equation}
 q(d)(\linof \ell, \linof r) = \begin{cases}
   \sum_{p \in \suppof{\linof \ell}} \linindexof \ell p \cdot \tuple{p, q} & \text{if } d = \leftT, \\
   \sum_{p \in \suppof{\linof r}} \linindexof r p \cdot \tuple{p, q} & \text{if } d = \rightT .
 \end{cases}
\end{equation}
For instance, $q(\leftT)(a p_a + b p_b, 0) = a \tuple{p_a, q} + b
\tuple{p_b, q}$.
We extend the notation to linear forms $\linof x \in
\linformsof{Q_2(\leftT,\rightT)}$ such that
$\linof x(\linof \ell, \linof r)$ is obtained from $\linof x$ by substituting every
occurrence of~$q(d)$ in~$\linof x$ by $q(d)(\linof \ell, \linof r)$ and
multiplying and summing up the coefficients to obtain a~linear form over $Q_1
\times Q_2$.

\begin{example}
Consider the linear form $\linof x \in \linformsof{Q_2(\leftT, \rightT)}$
and linear forms $\linof \ell, \linof r \in \linformsof{Q_1}$ such that
\begin{equation}
  \linof x =  \invsqrttwo q(\leftT) + q(\rightT) - 3s(\rightT),
  \qquad
  \linof \ell = p + 0u,\quad \text{and}\quad
  \linof r = -2p + \invsqrttwo u.
\end{equation}
Then
\begin{equation}
\begin{aligned}
  \linof x(\linof \ell, \linof r) & {}= 
\invsqrttwo\tuple{p,q} + 0\tuple{u,q}
-2\tuple{p,q} +\invsqrttwo\tuple{u,q}
+6\tuple{p,s} - \smallfrac 3 {\sqrt 2}\tuple{u,s} \\
  &{}= 
\smallfrac{1-2\sqrt 2}{\sqrt 2}\tuple{p,q} +\invsqrttwo\tuple{u,q}
+6\tuple{p,s} - \smallfrac 3 {\sqrt 2} \tuple{u,s},
\end{aligned}
\end{equation}
which is the resulting linear form.
\qed
\end{example}

The image of~$\aut_1$ w.r.t.~$\trn_2$ is then the SWTA $\trn_2(\aut_1) = \aut_3 = \tuple{Q_3,
\delta_3, \colors, \rootstate_3, E_3}$ where the components are defined in the
following way:
\begin{itemize}
  \item  $Q_3 = Q_1 \times Q_2$,  $\rootstate_3 = \tuple{\rootstate_1, \rootstate_2}$, $E_3 = E_1 \times E_2$, and
  \item  $\delta_3 = \{\transof{\tuple{q_1, q_2}} a {\tcacc 7 c}
    {\linof{\ell_2}(\linof{\ell_1}, \linof{r_1})}{\linof{r_2}(\linof{\ell_1}, \linof{r_1})} \mid
    \transof{q_1} a {\tcacc 7 c}{\linof{\ell_1}}{\linof{r_1}} \in \delta_1,
    \trntransof{q_2} a {\linof{\ell_2}}{\linof{r_2}} \in \delta_2\}$.
\end{itemize}
%
An example illustrating the computation of a transducer image can be found
in~\trOrAppendix{sec:image_compos_examples}.

\begin{theorem}\label{thm:image}
$\langof{\aut_3} = \trn_2(\langof{\aut_1})$.
\end{theorem}

\vspace{-0.0mm}
\subsection{Transducer Composition}\label{sec:trn_compose}
\vspace{-0.0mm}

Below, we give a~construction of a~transducer~$\trn_{\compose} = \trn_2 \compose
\trn_1$ representing the \emph{composition} of functions denoted by transducers 
$\trn_1 = \tuple{Q_1, \delta_1, \rootstate_1, E_1}$ and~$\trn_2 = \tuple{Q_2, \delta_2,
\rootstate_2, E_2}$ over~$\abc$.
We first extend the notation from the previous section such that for $q(s) \in
Q_2(\leftT, \rightT)$ with $s \in \{\leftT,
\rightT\}$ and $\linof \ell, \linof r \in
\linformsof{Q_1(\leftT,\rightT)}$, we use $q(s)(\linof \ell, \linof r)$ to denote the linear
form over $(Q_1 \times Q_2)(\leftT, \rightT)$ defined as
\begin{equation}
 q(s)(\linof \ell, \linof r) = \begin{cases}
   \sum_{\tuple{p,\leftT} \in \suppof{\linof \ell}} \linindexof \ell p \cdot \tuple{p, q}(\leftT) +
   \sum_{\tuple{p,\rightT} \in \suppof{\linof \ell}} \linindexof \ell p \cdot \tuple{p, q}(\rightT)
   & \text{if } s = \leftT, \\
   \sum_{\tuple{p,\leftT} \in \suppof{\linof r}} \linindexof r p \cdot \tuple{p, q}(\leftT) +
   \sum_{\tuple{p,\rightT} \in \suppof{\linof r}} \linindexof r p \cdot \tuple{p, q}(\rightT)
   & \text{if } s = \rightT .
 \end{cases}
\end{equation}
For instance, $q(\leftT)(a p_a(\leftT) + b p_b(\rightT), 0) = a \tuple{p_a,
q}(\leftT) + b \tuple{p_b, q}(\rightT)$.
Similarly as in the previous section,
we extend the notation from states $q \in Q_2$ to linear forms $\linof x \in
\linformsof{Q_2(\leftT,\rightT)}$ such that
$\linof x(\linof \ell, \linof r)$ is obtained from $\linof x$ by substituting every
occurrence of~$q(s)$ in~$\linof x$ by $q(s)(\linof \ell, \linof r)$ and
multiplying and summing up the coefficients to obtain a linear form over $(Q_1
\times Q_2)(\leftT, \rightT)$.

The composition $\trn_2 \compose \trn_1$ is then
the transducer~$\trn_{\compose} = (Q_{\compose}, \delta_{\compose},
\rootstate_{\compose}, E_{\compose})$ with its components defined as follows:
\begin{itemize}
  \item  $Q_{\compose} = Q_1 \times Q_2$, $\rootstate_{\compose} = \tuple{\rootstate_1, \rootstate_2}$, $E_{\compose} = E_1 \times E_2$, and
  \item  $\delta_{\compose} = \{
    \trntransof{\tuple{q_1, q_2}} a {\linof{\ell_2}(\linof{\ell_1}, \linof{r_1})}
                                    {\linof{r_2}(\linof{\ell_1}, \linof{r_1})}
                              \mid
    \trntransof{q_1} a {\linof{\ell_1}}{\linof{r_1}} \in \delta_1,
    \trntransof{q_2} a {\linof{\ell_2}}{\linof{r_2}} \in \delta_2\}$.
\end{itemize}
As usual, $\trn_{\compose}$ constructed by the procedure above may contain some
unreachable states and transitions; we therefore use a~reachability-based
construction starting from $\tuple{\rootstate_1, \rootstate_2}$ in the standard
way. An example showing how to compute a transducer composition can be found
in~\trOrAppendix{sec:image_compos_examples}.

\begin{theorem}\label{thm:trans_compose_correct}
For all trees $t \in \abctrees$ and WTTs~$\trn_1$ and~$\trn_2$ over~$\abc$,
  it holds that $(\trn_2 \compose \trn_1)(t) = \trn_2(\trn_1(t))$.
\end{theorem}


\vspace{-0.0mm}
\section{Transducers for Quantum Gates}\label{sec:quantumWTT}
\vspace{-0.0mm}

In this section, we provide constructions of the transducers for a broad class
of quantum gates, including those forming a universal gate set and quantum
Fourier transform (QFT). Building on the composition algorithm introduced
in~\cref{sec:trn_compose}, we construct the transducer for a fixed-size
circuit---referred to as a~\emph{box}---and further present an algorithm
systematically generating a~transducer for size-parameterized circuit families
of a specified repetition pattern. 
Importantly, the expressive power of transducers extends beyond such structured
design: they are capable of describing more general families of circuits. For
example, as shown in~\cref{sec:BV}, the transducer of the Bernstein-Vazirani
(BV) circuit cannot be \mbox{obtained through this pattern-based instantiation alone.}

\vspace{-0.0mm}
\subsection{Transducers for Atomic Quantum Gates}
\vspace{-0.0mm}
We assume that a~quantum circuit is operating on $m$ qubits, labeled in order
as $x_1, \ldots, x_m$.
The two main types of atomic quantum gates used in state-of-the-art quantum computers are single-qubit gates and (multi-)controlled gates. In general, a single-qubit gate is represented as a \emph{unitary complex matrix} $ 
U=\big(\begin{smallmatrix} a &b \\ c &d \end{smallmatrix}\big)$. 
A~controlled gate $CU$ uses another quantum gate~$U$ as its parameter.
$CU$~consists of a~control qubit~$x_i$ and the gate $U$ is applied only when the
control qubit~$x_i$ has value~$1$. 
In this section, we construct the transducers for a broad class of quantum
gates, including those forming a~universal gate set and quantum Fourier
transform (QFT), operating on a quantum system of a~fixed size~$m$.
Let us fix a~single-qubit gate $ U= \big(\begin{smallmatrix} a &b \\ c &d \end{smallmatrix}\big) $.

\paragraph{Single Qubit Gates}
\begin{changebar}
The transducer $\trn_{U,i,m}$ for a~unitary~$U$ operating on qubit~$x_i$ in a circuit with~$m$ qubits is constructed as $\trn_{U,i,m}=\tuple{Q, \delta, q_0, E}$ over the alphabet $\abc = \{x_1, \ldots, x_m\}$ where 
$Q= \{q_j \mid 0 \leq j \leq m\}$,
$E= \{q_m\}$, and $\delta$ consists of the following transitions:
\begin{align}
  q_{j} &\to x_{j+1}( q_{j+1} (\leftT) , q_{j+1} (\rightT)) && \text{for } j \neq i-1, j<m  \label{eqn:counter} \\
     q_{i-1} & \to  x_i(a q_i (\leftT) +  b q_i (\rightT),c q_i (\leftT) + d q_i (\rightT)) &&  \label{eq:gateApplication} 
\end{align}
Intuitively, we have a number of states $q_j$, which are used to count from~$0$ to~$m$ as shown in \cref{eqn:counter}. The effect of applying $U$ at $q_{i-1}$ is captured as~\cref{eq:gateApplication}.
\end{changebar}
This concludes the description of the transducer $\trn_{U,i,m}$. 
In the case where $U=I$ is the identity matrix, the transducer just corresponds to a~wire and the index $i$ does not matter. In this case we simply write $\trn_{I,m}$. 

\newcommand{\fzero}{\mathsf{zero}}
\paragraph{Multi-controlled Gates}
\begin{changebar}
Let $i < j \leq m$. We will explain our construction of the transducer $\trn_{C_jU_i,m}$ for the controlled-$U$ gate $C_jU_i$ with control qubit $x_j$ and target qubit $x_i$ in a~circuit with~$m$ qubits.
Let $\trn_1=\trn_{I,m}$ and $\trn_2=\trn_{U,i,m}$. 
Observe that, after applying $C_jU_i$, the $0$-subtrees below $x_j$ remain the
same but the $1$-subtrees below $x_j$ are updated to the corresponding ones
after applying $U$ to qubit~$x_i$. Let us first define the operation
$\fzero\leftT_j$ on WTTs. The operation changes the input transducer such that the resulting transducer modifies the leaves of $0$-branches
below~$x_j$ to zero
in the image of the transducer.
Formally, given a~transducer~$\trn$ for a~single-qubit
gate, $\fzero\leftT_j (\trn)$ is obtained by taking~$\trn$ and replacing the
transition 
\end{changebar}
\begin{equation}
\trntransof {q} {x_j}{\linof{\ell_1} (\leftT) +  \linof{r_1} (\rightT)}{\linof{\ell_2} (\leftT) + \linof{r_2} (\rightT)},
\end{equation}
which acts on the qubit~$x_j$, by the transition
\begin{changebar}
\begin{equation}
\trntransof {q} {x_j}{\fzero_{\linof{\ell_1}}(\leftT) +  \fzero_{\linof{r_1}} (\rightT)}{\linof{\ell_2} (\leftT) + \linof{r_2} (\rightT)}.
\end{equation}
Given a~linear form~$\linof v$, the linear form $\fzero_{\linof v}$ is obtained from~$\linof v$ by modifying all its coefficients to zero.
\end{changebar}
All transitions of~$\trn$ occurring over other symbols are retained. 
Note that we can similarly define $\fzero\rightT_j$, which modifies the
coefficients of the right-hand child. 
We set $\trn'_1=\fzero\rightT_j(\trn_1)$ and $\trn'_2 = \fzero\leftT_j(\trn_2)$. 
Finally, we construct 
\begin{equation}
\trn_{C_jU_i,m}=\trn'_2+\trn'_1=  \fzero\leftT_j(\trn_{U,i,m}) + \fzero\rightT_j(\trn_{I,m}).
\end{equation}
Here, addition of transducers is an operation that produces a~transducer
$\trn_a+\trn_b$, such that for a~tree~$t$, it holds that
$(\trn_a+\trn_b)(t) = \trn_a(t)+\trn_b(t)$.  
\begin{changebar}
Formally, $\trn_a + \trn_b$ is defined as follows.
Let $\trn_a=(Q_a,\delta_a,\rootstate_a, E_a)$ and
$\trn_b=(Q_b,\delta_b,\rootstate_b, E_b)$ be WTTs over~$\abc$ where $Q_a$ and
$Q_b$ are disjoint.
Let $t_a=\trntransof {\rootstate_a} {x_1}{\linof{\ell_1} (\leftT) +  \linof{r_1}
(\rightT)}{\linof{r_2} (\rightT) + \linof{\ell_2} (\leftT)} \in \delta_a$ and
$t_b=\trntransof {\rootstate_b} {x_1}{\linof{\ell'_1} (\leftT) +  \linof{r'_1}
(\rightT)}{\linof{r'_2} (\rightT) + \linof{\ell'_2} (\leftT)} \in \delta_b$.
Then $\trn_a+\trn_b=\trn_{C_jU_i,m}=(Q,\delta,\rootstate,E)$ where 
$Q=(Q_a \cup Q_b \cup \{\rootstate\}) \setminus \{\rootstate_a,\rootstate_b\}$,
$E=E_a \cup E_b$, and~$\delta$ contains all transitions in $(\delta_a \cup
\delta_b) \setminus \{t_a, t_b\}$ and also the transition $\trntransof {\rootstate}
{x_1}{(\linof{\ell_1}+ \linof{\ell'_1}) (\leftT) +  (\linof{r_1}+\linof{r'_1})
(\rightT)}{(\linof{r_2}+\linof{r'_2}) (\rightT) +
(\linof{\ell_2}+\linof{\ell'_2}) (\leftT)}$. 
\end{changebar}

The construction can be directly generalized to multi-controlled gates. For example, $$\trn_{C_kC_jU_i,m}= 
\fzero\leftT_k(\trn_{C_jU_i,m}) + \fzero\rightT_k(\trn_{I,m})=
\fzero\leftT_k(\fzero\leftT_j(\trn_{U,i,m}) + \fzero\rightT_j(\trn_{I,m})) + \fzero\rightT_k(\trn_{I,m}).$$ 

\begin{theorem}\label{thm:gates_WTT}
    The transducers for single-qubit gates and multi-controlled gates are semantically correct. Moreover, each transducer can be constructed using~$\bigO(m)$ states and~$\bigO(m)$ transitions, where~$m$ is the number of qubits in the circuit.
\end{theorem}

\newcommand{\figQFTcircuit}[0]{
\begin{figure}[t]
\resizebox{\textwidth}{!}{
\input{Figures/qft-circuit.tikz}
}
\vspace{-3mm}
\caption{The QFT circuit, where $R_{k}={\displaystyle \bigl( \begin{smallmatrix}1&0\\0&\gamma_k\end{smallmatrix} \bigr)}$ and $\gamma_k=\omega^{2^{n-k}}$ with $\omega = e^{ \frac{2\pi i}{n} }$.}
\label{fig:QFTcircuit}
\vspace*{-5mm}
\end{figure}
}

\begin{figure}[t]
\resizebox{\textwidth}{!}{
\input{Figures/qft-circuit.tikz}
}
\vspace{-3mm}
\caption{The QFT circuit, where $R_{k}={\displaystyle \bigl( \begin{smallmatrix}1&0\\0&\gamma_k\end{smallmatrix} \bigr)}$ and $\gamma_k=\omega^{2^{n-k}}$ with $\omega = e^{ \frac{2\pi i}{n} }$.}
\label{fig:QFTcircuit}
\vspace*{-5mm}
\end{figure}
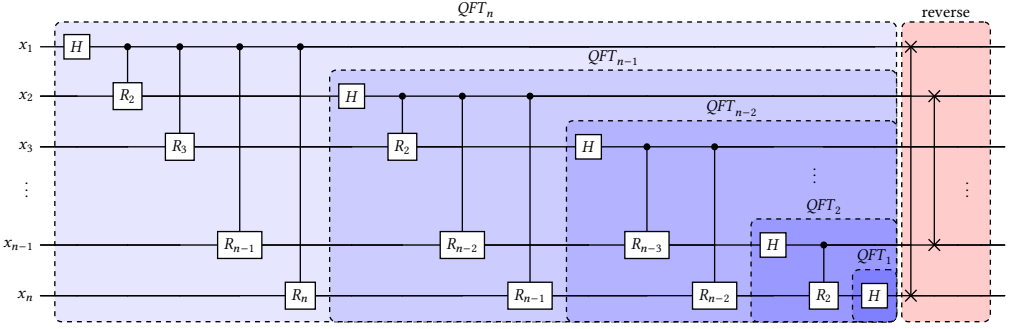

\paragraph{Quantum Fourier Transform Gate}
Given the central role of the Quantum Fourier Transform~(QFT), given in \cref{fig:QFTcircuit}, in many quantum algorithms, 
we explicitly discuss the construction of the transducer for the standard circuit implementation of the $n$-qubit QFT gate. 
The construction can be derived directly using the transducers for atomic gates and composition procedure described in~\cref{sec:trn_compose}. 
The construction is detailed in~\trOrAppendix{sec:QFT}, here we provide a short overview to highlight its structure and efficiency. 
%
The transducer for $\qft_{[1\ldots n]}$ acting on a~sequence of~$n \leq m$ contiguous qubits requires $\bigO(n^2)$ states and transitions and 
to generalize to the full $m$-qubit system setting, we append $\bigO(m)$ additional states and transitions in order to propagate qubits outside the QFT range unchanged.
Overall, the transducer has $\bigO(n^2+m)$ states and $\bigO(n^2+m)$ transitions.
\begin{changebar}
We~emphasize that the transducer is obtained via the standard composition
construction, which has an exponential upper bound on the size of the output
(because it is, essentially, a~product construction and we are applying
it~$\bigO(n^2)$ times).
The reason for the only quadratic size of the result comes from the structure
of the actual transducers being composed.
\end{changebar}

Intuitively,
the transducer for QFT
is built from basic components, transducers representing gates with a parametric number of qubits:
the transducer $\htrn i$ models the Hadamard gate applied on the qubit $x_i$, 
the transducer $\rtrn i$ models the
\begin{changebar}
`controlled-multi-rotation gate'---in
\end{changebar}
the figure, it represents the chain of gates $R_{2},\ldots,R_{n+1-i}$ in the $i$-th  blue
\begin{changebar}
`box'
\end{changebar}
from the left that are applied on qubits $x_{1+i},\ldots, x_{n}$ respectively and controlled by qubit~$x_i$.  
The transducer representing the whole QFT circuit is now obtained by composing these basic blocks from left to right as shown in the figure, as the transducer  
$((((\htrn 1\compose \rtrn 1)\compose \htrn 2) \compose \rtrn 2) \compose \cdots \compose  )\compose\htrn n$.
We do not model the reversion, the right-most red box, in the transducer, as we find it easier to include the reversion into the specification of the verification post-condition.
\begin{changebar}
Modelling the reversion would result in an exponential blow-up in the size of
the transducer.
\end{changebar}

\newcommand{\figMajSteps}[0]{
\begin{wrapfigure}[14]{r}{60mm}
\ifTR
\else
\vspace{-5mm}
\fi
\scalebox{0.7}{
\begin{quantikz}
  \lstick{$\ket{x_1}$}      & \gId & \gMAJp{1} & \gId  & \gId  & \gId & \qw\\
  \lstick{$\ket{x_2}$}      & \gId &       & \gId  & \gId  & \gId & \qw\\
    \lstick{$\ket{x_3}$}      & \gId &       & \gMAJp{2} & \gId  & \gId & \qw\\
  \lstick{$\ket{x_4}$}      & \gId & \gId  & \qw   & \gId  & \gId & \qw\\
    \lstick{$\ket{x_5}$} & \gId & \gId  &       & \gMAJp{3} & \gId & \qw\\
  \lstick{$\ket{x_6}$}      & \gId & \gId  & \gId  & \qw   & \gId & \qw\\
  \lstick{$\ket{x_7}$}      & \gId & \gId  & \gId  & \qw   & \gId & \qw\\[-5mm]
\hspace*{-5mm}\vdots\hspace*{5mm}\setwiretype{n} &   \vdots      &  \vdots & \vdots      & \vdots & \vdots \\[-2mm]
\end{quantikz}
}
\vspace{-6mm}
  \caption{Parameterized MAJ example}
  \label{fig:majid}
\end{wrapfigure}
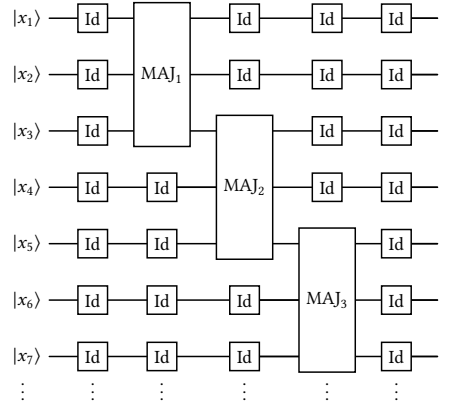%
}

\begin{remark}
In our use cases, the obtained trasducers were all of a~tractable size.
The transducer composition construction in the worst case, however, yields
a~transducer of a~quadratic size, so the size of a~transducer representing
the whole circuit might be exponential to the depth of a~circuit.
One can see this, e.g., when writing a~transducer modelling reversal (as done,
e.g., in the case of QFT, cf.\ \cref{fig:QFTcircuit}), the size of which is
exponential (reversion is a~purely combinatorial operation that performs
a~given permutation on the leaves of tree).
\end{remark}

\newcommand{\algParamTrans}[0]{
\SetKw{Continue}{continue}
\SetKwFunction{DF}{DiscoverTransition}
\SetKwProg{Fn}{Subroutine}{:}{}

\begin{figure}[t]
\begin{algorithm}[H]
\caption{Construction of Size-Parameterized Transducers}
\label{alg:compose_PWTT}
\KwIn{A transducer $\trn=(Q,\delta,\rootstate, E)$ over $\Gamma$ satisfying conditions in the text\\
\phantom{\bf Input: }a number $n>0$ denoting the offset of boxes\\
\phantom{\bf Input: }$d \in \{\mathit{left}, \mathit{right}\}$ indicating the direction in which the next box should appear}
\KwOut{A size-parameterized transducer $C$ implementing the composition of~$\trn$}

Initialize $\delta_C \gets \emptyset$, $Q_C \gets \emptyset$, $E_C \gets \emptyset$\;
$\mathit{WList} \gets \{((\tuple{\rootstate}, \compbgn),n)\}$\;
\While{$\mathit{WList}\neq \emptyset$}{
    Pop $(\mathbf{q}=(\tuple{q_1, \dots, q_k}, s),i)$ from $\mathit{WList}$, add $\mathbf{q}$ to $Q_C$ and, if $\{q_1,\dots, q_k\} \subseteq E$, also to $E_C$\;
    \ForEach{$w\in \Gamma$}{
      $(\linof \ell, \linof r) \gets \delta(q_k,w)\circ\cdots\circ\delta(q_2,w)\circ\delta(q_1,w)$\label{ln:firstcomp}\;
      \lIf{$(\linof{\ell},\linof{r}) = \bot$}{\Continue}
      \If(\tcp*[f]{continuing in the current box(es)}){$i \neq 1 \lor s = \compend$}{
        \lIf{$s = \compbgn$}{\DF{$\trntransof{\mathbf{q}} w {\linof \ell} {\linof r}, s, i-1$}}
        \lElse{\DF{$\trntransof{\mathbf{q}} {w'} {\linof \ell} {\linof r}, s, i-1$}}
      }
      \Else(\tcp*[f]{potentially start a new box}){
        $\DF{$\trntransof{\mathbf{q}} {w'} {\linof \ell} {\linof r}, \compend, i-1$}$\tcp*{no new box}
        \lIf(\tcp*[f]{new box to the left}){$d = \mathit{left}$\label{ln:secondcomp}}{
          $(\linof{\ell'}, \linof{r'}) \gets (\linof \ell, \linof r) \circ (\rootstate (\leftT),\rootstate(\rightT))$
        }
        \lElse(\tcp*[f]{new box to the right}){\label{ln:thirdcomp}
          $(\linof{\ell'}, \linof{r'}) \gets (\rootstate (\leftT),\rootstate(\rightT)) \circ (\linof \ell, \linof r)$
        }
        \DF{$\trntransof{\mathbf{q}} w {\linof{\ell'}}{\linof{r'}}, \compbgn, n$}\;
      }
    }
    \Return $(Q_C,\delta_C,(\tuple{\rootstate}, \compbgn),E_C)$\;
}
\BlankLine
\Fn{\DF{$\trntransof {\mathbf{q}} {u} {\linof{\ell}} {\linof{r}}, \mathit{StateTag}, \mathit{Pos}$}}{
  $\linof{\ell'} \gets \linof\ell\specsubst{\mathbf{t}}{(\mathbf{t}, \mathit{StateTag})}$;
  $\linof{r'} \gets \linof r\specsubst{\mathbf{t}}{(\mathbf{t}, \mathit{StateTag})}$\label{ln:strangesubst}\;
  $\delta_C \gets \delta_C \cup \{\trntransof{\mathbf{q}} {u} {\linof{\ell'}} {\linof{r'}}\}$\;
  \ForEach{$\mathbf{p} \in \suppof{\linof{\ell'}} \cup \suppof{\linof{r'}}$}{
    \lIf{$\mathbf{p}\notin Q_C \land \nexists j(\mathbf{p},j)\in \mathit{WList}$}{
      Add $(\mathbf{p},\mathit{Pos})$ to $\mathit{WList}$
    }
  }
}
\end{algorithm}
\vspace*{-4mm}
\end{figure}
}

\newcommand{\algParamTransNew}[0]{
\SetKw{Continue}{continue}
\begin{algorithm}[t]
\caption{Construction of Size-Parameterized Transducers \ol{????}}
\label{alg:compose_PWTT}
\KwIn{A transducer $\trn=(Q,\delta,\rootstate, E)$ over $\Gamma$ \mh{definined only on trees of height $H$} \\
  \phantom{\bf Input: }a number $n>0$ denoting the offset of boxes\\
}
\KwOut{Composed transducer $C$ for the entire circuit \ol{what?}}

Initialize $\delta_C \gets \emptyset$, $Q_C \gets \emptyset$, $E_C \gets \emptyset$\;
$\mathit{WList} \gets \{((\tuple{\rootstate}, \compbgn),0)\}$, $\mathit{Known} \gets \{((\tuple{\rootstate}, \compbgn),0)\}$\;
\While{$\mathit{WList}\neq \emptyset$}{
    Pop $(\mathbf{q}=(\tuple{q_1, \dots, q_k}, s),i)$ from $\mathit{WList}$, add $\mathbf{q}$ to $Q_C$ and, if $\{q_1,\dots, q_k\} \subseteq E$, also to $E_C$\;
    \ForEach{$w\in \Gamma$}{
      \If{$i = H$}
          {$(\linof \ell, \linof r) \gets \delta(q_k,w)\circ\cdots\circ\delta(q_2,w)$, $d \gets H - n$ \tcp*{leftmost box ends}} 
      \Else
          {$(\linof \ell, \linof r) \gets \delta(q_k,w)\circ\cdots\circ\delta(q_2,w)\circ\delta(q_1,w), d \gets i + 1$}
      \lIf{$(\linof{\ell},\linof{r}) = \bot$}{\Continue}
      \If(\tcp*[f]{decide whether staircase continues}){$i = n - 1 \land s = \compbgn$}{
        $\linof{\ell_\mathit{Cont}} \gets \linof \ell \circ (\rootstate (\leftT), \rootstate(\rightT))$;
        $\linof{r_\mathit{Cont}} \gets \linof r \circ (\rootstate (\leftT), \rootstate(\rightT))$\;
        $\mathsf{DiscoverTransition}({\trntransof{\mathbf{q}} {w} {\linof{\ell_\mathit{Cont}}} {\linof{r}_\mathit{Cont}}}, \compbgn)$\;
        $\mathsf{DiscoverTransition}({\trntransof{\mathbf{q}} {w'} {\linof{\ell}} {\linof{r}}}, \compend)$\;
      }
      \Else{
        $\mathsf{DiscoverTransition}({\trntransof{\mathbf{q}} {w} {\linof{\ell}} {\linof{r}}}, s)$\;
      }
    }
  }
  \Return $(Q_C,\delta_C,(\tuple{\rootstate}, \compbgn),E_C)$;

    \SetKwFunction{DF}{DiscoverTransition}
    \SetKwProg{Fn}{Subroutine}{:}{}
    \Fn{\DF{$\trntransof {\mathbf{q}} {w} {\linof{l}} {\linof{r}}, \mathit{StateTag}$}}{
      $\linof{\ell'} \gets \linof\ell\specsubst{\mathbf{t}}{(\mathbf{t}, \mathit{StateTag})}$;
      $\linof{r'} \gets \linof r\specsubst{\mathbf{t}}{(\mathbf{t}, \mathit{StateTag})}$\;
      $\delta_C \gets \delta_C \cup \{\trntransof{\mathbf{q}} {w} {\linof{\ell'}} {\linof{r'}}\}$\;
      \ForEach{$\mathbf{p} \in \suppof{\linof{\ell'}} \cup \suppof{\linof{r'}}$}{
        \lIf{$(\mathbf{p}, d) \notin \mathit{Known}$}{
          Add $(\mathbf{p}, d)$ to $\mathit{WList}$ and to $\mathit{Known}$
        }
      }
    }
\end{algorithm}
}

\vspace{-0.0mm}
\subsection{Constructing Size-Parameterized Transducers}\label{sec:composition_par}
\vspace{-0.0mm}
In this section, we describe our algorithm for systematically generating the transducer from a~fixed-size circuit---called a~\emph{box}---for a~size-parameterized circuit family with a~specific repetition pattern.

To illustrate this process step by step, we begin with a concrete example from~\cref{fig:rca}, where the circuit consists of a sequence of $\maj$ boxes.
We first build the transducer for this fixed-size $\maj$ circuit using the
transducer composition procedure from~\cref{sec:trn_compose}, and then
generalize it by uniformly removing subscripts from wire labels---for example,
replacing $x_i$ with $x$---to indicate that variables now serve as symbolic
placeholders rather than fixed wire indices.
Wires not involved in active gates are treated as carrying the identity operation,
denoted as $\mathrm{Id}$ (see~\cref{fig:majid}; for clarity, we retain
subscripted indices in the figure).
The transducer for~$\mathrm{Id}$ consists of a~single state~$\idf$
(both a~root and a~leaf state) and the transition $\trntransof \idf x
\idf \idf$.

\ifTR\else
\figMajSteps   
\fi


A key structural property of this circuit is that each qubit is acted upon by at most two non-identity boxes.
Moreover, different qubits can exhibit equivalent transducer behavior. For
example, qubits~$x_3$ and~$x_5$ are each processed by two instances of the same
$\maj$ transducer (both working with the same set of states).
Suppose that, when reading qubit~$x_3$, the composition of the transducers (including $\maj_1$ and $\maj_2$) is
in the state $\tuple{\dots, \idf, q_1, s_1, \idf, \dots}$, and when
reading~$x_5$, it is in the same configuration
except one additional $\idf$ precedes~$q_1$. The effect on both wires is
identical,
\ifTR
\figMajSteps   
\fi
indicating that the number of surrounding $\idf$ states does not influence the
computation. We can therefore abstract such configurations using the finite
tuple $\tuple{q_1, s_1}$.
Although this observation may appear trivial---since composing identity with
itself yields identity---it plays a~crucial role in our construction. Once an
active box becomes inactive, it behaves like~$\idf$, and the relevant state
tuple shifts rightwards.
This shifting mechanism governs how the transducer evolves as the circuit
grows, and is fundamental to our parameterized transducer construction.

We illustrate the idea behind our construction using the $\maj$-gate cascade
from \cref{fig:majid}.
Let $\maj$ be represented by a transducer with a~root state~$a$, leaf
state~$\idf$, and the following transitions:
\begin{equation}
\begin{aligned}
  \initmark & \rlap{$\trntransof a x {b(\leftT) + c(\rightT)}{e(\leftT) + d(\rightT)}$} &  &                                          &  & \trntransof f x {\idf(\leftT)} {0\idf(\rightT)}       & \quad & \trntransof k x {0\idf(\rightT)} {\idf(\leftT)} \\
            & \trntransof b x {f(\leftT)}{f(\rightT)}                                   &  & \trntransof d x {f(\leftT)}{k(\rightT)}  &  & \trntransof g x {\idf(\rightT)} {0 \idf(\leftT)}     & \quad & \trntransof \idf x {\idf(\leftT)} {\idf(\rightT)}  \\
            & \trntransof c x {h(\rightT)}{h(\leftT)}                                   &  & \trntransof e x {h(\rightT)}{g(\leftT)}  &  & \trntransof h x {0 \idf(\leftT)} {\idf(\rightT)}
\end{aligned}
\end{equation}
Notice that the transducer contains the state~$\idf$ and transition
$\trntransof \idf x {\idf(\leftT)} {\idf(\rightT)}$ from the $\mathrm{Id}$
transducer---this is our way of implementing the sequence of $\mathrm{Id}$
gates below the box (we give formal requirements that the input transducer
needs to satisfy later).

We now describe how to compute the transducer for the entire size-parameterized
circuit.
Our product construction
starts in the state $\tuple{a}$, representing the circuit acting on the first qubit.
The transition from $\tuple{a}$ is the same as the one of $\maj$
from~$a$, just with states decorated by~$\tuple \cdot$ to be formally
a~one-tuple (since there is just one transducer active at this point):
\begin{equation}
	\trntransof{\tuple a} x {\tuple b(\leftT) + \tuple c(\rightT)}{\tuple e(\leftT)+\tuple d(\rightT)}.
\end{equation}
Consider next the state $\tuple b$. 
As this point, one should determine whether the cascade continues with another
$\maj$ or ends here (since the endpoints of the transition from $\tuple b$ will
need to either be pairs of states from the $\maj$ transducer or just one state
from~$\maj$ that finishes reading out the rest of the input tree without
spawning new $\maj$ instances).
One could naturally implement this via nondeterminism, which is not possible
with the WTT model since it is (top-down) deterministic.
We can, however, get around the limitation by slightly changing the input
alphabet: adding
a~primed version~$x'$ of the symbol~$x$, which is used to label the last two
qubits in the tree (encoding the fact that there is no further occurrence of
$\maj$).
Naturally, one also needs to change the alphabet and structure of the input
SWTA to make it compatible with the WTT, which may need its partial unfolding.
In the case that we do not start a new instance of~$\maj$, we continue with the
transition
\begin{equation}
  \trntransof{\tuple b} {x'} {\tuple f(\leftT)}{\tuple f(\rightT)}.
\end{equation}
Alternatively, we can continue spawning more $\maj$ instances, using the transition
\begin{equation}
	\trntransof{\tuple b} x {\tuple{f, a}(\leftT)}{\tuple{f, a}(\rightT)},
\end{equation}
where the `$a$'-component of $\tuple{f, a}$ is the root state of~$\maj$,
meaning that the next $\maj$ box is active.

\algParamTrans  

Next, we compute the transition from~$\tuple{f,a}$ similarly to
transducer composition (cf.\ \cref{sec:trn_compose}):
\begin{equation}
  \trntransof{\tuple{f, a}} x
  {\tuple{\idf,b}(\leftT) + 0\tuple{\idf, c}(\rightT)}
  {\tuple{\idf, e}(\leftT) + 0\tuple{\idf, d}(\rightT)}.
\end{equation}
We, however, note that~$\idf$ is just an identity (the first transducer became
inactive) and the state can therefore be removed from the tuples, yielding
instead the transition
\begin{equation}
  \trntransof{\tuple{f, a}} x
  {\tuple b(\leftT) + 0\tuple c(\rightT)}
  {\tuple e(\leftT) + 0\tuple d(\rightT)}.
\end{equation}

The full construction is more complex and is formalized in
\cref{alg:compose_PWTT}.
The input of the algorithm is a~transducer~$\trn=(Q,\delta,\rootstate, E)$ over
$\Gamma$ that needs to satisfy the following structural requirements:
\begin{enumerate}
  \item  $Q$ contains the state $\idf$ and $E = \{\idf\}$,
  \item for every $a \in \Gamma \cup \Gamma'$, the transition function $\delta$
    contains the transition $\trntransof \idf a {\idf(\leftT)}{\idf(\rightT)}$, and
  \item  all states from~$Q$ except~$\idf$ can occur only at a~particular
    depth (distance from the root) in a~run of~$\trn$ on any tree (and the state
    $\idf$ appears for the first time at the same depth on all branches);
    formally, there exists a~function $d\colon Q \to \naturals$ such that
    $d(\rootstate) = 0$ and for every $q \in Q \setminus \{\idf\}$, if $d(q) =
    m$ and $\trntransof q a {\linof \ell}{\linof r} \in \delta$, then for all
    $p \in \suppof{\linof \ell} \cup \suppof{\linof r}$ we
    have that $d(p) = m+1$.
    %
\end{enumerate}
Additional the inputs of the algorithm are the offset of boxes~$n$ ($n = 2$ in
the example above) and the side~$d \in \{\mathit{left}, \mathit{right}\}$ on
which the boxes grow ($d = \mathit{right}$ in the example).

The algorithm constructs a~WTT with states of the form~$(\tuple{q_1, \ldots, q_k},
\mathit{Tag})$ where~$ \tuple{q_1, \ldots, q_k}$ is a~sequence of
states of~$\trn$ and~$\mathit{Tag} \in \{\compbgn, \compend\}$ denotes whether
we can still spawn new boxes ($\compbgn$) or not any more ($\compend$).
For a~tuple $\tuple{q_1, \dots, q_k}$, we assume all $\idf$'s are automatically
deleted.
The set~$\mathit{WList}$ keeps states to be explored, together with
their offset (to know when to start spawning new instances of~$\trn$).
In the algorithm, we write $\delta(q, a) = \bot$ if $\delta$ is undefined on the
input $(q, a)$ and define the composition $\circ$ of pairs of linear forms
(used on Lines~\ref{ln:firstcomp}, \ref{ln:secondcomp}, and~\ref{ln:thirdcomp})
as follows:
\begin{inparaenum}[(i)]
  \item  $\bot \circ x = x \circ \bot = \bot$ for any~$x$ and
  \item  $(\linof{\ell_2}, \linof{r_2})\circ
    (\linof{\ell_1},\linof{r_1})=(\linof{\ell_2}(\linof{\ell_1}, \linof{r_1}),
                                    \linof{r_2}(\linof{\ell_1}, \linof{r_1}))$,
                                    using the notation from \cref{sec:trn_image}.
\end{inparaenum}
On \lnref{ln:strangesubst}, for a~linear form~$\linof v$, we use $\linof
v\specsubst{\mathbf{t}}{(\mathbf{t}, s)}$ to denote the linear form obtained
from~$\linof v$ by substituting every state (captured by) $\mathbf{t}$
occurring in~$\linof v$ by $(\mathbf{t}, s)$, e.g., for $\linof v = \tuple
b(\leftT) + \tuple c(\rightT)$, the result of~$\linof
v\specsubst{\mathbf{t}}{(\mathbf{t}, s)}$ is the linear form $(\tuple b,
s)(\leftT) + (\tuple c ,s)(\rightT)$.


\vspace{-0.0mm}
\section{Related Work}\label{sec:related}
\vspace{-0.0mm}
\begin{changebar}
In recent years, many techniques for analyzing, simulating, and verifying quantum circuits and programs have emerged in the formal methods community.
In this section, we provide their overview and explain where our approach stands
among them.
\end{changebar}

\paragraph{Symbolic Quantum Circuit Verifiers.}
These tools are often fully automated and flexible in specifying different verification properties. 
The closest work to ours in this category are the automata-based approaches of~\cite{ChenCLLTY23,ChenCLLT23,cacm25,cacm25-pa,abdulla2025verifying,chen2025autoq}, which use tree automata~\cite{tata} and their variant \emph{level-synchronized tree automata (LSTAs)} to encode predicates representing sets of quantum states.
SWTAs generalize the previous tree automata models by the use of colors for synchronization
  across branches (this was already in LSTAs), weights, and possible
interactions between multiple runs via linear forms.
These techniques help SWTAs alleviate the scalability issues of the previous
models in many cases where the representation using the previous model blew up.
More concretely, our single-qubit gate operations are linear, while those
of~\cite{ChenCLLTY23,ChenCLLT23} are exponential, and those
in~\cite{abdulla2025verifying,chen2025autoq} are quadratic
(on the negative side, language inclusion/equivalence are undecidable now, but
we use the decidable property of functional equivalence/inclusion to solve
the verification problem).
In addition, in the current paper, we propose a complete framework that uses
transducers for implementing the quantum gates,
whereas the previous works used specialized automata-manipulating algorithms.

There are a~few more tools that belong to this category:
\symqv~\cite{BauerMarquartLS23} is based on \emph{symbolic
execution}~\cite{King76} to verify input-output relationship with queries
discharged by SMT solvers over the theory of reals.
SMT solvers have also been used in the verification of \emph{quantum error correction code (QECC)}~\cite{chen2025verifying,fang2024symbolic}.
The SMT array theory approach of~\cite{chen2023theory} improved the 
previous by allowing a~polynomial-sized circuit encoding, but still faces
a~similar scalability problem.
Another scalable fully automated approach for analysis of quantum circuits is
  \emph{quantum abstract
  interpretation}~\cite{yu2021quantum,perdrix2008quantum}, which, however, uses
  a~quite coarse abstraction, which makes the properties that it can reason
  about quite limited.
Among the approaches mentioned above, the one based on
  LSTAs~\cite{abdulla2025verifying,chen2025autoq} is the only method that
  supports verification of quantum circuits where the size is a~parameter.
It can, however, handle only a~limited set of parametrized gates, e.g., it does not support the application of the Hadamard gate to every qubit, which is a~common pattern in most interesting circuits.

\paragraph{Deductive Quantum Circuit and Program Verifiers.}
The tools in this category allow verification of quantum programs
w.r.t.\ expressive specification languages.
  The most prominent family of approaches are those based on the so-called \emph{quantum Hoare
logic}~\cite{zhou2019applied,unruh2019quantum,feng2021quantum,ying2012floyd,liu2019formal}.
These approaches, however, require significant manual work and user expertise
(they are often based on the use of interactive theorem provers such as
\isabelle~\cite{nipkow2002isabelle} and \rocq/\coq~\cite{bertot2013interactive}). 
The work in~\cite{huang2025efficient} adopts a~syntactic approach that enables
automatic reasoning within a fragment of quantum Hoare logic, and applies it
to the verification of QECC.
The work most closely related to ours in this category is \qbricks~\cite{Chareton2021},
which also targets the verification of size-parametrized quantum circuits and leverages SMT solvers
to discharge verification conditions automatically.
While \qbricks attempts to improve automation by generating proof obligations and solving them using SMT,
the approach still requires substantial user interaction.
For example, in the case of verifying Grover's search for an arbitrary number of qubits, the experiment involved over one hundred manual interactions.

\paragraph{Quantum Circuit Equivalence Checkers.} These tools
are usually fully automated but are limited to only checking equivalence of two circuits.
In contrast, our approach is flexible in specifying custom properties.
Equivalence checkers are based on several approaches.
One approach is based on the \emph{ZX-calculus}~\cite{Coecke_2011}, which is a~graphical language used for reasoning about quantum circuits.
The \emph{path-sum} approach (implemented, e.g., within the tool
\feynman~\cite{amy2018towards}) uses rewrite rules.
Pre-computed equivalence sets are used to prove equivalence in
\quartz~\cite{xu2022quartz}.
\qcec~\cite{burgholzer2020advanced} is an equivalence checker that uses decision
diagrams and ZX-calculus, and \sliqec~\cite{ChenJH22,WeiTJJ22} also uses decision
diagrams for (partial) equivalence checking.
An approach based on working with the so-called \emph{stabilizer
states} in~\cite{ThanosCL23} can be used to verify the equivalence of circuits with
Clifford gates in polynomial time.
There have also been works that approach quantum equivalence via model
counting~\cite{MeiCoopmansLaarman2024Eq}.

\paragraph{Quantum Circuit Simulators.}
Quantum circuit simulators can be broadly categorized into four classes:
\emph{decision-diagram-based}~\cite{TsaiJJ21,ZulehnerW19,SistlaCR23,MillerT06,ChenJ25,Chen0JJL24},
\emph{state-vector-based}~\cite{li2021svsim},
\emph{tensor-network-based}~\cite{markov2008simulating,PhysRevLett.91.147902},
\emph{ZX-calculus-based}~\cite{Kissinger_2022}, and
\emph{model-counting-based}~\cite{MeiBonsangueLaarman2024Sim}.
These simulators are primarily designed to compute the output of a quantum circuit for a \emph{single input state}.
They can also be employed to verify circuit behavior over a finite number of input states by simulating each one individually.
However, they cannot be used for size-parametrized verification, which can be handled by the SWTA-based framework.

SWTAs can be viewed as a generalization of decision diagrams,
particularly the \emph{quantum multiple-valued decision diagrams (QMDDs)}~\cite{ZulehnerW19}
and \emph{tensor decision diagrams}~\cite{HongZLFY22},
which place weights on edges to compactly represent quantum states.
Unlike traditional simulators, our method is designed to handle \emph{sets of input states simultaneously}, improving scalability.
\hide{
Tensor-network-based simulation offers two notable advantages:
First, by representing quantum gates and states as \emph{tensors} (i.e., high-dimensional matrices), gate operations can be applied \emph{locally} without requiring global updates to the entire state.
Our approach also benefits from this property. Second, through \emph{singular value decomposition (SVD)}, quantum states can be approximated by truncating small singular values,
significantly reducing memory usage. However, the approximations lack formal
error bounds, making tensor-network methods unsuitable for applications
requiring \emph{precise verification guarantees}, like equivalence checking.}


\newcommand{\ackPhdTalent}[0]{
\noindent
The work of Michal Hečko, a~Brno Ph.D.\ Talent Scholarship
\raisebox{-6pt}{\protect\includegraphics[height=17pt,clip,trim={13mm 10mm 13mm 10mm}]{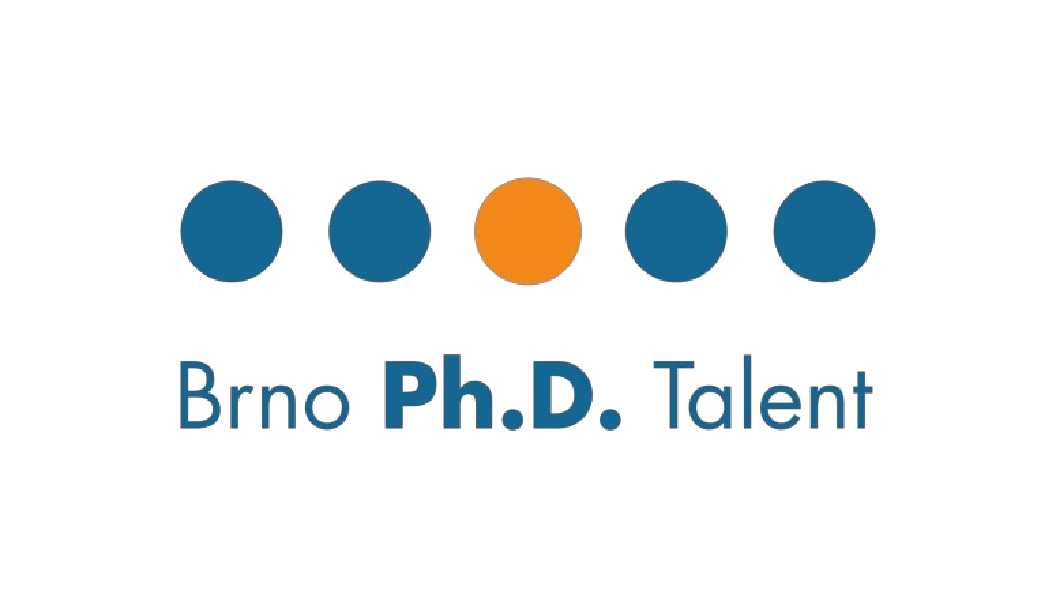}}
Holder, is funded by the Brno City Municipality.\xspace
}

\begin{acks}
We thank the anonymous reviewers for their feedback that improved the quality of the paper.
This work was supported by 
the Czech Science Foundation projects 25-18318S and 25-17934S;
the FIT BUT internal project FIT-S-23-8151;
National Science and Technology Council, R.O.C., projects
NSTC 114-2221-E-027-044 -MY2 and NSTC 114-2119-M-001-002-;
Air Force Office of Scientific Research project FA2386-23-1-4107; and
Academia Sinica Investigator Project Grant AS-IV-114-M07.
\ackPhdTalent

\end{acks}

\bibliographystyle{ACM-Reference-Format}
\bibliography{literature}

\ifTR

\appendix
\newpage

\crefalias{section}{appendix}

\input{appendix.tex}

\fi

\end{document}

%% file: paroshMacros.tex
\newcommand\xvar{x}

\newcommand{\notgate}{\mathsf{X}}
\newcommand{\cnotgate}\cnot
\newcommand\hgate{\mathsf{H}}

%% file: macros.tex
\newcommand{\ol}[1]{\textcolor{blue}{\ifmmode \text{[OL: #1]}\else [OL: #1] \fi}}

\newcommand{\yfc}[1]{\textcolor{green!50!blue}{\ifmmode \text{[YF: #1]}\else [YF: #1] \fi}}
\newcommand{\mh}[1]{\textcolor{teal}{\ifmmode \text{[MH: #1]}\else [MH: #1] \fi}}
\newcommand{\lh}[1]{\textcolor{magenta}{\ifmmode \text{[LH: #1]}\else [LH: #1] \fi}}
\newcommand{\jal}[1]{\textcolor{pink}{\ifmmode \text{[JA: #1]}\else [JA: #1] \fi}}

\newcommand{\newdef}[1]{}                  

\newcommand{\hide}[1]{}

\newcommand{\initmark}[0]{\to}

\newcommand{\naturals}[0]{\mathbb{N}}
\newcommand{\integers}[0]{\mathbb{Z}}
\newcommand{\complex}[0]{\mathbb{C}}

\newcommand{\abc}[0]{\Gamma}              
\newcommand{\colors}[0]{\Omega}
\newcommand{\partialto}[0]{\rightharpoonup}

\newcommand{\domof}[1]{\mathrm{dom}(#1)}
\newcommand{\height}[1]{h(#1)}

\newcommand{\semof}[1]{\llbracket #1 \rrbracket}
\newcommand{\semofof}[2]{\semof{#1}(#2)}

\newcommand{\trees}[0]{\mathbb{T}}
\newcommand{\treesof}[1]{\trees_{#1}}
\newcommand{\abctrees}[0]{\treesof \abc}

\newcommand{\concat}[0]{\mathbin{\bullet}}
\newcommand{\subtreeof}[2]{#1_{|#2}}

\newcommand{\treecons}[3]{\mathsf{cons}(#1, #2, #3)}

\newcommand{\ketof}[1]{\ket{#1}}

\makeatletter
\DeclareRobustCommand{\shortto}{%
  \mathrel{\mathpalette\short@to\relax}%
}

\DeclareRobustCommand{\shortminus}{%
  \mathrel{\mathpalette\short@minus\relax}%
}

\newcommand{\short@to}[2]{%
  \mkern2mu
  \clipbox{{.5\width} 0 0 0}{$\m@th#1\vphantom{+}{\rightarrow}$}%
}

\newcommand{\short@minus}[2]{%
  \mkern2mu
  \clipbox{{.5\width} 0 0 0}{$\m@th#1\vphantom{+}{-}$}%
}
\makeatother

\newcommand{\newlabeledto}[1]{\mbox{$\longrightarrow\rlap{\hspace*{-4.95mm}\raisebox{0.00mm}{\scalebox{0.8}{#1}}}$\!}}

\newcommand{\aut}[0]{\mathcal{A}}
\newcommand{\autex}[0]{\aut_{\mathit{ex}}}
\newcommand{\autbases}[0]{\aut_{\mathit{bases}}}
\newcommand{\but}[0]{\autb}
\newcommand{\cut}[0]{\mathcal{C}}
\newcommand{\butex}[0]{\but_{\mathit{ex}}}
\newcommand{\nut}[0]{\mathcal{N}}
\newcommand{\autb}[0]{\mathcal{B}}
\newcommand{\trn}[0]{\mathcal{T}}
\newcommand{\trnex}[0]{\trn_{\mathit{ex}}}
\newcommand{\transf}[0]{\mathcal{T}}
\newcommand{\lang}[0]{\mathcal{L}}
\newcommand{\langof}[1]{\lang(#1)}

\newcommand{\linforms}[0]{\mathbb{L}}
\newcommand{\linformsof}[1]{\linforms_{#1}}
\newcommand{\linof}[1]{\boldsymbol{#1}}
\newcommand{\linindexof}[2]{\linof{#1}[#2]}
\newcommand{\compose}[0]{\mathbin{\circ}}

\newcommand{\sumlfwrt}[2]{\sum {#1}(#2)}

\newcommand{\rootstate}[0]{\mathit{root}}

\newcommand{\supp}[0]{\text{st}}
\newcommand{\suppof}[1]{\supp(#1)}

\newcommand{\trntransof}[4]{#1 \to #2(#3, #4)}
\newcommand{\transof}[5]{#1 \newlabeledto{#3} #2(#4, #5)}

\newcommand{\transcof}[5]{\transof{#1}{#2}{\tacc{#3}}{#4}{#5}}

\newcommand{\trnimageof}[2]{#1(#2)}

\newcommand{\smallfrac}[2]{{\textstyle \frac{#1}{#2}}}
\newcommand{\invtwo}[0]{\smallfrac 1 2}

\newcommand{\invsqrttwo}[0]{{\smallfrac 1 {\sqrt{2}}}}

\newcommand{\tuple}[1]{\langle #1 \rangle}

\newcommand{\finstof}[1]{#1}

\newcommand{\leftT}[0]{\mathtt{L}}
\newcommand{\rightT}[0]{\mathtt{R}}

\newcommand{\pspace}[0]{\textsc{PSpace}\xspace}
\newcommand{\expspace}[0]{\textsc{ExpSpace}\xspace}

\newcommand{\coloraut}[1]{\mathbf{dom}(#1)}
\newcommand{\pre}[0]{\mathrm{Pre}}
\newcommand{\post}[0]{\mathrm{Post}}

\newcommand{\autRes}[0]{\aut_{\mathsf{Result}}}

\newcommand{\had}[0]{\mathsf{H}}
\newcommand{\hadtensor}[0]{\had^{\otimes}}
\newcommand{\sgate}[0]{\mathsf{S}}
\newcommand{\xgate}[0]{\mathsf{X}}
\newcommand{\cnot}[0]{\mathsf{CX}}
\newcommand{\ccnot}[0]{\mathsf{CCX}}
\newcommand{\rotx}[0]{\mathsf{R}_x(\frac \pi 2)}

\newcommand{\transfH}[0]{\trn_{\hadtensor}}
\newcommand{\transfCX}[0]{\transf_{\mathit{Oracle}}}

\newcommand{\ltstr}[0]{\mathit{Tr}}
\newcommand{\ltr}[1]{\xrightarrow{\smash{\raisebox{-0.7mm}{$\scriptstyle #1$}}}}

\newcommand{\bigO}[0]{\mathcal{O}}

\newcommand{\minus}[0]{\mathit{minus}}

\newcommand{\matrixof}[1]{\mathit{mat}(#1)}
\newcommand{\rowof}[1]{\mathit{row}(#1)}
\newcommand{\vectorof}[1]{\mathit{vec}(#1)}

\newcommand{\qft}[0]{\mathit{QFT}}

\newcommand{\higate}[1]{\mathsf{H}_{#1}}
\newcommand{\rgate}[1]{\mathsf{R}_{#1}}
\newcommand{\htrn}[1]{\trn_{\higate{#1}}}
\newcommand{\rtrn}[1]{\trn_{\rgate{#1}}}

\newcommand{\maj}[0]{\mathrm{MAJ}}

\newcommand{\gMAJp}[1]{\gate[3]{\mathrm{MAJ}_{#1}}}
\newcommand{\uma}[0]{\mathrm{UMA}}
\newcommand{\idgate}[0]{\mathrm{Id}}
\newcommand{\gId}[0]{\gate{\idgate}}
\newcommand{\idf}[0]{\iota}

\newcommand{\tool}[0]{\textsc{AutoQ-Para}\xspace}
\newcommand{\autoq}[0]{\textsc{AutoQ}\xspace}

\newcommand{\sliqec}[0]{\textsc{SliQEC}\xspace}
\newcommand{\feynman}[0]{\textsc{Feynman}\xspace}

\newcommand{\qcec}[0]{\textsc{Qcec}\xspace}

\newcommand{\qbricks}[0]{\textsc{Qbricks}\xspace}

\newcommand{\quartz}[0]{\textsc{Quartz}\xspace}
\newcommand{\symqv}[0]{\textsc{symQV}\xspace}

\newcommand{\isabelle}[0]{\textsc{Isabelle}\xspace}
\newcommand{\coq}[0]{\textsc{Coq}\xspace}
\newcommand{\rocq}[0]{\textsc{Rocq}\xspace}

\newcommand{\numof}[1]{\mathit{num}(#1)}
\newcommand{\revof}[1]{\mathit{rev}(#1)}
\newcommand{\encof}[1]{\mathit{enc}(#1)}

\tikzstyle{subcircstyle}=[dashed, rounded corners, inner xsep=2pt]
\tikzstyle{sublabstyle}=[label position=above,anchor=south,yshift=-2.3mm]
\newcommand{\gtgroup}[4]{\gategroup[#2,steps=#3,style={subcircstyle,fill=#4},background,label style={sublabstyle}]{#1}}
\newcommand{\grovergroup}[4]{\gtgroup{#1}{#2}{#3}{#4}}
\newcommand{\oraclecolor}[0]{blue!10}
\newcommand{\diffusercolor}[0]{red!10}

\newcommand{\rzz}[0]{R_{zz}(2\delta)}
\newcommand{\rz}[0]{R_{z}(2\delta)}
\newcommand{\rxx}[0]{R_{xx}(2\delta)}
\newcommand{\ryy}[0]{R_{yy}(2\delta)}

\newcommand{\Hheis}[0]{\hamilton_{\mathit{Heis}_n}}
\newcommand{\Uheis}[0]{U_{\mathit{Heis}_n}}

\newcommand{\hamilton}[0]{\mathcal{H}}

\newcommand{\smalltreeof}[3]{%
\scalebox{0.75}{%
\raisebox{3mm}{%
\tikz[level distance=2mm,sibling distance=10mm,baseline,inner sep=0.2mm]{%
\node {$#1$} child {node {$#2$}}%
child {node {$#3$}};}}}%
}
\newcommand{\bigtreeof}[6]{%
\scalebox{0.75}{%
\raisebox{5mm}{%
\tikz[level distance=2mm,
  level 1/.style={sibling distance=16mm},
  level 2/.style={sibling distance=10mm},
  baseline,inner sep=0.2mm]{%
\node {$#1$} child {node {$#2$}
                    child {node {$#3$}}
                    child {node {$#4$}}
                   }%
             child {node {$#2$}
                    child {node {$#5$}}
                    child {node {$#6$}}
                   };
}}}%
}

\newcommand{\trOrAppendix}[1]{\ifTR{}\cref{#1}\else{}\cite{techrep}\fi}

\newcommand{\compbgn}[0]{\mathtt{bgn}}
\newcommand{\compend}[0]{\mathtt{end}}

\newcommand{\specsubst}[2]{\lbag #1 \mapsto #2\rbag}
\newcommand{\lnref}[1]{Line~\ref{#1}}

%% file: tikzlibrarymyautomata.code.tex
\usetikzlibrary{automata}
\usetikzlibrary{arrows.meta}
\usetikzlibrary{bending}
\usetikzlibrary{shapes.callouts}
\usetikzlibrary{quotes}
\usetikzlibrary{positioning}
\usetikzlibrary{calc}
\usetikzlibrary{matrix}

\definecolor{spotblue}{RGB}{31,120,180}
\definecolor{spotpink}{RGB}{255,77,160}
\definecolor{spotorange}{RGB}{255,127,0}
\definecolor{spotpurple}{RGB}{106,61,154}
\definecolor{spotgreen}{RGB}{51,160,44}
\definecolor{spotred}{RGB}{227,26,28}
\definecolor{spotyellowish}{RGB}{196,196,0}
\definecolor{spotgray}{RGB}{80,80,80}
\definecolor{spotlight blue}{RGB}{107,246,255}
\definecolor{spotlight pink}{RGB}{255,154,255}
\definecolor{spotlight orange}{RGB}{255,156,103}
\definecolor{spotlight purple}{RGB}{178,164,255}
\definecolor{spotlight green}{RGB}{167,237,121}
\definecolor{spotlight red}{RGB}{255,104,104}
\definecolor{spotlight yellowish}{RGB}{255,224,64}
\definecolor{spotlight gray}{RGB}{192,192,144}

\tikzset{
  >={Stealth[round,bend]},
}

\tikzstyle{automaton}=[
  semithick,shorten >=0pt,>={Stealth[round,bend]},
  node distance=1.5cm,
  initial text=,
  every initial by arrow/.style={every node/.style={inner sep=0pt}},
  every state/.style={
    align=center,
    fill=white,
    minimum size=7.5mm,
    inner sep=0pt,
    execute at begin node=\strut,
  },%
]
\tikzset{
  scc/.style={draw=gray,fill=black!10,rounded corners=2mm},
  lstate/.style={state},
  }

\tikzstyle{smallautomaton}=[
  automaton,
  node distance=7mm,
  every state/.style={minimum size=4mm,fill=white,inner sep=1.5pt}
]

\tikzstyle{mediumautomaton}=[
  automaton,
  node distance=1cm,
  every state/.style={minimum size=6mm,fill=white,inner sep=2pt}
]

\tikzstyle{cstate}=[state,capsule,text width=,inner xsep=-5pt]
\tikzstyle{dot}=[fill=black,circle,minimum size=4pt,inner sep=0]

\makeatletter
\tikzoption{initial angle}{\tikzaddafternodepathoption{\def\tikz@initial@angle{#1}}}
\makeatother

\tikzstyle{initial overlay}=[every initial by arrow/.append style={overlay}]

\tikzstyle{unreachable} = [densely dotted]

\tikzstyle{acclabel} = [
  fill=yellow!20,
  rounded corners=3pt,
  draw=black,
  inner ysep=0pt,
  inner xsep=3pt,
]
\tikzstyle{ltllabel} = [acclabel,fill=darkgreen!20]

\tikzstyle{namelabel} = [
  execute at begin node = {$($},%
  execute at end node = {$)$}
]

\tikzstyle{matrix of states} = [
matrix of nodes,
every node/.style={state},
column sep=.8cm,
row sep=.8cm,
execute at empty cell={
  \node[transparent,name=
    \tikzmatrixname-\the\pgfmatrixcurrentrow-\the\pgfmatrixcurrentcolumn]{};]},
]

\tikzstyle{accset}=[
  circle,inner sep=.5pt,
  draw=none, 
  solid,
  collacc0, text=white,
  thin,
  anchor=center,
  minimum size=3.3mm,
  text width={}, 
  font=\bfseries\sffamily\footnotesize
]
\tikzstyle{accsquare}=[accset,rectangle,inner sep=1.9pt,rounded corners=0pt]

\tikzset{
  collacc0/.style={fill=spotorange},
  collacc1/.style={fill=spotpink},
  collacc2/.style={fill=spotblue},
  collacc3/.style={fill=spotpurple},
  collacc4/.style={fill=spotgreen},
  collacc5/.style={fill=spotred},
  collacc6/.style={fill=spotyellowish,draw=black,text=black},
  collacc7/.style={fill=spotgray},
  collacc8/.style={fill=spotlight blue,draw=black,text=black},
  collacc9/.style={fill=spotlight pink},
  collacc10/.style={fill=spotlight orange},
  collacc11/.style={fill=spotlight purple},
  collacc12/.style={fill=spotlight green},
  collacc13/.style={fill=spotlight red},
  collacc14/.style={fill=spotlight yellowish},
  collacc15/.style={fill=spotlight gray},
}

\pgfkeyssetvalue{/sacc/where}{center}
\tikzset{
  sacc where/.code={
    \pgfkeyssetvalue{/sacc/where}{#1}
  }
}

\tikzset{
  l/.pic={\node[outer sep=2pt] {#1};},%
  acc/.pic={\node[accset,collacc#1]{#1};},%
  accsq/.pic={\node[accsquare,collacc#1]{#1};},%
  eacc/.pic={\node[accset,collacc#1]{\emptyacc};},%
  eaccsq/.pic={\node[accsquare,collacc#1]{\emptyacc};},%
  sacc/.style = {
    append after command=
      {pic at (\tikzlastnode.\pgfkeysvalueof{/sacc/where}) {acc=#1}}
  },
  saccsq/.style = {
    append after command=
      {pic at (\tikzlastnode.\pgfkeysvalueof{/sacc/where}) {accsq=#1}}
  },
  esacc/.style = {
    append after command=
      {pic at (\tikzlastnode.\pgfkeysvalueof{/sacc/where}) {eacc=#1}}
  },
  esaccsq/.style = {
    append after command=
      {pic at (\tikzlastnode.\pgfkeysvalueof{/sacc/where}) {eaccsq=#1}}
  },
  pics/cacc/.style 2 args={%
    code={\node[accset,collacc#1]{#2};}%
  },%
  pics/caccsq/.style 2 args={%
      code={\node[accsquare,collacc#1]{#2};}%
    },%
  csacc/.style 2 args = {%
      append after command=%
        {pic at (\tikzlastnode.\pgfkeysvalueof{/sacc/where}) {cacc={#1}{#2}}}%
  },%
  csaccsq/.style 2 args = {%
      append after command=%
        {pic at (\tikzlastnode.\pgfkeysvalueof{/sacc/where}) {caccsq={#1}{#2}}}%
  },%
}

\def\markbaseline{-.33em}
\def\tacc#1{\tikz[baseline=\markbaseline]\pic{acc=#1};\xspace}
\def\tcacc#1#2{\tikz[baseline=\markbaseline]\pic{cacc={#1}{#2}};\xspace}

\def\emptyacc{\phantom{0}}

\tikzstyle{removed} = [opacity=0, overlay]
\tikzstyle{SCC} = [
  rounded corners,
  draw=black!50, very thin,
  fill=black!7
]
\tikzstyle{trivial} = [dashed]
\tikzstyle{sccname} = [red,anchor=south east,outer xsep=13pt]

%% file: Figures/all-basis-ta.tikz
 \begin{tikzpicture}[>=stealth',node distance=20mm]

  \pgfsetlinewidth{1bp}
  \tikzstyle{bddnode}=[draw,rectangle,rounded corners=2mm]
  \tikzstyle{bddleaf}=[]
  \tikzstyle{trans}=[->,>=stealth']
  \tikzstyle{translow}=[->,>=stealth',dashed]
  \tikzstyle{stick}=[-,>=stealth']
  \tikzstyle{hidtrans}=[]
  \tikzstyle{ark}=[]
  \tikzstyle{blueark}=[fill=black,opacity=0.2]
  \tikzstyle{redark}=[fill=red,opacity=0.6]

  \tikzstyle{outp}=[scale=0.75,fill=black!30,inner sep=0.6mm]

  \tikzstyle{bddnodex}=[bddnode,inner sep=1mm]


  \node[bddnodex] (p) {$q$};
  \node[left of=p,xshift=10mm] (root) {};
  \node[bddnodex,below left of=p,yshift=-5mm] (q1) {$q_1^1$};
  \node[bddnodex,below right of=p,yshift=-5mm] (q0) {$q_0^1$};
  \node[bddnodex,below of=q1] (r1) {$q_1^2$};
  \node[bddnodex,below of=q0] (r0) {$q_0^2$};
  \node[bddnodex,below of=r1] (s1) {$q_1^3$};
  \node[bddnodex,below of=r0] (s0) {$q_0^3$};
  \node[bddleaf, below of=s1,yshift=10mm] (c1) {$1$};
  \node[bddleaf, below of=s0,yshift=10mm] (c0) {$0$};

  \draw (p) coordinate[xshift=-5mm,yshift=-5mm] (pa);
  \draw (p) coordinate[xshift= 5mm,yshift=-5mm] (pb);

  \draw (q1) coordinate[xshift=0mm,yshift=-5mm] (q1a);
  \draw (q1) coordinate[xshift=5mm,yshift=-5mm] (q1b);

  \draw (r1) coordinate[xshift=0mm,yshift=-5mm] (r1a);
  \draw (r1) coordinate[xshift=5mm,yshift=-5mm] (r1b);
  
  \draw (q0) coordinate[xshift=0mm,yshift=-5mm] (q0a);
  \draw (r0) coordinate[xshift=0mm,yshift=-5mm] (r0a);

  \draw[translow] (pa)
    to[bend right]
    coordinate[pos=0.6] (pa_1)
    (q1);

  \draw[trans] (p) to 
    (pa)
    to[bend right]
    coordinate[pos=0.3] (pa_2)
    (q0);

  \filldraw[blueark] (pa) to[bend right=15] (pa_1) to[bend right=30] (pa_2) to[bend left=10] cycle;

  \draw[trans] (p) to 
    (pb)
    to[bend left]
    coordinate[pos=0.3] (pb_1)
    (q1);

  \draw[translow] (pb)
    to[bend left]
    coordinate[pos=0.6] (pb_2)
    (q0);

  \filldraw[blueark] (pb) to[bend left=10] (pb_1) to[bend
  right=30] (pb_2) to[bend right=15] cycle;

  \draw[translow] (q1a)
    to[bend right]
    coordinate[pos=0.5] (q1a_1)
    (r1);

  \draw[trans] (q1) to 
    (q1a)
    to[bend right]
    coordinate[pos=0.2] (q1a_2)
    (r0);

  \filldraw[blueark] (q1a) to[bend right=15] (q1a_1) to[bend
  right=30] (q1a_2) to[bend left=5] cycle;

  \draw[trans] (q1) to 
    (q1b)
    to[bend left]
    coordinate[pos=0.5] (q1b_1)
    (r1);

  \draw[translow] (q1b)
    to[bend left=5]
    coordinate[pos=0.3] (q1b_2)
    (r0);

  \filldraw[blueark] (q1b) to[bend left=15] (q1b_1) to[bend
  right=30] (q1b_2) to[bend right=5] cycle;
  
  \draw[trans] (q0) to 
    (q0a)
    to[bend left]
    coordinate[pos=0.6] (q0a_2)
    (r0);

  \draw[translow] (q0a)
    to[bend right]
    coordinate[pos=0.6] (q0a_1)
    (r0);

  \filldraw[blueark] (q0a) to[bend right=15] (q0a_1) to[bend
  right=30] (q0a_2) to[bend right=15] cycle;

  \draw[trans] (r1) to 
    (r1b)
    to[bend left]
    coordinate[pos=0.5] (r1b_1)
    (s1);

  \draw[translow] (r1b)
    to[bend left=5]
    coordinate[pos=0.3] (r1b_2)
    (s0);

  \filldraw[blueark] (r1b) to[bend left=15] (r1b_1) to[bend
  right=30] (r1b_2) to[bend right=5] cycle;

  \draw[translow] (r1a)
    to[bend right]
    coordinate[pos=0.5] (r1a_1)
    (s1);

  \draw[trans] (r1) to 
    (r1a)
    to[bend right]
    coordinate[pos=0.2] (r1a_2)
    (s0);

  \filldraw[blueark] (r1a) to[bend right=15] (r1a_1) to[bend
  right=30] (r1a_2) to[bend left=5] cycle;

  \draw[trans] (r0) to 
    (r0a)
    to[bend left]
    coordinate[pos=0.6] (r0a_2)
    (s0);

  \draw[translow] (r0a)
    to[bend right]
    coordinate[pos=0.6] (r0a_1)
    (s0);

  \filldraw[blueark] (r0a) to[bend right=15] (r0a_1) to[bend
  right=30] (r0a_2) to[bend right=15] cycle;

  \draw[trans] (root) to (p);
  \draw[->] (s1) to
    (c1);
  \draw[->] (s0) to
    (c0);
\end{tikzpicture}

%% file: Figures/lsta-ex.tikz
 \begin{tikzpicture}[>=stealth',node distance=20mm]

  \pgfsetlinewidth{1bp}
  \tikzstyle{bddnode}=[draw,rectangle,rounded corners=2mm]
  \tikzstyle{bddleaf}=[]
  \tikzstyle{trans}=[->,>=stealth']
  \tikzstyle{translow}=[->,>=stealth',dashed]
  \tikzstyle{stick}=[-,>=stealth']
  \tikzstyle{hidtrans}=[]
  \tikzstyle{ark}=[]
  \tikzstyle{greyark}=[fill=black,opacity=0.2]
  \tikzstyle{blueark}=[fill=blue,opacity=0.2,draw=blue!60]
  \tikzstyle{redark}=[fill=red,opacity=0.2,draw=red!60]

  \tikzstyle{outp}=[scale=0.75,fill=black!30,inner sep=0.6mm]

  \tikzstyle{bddnodex}=[bddnode,inner sep=1mm]


  \node[bddnodex,below right of=p,yshift=-5mm] (q0) {$q_0$};
  \node[left of=q0,xshift=10mm] (root) {};
  \node[bddnodex,below of=q0,yshift=4mm] (r1) {$q_1$};
  \node[bddnodex,below left of=r1,xshift=6mm] (s0) {$q_2$};
  \node[bddnodex,below right of=r1,xshift=-6mm] (s1) {$q_3$};
  \node[bddleaf, below of=s1,yshift=10mm] (c1) {$\invsqrttwo$};
  \node[bddleaf, below of=s0,yshift=10mm] (c0) {$0$};

  \node[bddnodex,right of=r1,xshift=-4mm] (u1) {$q_4$};
  \node[bddleaf, below of=u1,yshift=10mm] (c2) {$\invtwo$};

  \draw (r1) coordinate[xshift=0mm,yshift=-5mm] (r1a);
  \draw (r1) coordinate[xshift=5mm,yshift=0mm] (r1b);

  \draw (q0) coordinate[xshift=0mm,yshift=-5mm] (q0a);

  \draw[trans] (q0) to 
    (q0a)
    to[bend left]
    coordinate[pos=0.6] (q0a_2)
    (r1);

  \draw[translow] (q0a)
    to[bend right]
    coordinate[pos=0.6] (q0a_1)
    (r1);

  \filldraw[greyark] (q0a) to[bend right=15] (q0a_1) to[bend
  right=30] (q0a_2) to[bend right=15] cycle;

  \draw[trans]  (r1) to (r1a)
    to
    coordinate[pos=0.5] (r1a_1)
    (s1);

  \draw[translow]
    (r1a)
    to
    coordinate[pos=0.5] (r1a_2)
    (s0);

  \filldraw[blueark] (r1a) to (r1a_2) to[bend
  right=30] node[opacity=1] {\tacc 2} (r1a_1) to cycle;

  \draw[trans] (r1) to 
    (r1b)
    to[bend right]
    coordinate[pos=0.6] (r1b_2)
    (u1);

  \draw[translow] (r1b)
    to[bend left]
    coordinate[pos=0.6] (r1b_1)
    (u1);

  \filldraw[redark] (r1b) to[bend right=15] (r1b_2) to[bend
  right=30] node[opacity=1] {\tacc 1} (r1b_1) to[bend right=15] cycle;

  \draw[trans] (root) to (q0);
  \draw[->] (s1) to
    (c1);
  \draw[->] (s0) to
    (c0);
  \draw[->] (u1) to
    (c2);
\end{tikzpicture}

%% file: Figures/bv-circuit.tikz
\begin{quantikz}[row sep=0.4cm]
  \lstick{$w_1 = \ket 0$}                & \gate{H} & \ctrl{7} \grovergroup{Oracle} 8 5 {purple!10}& \qw & \qw     &\
  \cdots\ & \qw              & \gate{H} & \qw \\
  \lstick{$w_2 = \ket 0$}                & \gate{H} & \qw      & \qw & \qw     &\
  \cdots\ & \qw              & \gate{H} & \qw\\
  \lstick{$w_3 = \ket 0$}                & \gate{H} & \qw      & \qw & \ctrl{5}&\
  \cdots\ & \qw                   & \gate{H} & \qw \\
  \setwiretype{n}                   &          &          &               &    &          &          \\[-1.3em]
  \hspace*{-5mm}\vdots\hspace*{5mm}\setwiretype{n} &  \vdots  &    &  \vdots & &  \vdots&  \vdots&  \vdots           \\[-.8em]
  \setwiretype{n}                   &          &          &                       &          \\[-0.5em]
  \lstick{$w_n = \ket 0$}                & \gate{H} & \qw & \qw& \qw           &\
  \cdots\ & \ctrl{1}  &\gate{H} & \qw \\
  \lstick{$a = \ket 1$}                & \gate{H} & \targ{} & \qw& \targ{}   &\
  \cdots\ &  \targ{}          &\gate{H} &\qw
\end{quantikz}

%% file: Figures/adder-circuit.tikz
\begin{quantikz}
	\lstick{$\ket{c_1}$}         & \gate[3]{\maj}& \qw   & \qw      & \qw            &\qw  & \qw           & \qw          & \qw      & \gate[3]{\uma} & \rstick{$\ket{0}$}\\
	\lstick{$\ket{b_1}$}       & \qw           & \qw           & \qw      & \qw            &\qw        & \qw           & \qw           & \qw      & \qw           &\rstick{$\ket{a_1 \oplus b_1 \oplus c_1}$}\\
	\lstick{$\ket{a_1}$}       & \qw           & \gate[3]{\maj} & \qw      & \qw            &\qw        & \qw           & \qw           & \gate[3]{\uma}& \qw       & \rstick{$\ket{a_1}$}\\
	\lstick{$\ket{b_2}$}       & \qw           & \qw           & \qw      & \qw            &\qw        & \qw           & \qw           & \qw      & \qw           &\rstick{$\ket{a_2 \oplus b_2 \oplus c_2}$}\\
	\lstick{$\ket{a_2}$}       & \qw           & \qw           & \qw      & \qw            &\qw        & \qw           & \qw           & \qw      & \qw           &\rstick{$\ket{a_2}$}\\
	  \setwiretype{n}            & \vdots        & \vdots        & \vdots   & \vdots         & \vdots    &  \vdots       & \vdots      & \vdots   & \vdots        &\\
	\lstick{$\ket{a_{n-1}}$}   & \qw           & \qw           & \qw      & \gate[3]{\maj}  &\qw        & \gate[3]{\uma} & \qw           & \qw      & \qw           &\rstick{$\ket{a_{n-1}}$}\\
	\lstick{$\ket{b_n}$}       & \qw           & \qw           & \qw      & \qw            &\qw        & \qw           & \qw           & \qw      & \qw           & \rstick{$\ket{a_n \oplus b_n \oplus c_n}$}\\
	\lstick{$\ket{a_n}$}       & \qw           & \qw           & \qw      & \qw            &\ctrl{1}   & \qw           & \qw           & \qw      & \qw           & \rstick{$\ket{a_n}$}\\
	\lstick{$\ket{0}$}         & \qw           & \qw           & \qw      & \qw            &\targ{}   & \qw           & \qw           & \qw      & \qw           &\rstick{$\ket{c_{n+1}}$}\\
\end{quantikz}

%% file: Figures/qecc-circuit.tikz
\begin{quantikz}
  \lstick{$x_1$} & \ctrl{3} \gategroup[wires=4,steps=2,style={dotted,cap=round,inner sep=7pt},label style={label position=above}]{$B$} & \qw & \qw & \qw & \qw & \qw & \qw & \qw & \qw \\
  \lstick{$a_1$} & \qw & \qw & \qw & \qw & \qw & \qw & \qw & \qw  & \qw\\
  \lstick{$x_2$} & \qw & \ctrl{1} & \qw & \ctrl{3} \gategroup[wires=4,steps=2,style={dotted,cap=round,inner sep=7pt},label style={label position=above,yshift=1mm}]{$B$} & \qw & \qw & \qw & \qw  & \qw\\
  \lstick{$a_2$} & \targ{} & \targ{} & \qw & \qw & \qw & \qw & \qw & \qw  & \qw\\
  \lstick{$x_3$} & \qw & \qw & \qw & \qw & \ctrl{1} & \qw & \ctrl{3} \gategroup[wires=4,steps=2,style={dotted,cap=round,inner sep=7pt},label style={label position=above,yshift=-1mm}]{$B$} & \qw  & \qw\\
  \lstick{$a_3$} & \qw & \qw & \qw & \targ{} & \targ{} & \qw & \qw & \qw  & \qw\\
  \lstick{$x_4$} & \qw & \qw & \qw & \qw & \qw & \qw & \qw & \ctrl{1}  & \qw\\
  \lstick{$a_4$} & \qw & \qw & \qw & \qw & \qw & \qw & \targ{} & \targ{} & \qw
\end{quantikz}

%% file: Figures/grover-circuit1.tikz
\begin{quantikz}[row sep=0.4cm]
  \lstick{$q_1$} & \gate{X}\grovergroup{Oracle}{10}{3}{\oraclecolor} & \ctrl{1} & \gate{X}  & \gate{H} \grovergroup{Diffuser}{10}{5}{\diffusercolor}& \gate{X} & \ctrl{1}  & \gate{X} & \gate{H}  & \qw\\
  \lstick{$q_2$}                & \gate{X} & \ctrl{2} & \gate{X}                   & \gate{H} & \gate{X} & \ctrl{2}  & \gate{X} & \gate{H}  & \qw\\
  \lstick{$\ket{0}$}                & \qw & \qw & \qw  & \qw & \qw & \qw  & \qw & \qw & \qw \\
  \lstick{$q_3$}                & \gate{X} & \ctrl{2} & \gate{X}                   & \gate{H} & \gate{X} & \ctrl{2}  & \gate{X} & \gate{H} & \qw \\
  \lstick{$\ket{0}$}                & \qw & \qw & \qw  & \qw & \qw & \qw  & \qw & \qw & \qw \\
  \setwiretype{n}                     &          &          &                             &          &          &           &          &          \\[-1.3em]
  \hspace*{-5mm}\vdots\hspace*{5mm}\setwiretype{n} & \vdots   & \vdots   &  \vdots                     & \vdots   & \vdots   & \vdots    & \vdots   &  \vdots  \\[-.8em]
  \setwiretype{n}                     &          & \vqw{1} &                             &          &          &           &          &          \\[-0.5em]
  \lstick{$q_n$}          & \gate{X} & \ctrl{1} & \gate{X}                    & \gate{H} & \gate{X} & \ctrl{-1}  & \gate{X} & \gate{H} & \qw \\
  \lstick{$\ket{0}$}                  & \qw      & \ctrl{-1} & \qw                         & \qw     & \qw      & \qw       & \qw      & \qw & \qw 
\end{quantikz}

%% file: Figures/grover-circuit2.tikz
\begin{quantikz}[row sep=0.4cm]
  \lstick{$q_1$}  & \gate{X}\grovergroup{Oracle}{10}{9}{\oraclecolor} &
  \ctrl{1} & \qw & \qw & \qw & \qw & \qw & \ctrl{1}& \gate{X} & \gate{H}\grovergroup{Diffuser}{10}{11}{\diffusercolor} & \gate{X} & \ctrl{1} & \qw & \qw & \qw & \qw & \qw & \ctrl{1}& \gate{X} & \gate{H} & \qw\\
  \lstick{$q_2$}                & \gate{X} & \ctrl{1} & \qw & \qw & \qw & \qw & \qw & \ctrl{1}& \gate{X} & \gate{H} & \gate{X} & \ctrl{1} & \qw & \qw & \qw & \qw & \qw & \ctrl{1}& \gate{X} & \gate{H} & \qw\\
  \lstick{$\ket{0}$}                  & \qw      & \targ{}  & \ctrl{1} & \qw & \qw & \qw& \ctrl{1}& \targ{}      & \qw & \qw & \qw  & \targ{}  & \ctrl{1} & \qw & \qw & \qw& \ctrl{1}& \targ{}      & \qw & \qw & \qw\\
  \lstick{$q_3$}                & \gate{X} & \qw      &\ctrl{1} & \qw & \qw & \qw& \ctrl{1} & \qw & \gate{X}                   & \gate{H} & \gate{X} &\qw &\ctrl{1} & \qw & \qw & \qw& \ctrl{1} & \qw&  \gate{X} & \gate{H} & \qw\\
  \lstick{$\ket{0}$}                  & \qw      & \qw      & \targ{}  & \ctrl{1} & \qw & \ctrl{1}& \targ{}      & \qw & \qw & \qw  & \qw & \qw  & \targ{}  & \ctrl{1} & \qw & \ctrl{1}& \targ{} & \qw & \qw & \qw & \qw    \\
  \setwiretype{n}                     &          &          &  &  &  & &                             &          &          &           &          &      &  &  &  & &     \\[-1.3em]
  \hspace*{-5mm}\vdots\hspace*{5mm}\setwiretype{n} & \vdots   & \vdots   &  \vdots                     & \vdots   & \vdots   & \vdots    & \vdots   &  \vdots& \vdots   &  \vdots& \vdots   &  \vdots& \vdots   &  \vdots& \vdots   &  \vdots& \vdots   &  \vdots& \vdots   &  \vdots  \\[-.8em]
  \setwiretype{n}                     &          &  &                             &  &          &  &                             &          &          &           &          & &          &           &          &          \\[-0.5em]
  \lstick{$q_n$}          & \gate{X} & \qw & \qw & \qw & \ctrl{-1} & \qw & \qw & \qw& \gate{X}                    & \gate{H} & \gate{X} & \qw & \qw & \qw & \ctrl{-1}  & \qw & \qw & \qw & \gate{X} & \gate{H} & \qw\\
  \lstick{$\ket{0}$}                  & \qw      & \qw & \qw & \qw & \ctrl{-1} & \qw & \qw & \qw                         & \qw     & \qw      & \qw       & \qw & \qw & \qw      & \qw       & \qw      & \qw & \qw      & \qw        & \qw & \qw
\end{quantikz}

%% file: Figures/unitary-heis-circuit1.tikz
\begin{quantikz}[row sep=.1cm]
    \lstick{} & \gate[2]{\rxx}\gategroup[wires=6,steps=13,
     style={dashed,draw=red},
     label style={label position=above,yshift=0.3em}]{repeat $r$ times}
    & \qw &\qw &\qw 
    & \gate[2]{\ryy}& \qw &\qw &\qw 
    & \gate[2]{\rzz}& \qw &\qw &\qw &\qw 
    \\
    \lstick{} & \qw                  & \gate[2]{\rxx} & \qw &\qw 
    & \qw & \gate[2]{\ryy } & \qw &\qw 
    & \qw & \gate[2]{\rzz } & \qw &\qw &\qw
    \\
    \lstick{} & \qw                   & \qw &\qw &\qw 
    & \qw & \qw &\qw &\qw 
    & \qw & \qw &\qw &\qw&\qw 
    \\
    \setwiretype{n}    & \vdots             & \vdots &\vdots &\vdots 
    & \vdots             & \vdots &\vdots &\vdots 
    & \vdots             & \vdots &\vdots &\vdots &\vdots  
    \\
    \lstick{} & \qw                  &\qw   & \qw &\gate[2]{\rxx } 
    & \qw                  &\qw   & \qw &\gate[2]{\ryy } 
    & \qw                  &\qw   & \qw &\gate[2]{\rzz }&\qw   
    \\
    \lstick{} & \qw                   & \qw &\qw &\qw 
    & \qw                   & \qw &\qw &\qw 
    & \qw                   & \qw &\qw &\qw&\qw 
\end{quantikz}

%% file: Figures/unitary-heis-circuit2.tikz
\begin{quantikz}[row sep=.1cm]
\lstick{} &\gate{H}
\gategroup[wires=6,steps=15,
style={dashed,draw=red},
label style={label position=above,yshift=0.3em}]{repeat $r$ times}
&\gate[2]{\rzz}&&
&\gate{\sqrt{X}}
&\gate[2]{U_{zz}(2\delta)}&&&
&\gate{H}&\gate{S}
&\gate[2]{\rzz}&&&\qw
\\
\lstick{} &\gate{H} 
&&\gate[2]{\rzz}&
&\gate{\sqrt{X}}
&&\gate[2]{U_{zz}(2\delta)}&&
&\gate{H}&\gate{S}
&&\gate[2]{\rzz}&&\qw
\\
\lstick{} &\gate{H} 
&&&
&\gate{\sqrt{X}}
&&&&
&\gate{H}&\gate{S}
&&&&\qw
\\
\setwiretype{n}   &\vdots    
&\vdots&\vdots&\vdots 
&\vdots
&\vdots&\vdots&\vdots &\vdots
&\vdots &\vdots
&\vdots&\vdots&\vdots &\vdots
\\
\lstick{} &\gate{H} 
&&&\gate[2]{\rzz}
&\gate{\sqrt{X}}
&&&\gate[2]{U_{zz}(2\delta)} &
&\gate{H}&\gate{S}
&&&\gate[2]{\rzz}&\qw
\\
\lstick{} &\gate{H} 
&&&
&\gate{\sqrt{X}}
&&&&\gate{X}
&\gate{H}&\gate{S}
&&&&\qw
\end{quantikz}

%% file: Figures/qft-circuit.tikz
{
\newcommand{\makedots}[0]{\lstick[label style={xshift=-1mm,yshift=-8mm}]{\vdots}}
\newcommand{\makedotsb}[0]{\lstick[label style={xshift=0mm,yshift=-5mm}]{\vdots}}
\newcommand{\qftgroup}[4]{\gtgroup{$\mathit{QFT}_{#1}$}{#2}{#3}{#4}}
\begin{quantikz}
  \lstick{$x_1$}          & \gate{H}\qftgroup n 5 {15} {blue!10}& \ctrl{1}   & \ctrl{2}   & \ctrl{3}       & \ctrl{4}   & \qw     & \qw        & \qw            & \qw            & \qw      \qw & \qw            & \qw            & \qw     & \qw        & \qw     & \swap{4} \gtgroup {reverse} 5 4 {red!20} & \qw & \qw      & \qw  &\qw\\
  \lstick{$x_2$}          & \qw      & \gate{R_2} & \qw        & \qw            & \qw        & \gate H \qftgroup{n-1} 4 {10}{blue!20} & \ctrl{1}   & \ctrl{2}       & \ctrl{3}       & \qw          & \qw            & \qw            & \qw     & \qw        & \qw     & \qw      & \swap{2} & \qw & \qw &\qw \\
  \lstick{$x_3$}\makedots & \qw      & \qw        & \gate{R_3} & \qw            & \qw        & \qw     & \gate{R_2} & \qw            &            & \gate{H}\qftgroup{n-2} 3 6 {blue!30}    & \ctrl 1        & \ctrl 2        & \qw     & \makedotsb         &  \qw & \qw      & \qw      & \qw & \qw\makedots & \qw\\[10mm]
  \lstick{$x_{n-1}$}      & \qw      & \qw        & \qw        & \gate{R_{n-1}} & \qw        & \qw     & \qw        & \gate{R_{n-2}} & \qw            & \qw          & \gate{R_{n-3}} & \qw            & \gate H \qftgroup 2 2 3 {blue!40} & \ctrl 1    & \qw     & \qw      & \targX{} & \qw & \qw & \qw\\
  \lstick{$x_n$}          & \qw      & \qw        & \qw        & \qw            & \gate{R_n} & \qw     & \qw        & \qw            & \gate{R_{n-1}} & \qw          & \qw            & \gate{R_{n-2}} & \qw     & \gate{R_2} & \gate H \qftgroup 1 1 1 {blue!50}& \targX{} & \qw      & \qw & \qw & \qw
\end{quantikz}
}

%% file: appendix.tex
\section{Weighted tree transducer of QFT}\label{sec:QFT}

The transducer of $Box_i=H_iCR^i_2CR^i_3\ldots CR^i_{n-i+1}$ consist of the following transitions, where $q_1^i$ is the root state and $q_{i+1},id_{i+1}$ are the leaf states. We set $m=2^n$.
\begin{align}
\trntransof {q^i_a} {x_a} {q^i_{a+1} (\leftT) }{q^i_{a+1} (\rightT)},& &a\in [1,i-1] \label{ln:qft1}\\
\trntransof {id_a} {x_a} {id_{a+1} (\leftT) }{id_{a+1} (\rightT)},& &a\in [1,n] \label{ln:qft2}\\
\trntransof {q^i_i} {x_i} {\invsqrttwo id_{i+1} (\leftT) + \invsqrttwo id_{i+1} (\rightT)}{\invsqrttwo q^i_{i+1} (\leftT) -  \invsqrttwo q^i_{i+1} (\rightT)} \label{ln:qft3}\\
\trntransof {q^i_a} {x_a} {id_{a+1} (\leftT) }{\gamma_{a-i+1}q^i_{a+1} (\rightT)},& &a\in [i+1,n] \label{ln:qft4}
\end{align}

Here, the transitions \cref{ln:qft1} and \cref{ln:qft2} indicate that the left and right subtrees remain unchanged. However, the state $q^i_a$ additionally encodes that we are at position~$a$, and that gate applications should begin once position~$i$ is reached.
The transition \cref{ln:qft3} is obtained by the composition of the transition $\trntransof {q^i_i} {x_i} {id_{i+1} (\leftT) + id_{i+1} (\rightT)}{id_{i+1} (\leftT) - id_{i+1} (\rightT)}$ of the $H_i$ gate transducer to the corresponding transition from the transducer of the $CR^i_2CR^i_3\ldots CR^i_{n-i+1}$ gates, which starts a sequence of rotations $R_2R_3\ldots R_{n-i+1}$ on the consequent right substrees, so we update the $id_{i+1}$ on the right to $q^i_{i+1}$. For positions $a>i$, the state $q^i_a$ performs the gate $R_k=R_{a-i+1}$, and hence multiplies the right subtree with $\gamma_k=\omega^{2^{n-k}}=\omega^{2^{n-(a-i+1)}}=\gamma_{a-i+1}$. The transitions \cref{ln:qft4} perform the corresponding rotation on the right subtrees.

If we look at the transducer of $Box_1$
\begin{align*}
\trntransof {id_a} {x_a} {id_{a+1} (\leftT) }{id_{a+1} (\rightT)},& &a\in [1,n]\\
\trntransof {q^1_1} {x_1} {\invsqrttwo id_{2} (\leftT) + \invsqrttwo id_{2} (\rightT)}{\invsqrttwo q^1_{2} (\leftT) -  \invsqrttwo q^1_{2} (\rightT)}\\
\trntransof {q^1_a} {x_a} {id_{a+1} (\leftT) }{\gamma_a {q^1_{a+1}}(\rightT)},& &a\in [2,n]
\end{align*}
and $Box_2$
\begin{align*}
\trntransof {q^2_1} {x_1} {q^2_{2} (\leftT) }{q^2_{2} (\rightT)}& &\\
\trntransof {id_a} {x_a} {id_{a+1} (\leftT) }{id_{a+1} (\rightT)},& &a\in [1,n] \\
\trntransof {q^2_2} {x_2} {\invsqrttwo id_{3} (\leftT) + \invsqrttwo id_{3} (\rightT)}{\invsqrttwo q^2_{3} (\leftT) -  \invsqrttwo q^2_{3} (\rightT)}\\
\trntransof {q^2_a} {x_a} {id_{a+1} (\leftT) }{\gamma_{a-1}q^2_{a+1} (\rightT)},& &a\in [3,n]
\end{align*}
Their composition $Box_2\circ Box_1$ is the following. Below we use $p^2_a$ to replace the pair $(q^2_a,q^1_a)$ state, $q^1_a$ to replace $(id_a,q^1_a)$, $q^2_a$ to replace $(q^1_a,id_a)$, and $id_a$ to replace $(id_a,id_a)$.
\begin{align*}
\trntransof {p^2_1} {x_1} {\invsqrttwo q^2_{2} (\leftT)  + \invsqrttwo q^2_{2} (\rightT)}{\invsqrttwo p^2_2 (\leftT)  - \invsqrttwo p^2_2 (\rightT)}& &\\
\trntransof {p^2_2} {x_2} {\invsqrttwo id_{3} (\leftT)  + \invsqrttwo \gamma_2 q^1_3 (\rightT)}{\invsqrttwo q^2_3 (\leftT)  - \invsqrttwo \gamma_2 p^2_3 (\rightT)}& &\\
\trntransof {p^2_a} {x_a} {id_{a+1} (\leftT)}{\gamma_a\gamma_{a-1} p^2_{a+1} (\rightT)}, & &a \in [3,n]\\ 
\hline
\trntransof {id_a} {x_a} {id_{a+1}(\leftT) }{id_{a+1}(\rightT)},& &a\in [1,n]\\
\hline
\trntransof {q^i_a} {x_a} {q^i_{a+1} (\leftT) }{q^i_{a+1} (\rightT)},& &i\in[2,2],  a\in [1,i-1] \\
\trntransof {q^i_i} {x_i} {\invsqrttwo id_{i+1} (\leftT) + \invsqrttwo id_{i+1} (\rightT)}{\invsqrttwo q^i_{i+1} (\leftT) -  \invsqrttwo q^i_{i+1} (\rightT)},& &i\in[1,2]\\
\trntransof {q^i_a} {x_a} {id_{a+1} (\leftT) }{\gamma_{a-i+1}q^i_{a+1} (\rightT)},& &i\in[1,2],a\in [i+1,n] 
\end{align*}
Notice that in the composed transducer, we have two types of transitions. (1) Those start from $p^2_a$; we have $n$ such states and transitions. (2) Those start from $id_a$, $q^1_a$, or $q^2_a$; we also have $3n$ such states and transitions. Observe that those transitions are exactly those from $Box_1$ and $Box_2$.

The transducer of the entire $QFT_n$ is the composition of $Box_n\circ Box_{n-1}\circ \ldots\circ Box_1$, which follows the same pattern. It totally has $n^2$ both types of states and transitions. Below is the complete transducer transitions.

\begin{align*}
\trntransof {p^2_1} {x_1} {\invsqrttwo q^2_{2} (\leftT)  + \invsqrttwo q^2_{2} (\rightT)}{\invsqrttwo p^2_2 (\leftT)  - \invsqrttwo p^2_2 (\rightT)}& &\\
\trntransof {p^2_2} {x_2} {\invsqrttwo id_{3} (\leftT)  + \invsqrttwo \gamma_2 q^1_3 (\rightT)}{\invsqrttwo q^2_3 (\leftT)  - \invsqrttwo \gamma_2 p^2_3 (\rightT)}& &\\
\trntransof {p^3_3} {x_3} {\invsqrttwo id_{4} (\leftT)  + \invsqrttwo \gamma_3 \gamma_2 q^2_4 (\rightT)}{\invsqrttwo q^3_4 (\leftT)  - \invsqrttwo \gamma_3\gamma_2 p^3_4 (\rightT)}& &\\
\ldots\\
\trntransof {p^n_n} {x_n} {\invsqrttwo id_{n+1} (\leftT)  + \invsqrttwo \gamma_n\ldots\gamma_3 \gamma_2 q^{n-1}_{n+1} (\rightT)}{\invsqrttwo q^n_{n+1} (\leftT)  - \invsqrttwo \gamma_n\ldots\gamma_3\gamma_2 p^n_{n+1} (\rightT)}& &\\
\trntransof {p^2_a} {x_a} {id_{a+1} (\leftT)}{\gamma_a\gamma_{a-1} p^2_{a+1} (\rightT)}, & &a \in [3,n]\\ 
\trntransof {p^3_a} {x_a} {id_{a+1} (\leftT)}{\gamma_a\gamma_{a-1}\gamma_{a-2} p^3_{a+1} (\rightT)}, & &a \in [4,n]\\ 
\ldots\\
\trntransof {p^{n-1}_a} {x_a} {id_{a+1} (\leftT)}{\gamma_a\gamma_{a-1}\gamma_{a-2}\ldots \gamma_{a-n+2} p^{n-1}_{a+1} (\rightT)}, & &a \in [n,n]\\ 
\hline
\trntransof {id_a} {x_a} {id_{a+1}(\leftT) }{id_{a+1}(\rightT)},& &a\in [1,n]\\
\hline
\trntransof {q^i_a} {x_a} {q^i_{a+1} (\leftT) }{q^i_{a+1} (\rightT)},& &i\in[2,n],  a\in [1,i-1] \\
\trntransof {q^i_i} {x_i} {\invsqrttwo id_{i+1} (\leftT) + \invsqrttwo id_{i+1} (\rightT)}{\invsqrttwo q^i_{i+1} (\leftT) -  \invsqrttwo q^i_{i+1} (\rightT)},& &i\in[1,n]\\
\trntransof {q^i_a} {x_a} {id_{a+1} (\leftT) }{\gamma_{a-i+1}q^i_{a+1} (\rightT)},& &i\in[1,n-1],a\in [i+1,n] 
\end{align*}

\begin{theorem}\label{thm:qft_WTT}
The transducer for $QFT_{[j+1\ldots j+n]}$ can be constructed using $O(n^2+m)$ states and $O(n^2+m)$ transitions.
\end{theorem}

\vspace{-0.0mm}
\section{Transducer Image and Composition Examples}\label{sec:image_compos_examples}
\vspace{-0.0mm}

\begin{example}\label{ex:image}
Let $\autbases$ be the SWTA defined as
\begin{equation}
\begin{aligned}
  \initmark & \transcof{q} a 1 {q}{0q} &\qquad
  \initmark & \transcof{q} a 2 {0q}{q}
\end{aligned}
\end{equation}
where $q$ is a~leaf state.
It is easy to see that~$\autbases$ accepts all trees where all internal nodes
are labelled with~$a$ and exactly one leaf node is labelled with~1 (and the
other leaf nodes are labelled with~0)---these correspond to trees representing
all computational basis states in a~quantum circuit for any number of qubits.
Next, let $\trn_{\hadtensor}$ be the WTT defined as
\begin{equation}
  \initmark \trntransof s a
  {\invsqrttwo s(\leftT) + \invsqrttwo s(\rightT)}
  {\invsqrttwo s(\leftT) - \invsqrttwo s(\rightT)}
\end{equation}
where~$s$ is a~leaf state.
$\trn_{\hadtensor}$ represents a~quantum gate applying the~$\had$ (Hadamard)
gate to all qubits in the circuit.
If we compute $\trn_{\hadtensor}(\autbases)$, we obtain the SWTA~$\aut'$ given as
\begin{equation}
\begin{aligned}
  \initmark & \transcof{\tuple{q, s}} a 1 {\invsqrttwo\tuple{q,s}}{\invsqrttwo\tuple{q,s}} &\qquad
  \initmark & \transcof{\tuple{q, s}} a 2 {\invsqrttwo\tuple{q,s}}{-\invsqrttwo\tuple{q,s}} &\qquad
\end{aligned}
\end{equation}
with $\tuple{q,s}$ being a~leaf state.
The language of~$\aut'$ contains trees such as
1 (a single leaf node),
\smalltreeof a {\invsqrttwo}{\invsqrttwo},
\smalltreeof a {\invsqrttwo}{\text{-}\invsqrttwo}, or
\bigtreeof a a {\invtwo}{\invtwo}{\text{-}\invtwo}{\text{-}\invtwo}.

Let us now apply~$\trn_{\hadtensor}$ on~$\aut'$ again.
We obtain the SWTA~$\aut'' =\trn_{\hadtensor}(\aut')$ given as
\begin{equation}
\begin{aligned}
  \initmark & \transcof{\tuple{\tuple{q,s},s}} a 1 {\tuple{\tuple{q,s},s}}{0\tuple{\tuple{q,s},s}} &\qquad
  \initmark & \transcof{\tuple{\tuple{q,s},s}} a 2 {0\tuple{\tuple{q,s},s}}{\tuple{\tuple{q,s},s}}
\end{aligned}
\end{equation}
with the leaf state~$\tuple{\tuple{q,s},s}$.  We note that $\semof{\autbases} =
\semof{\aut''}$, which is the expected behavior of applying the $\had$ gate
twice.
\qed
\end{example}

\begin{example}\label{ex:compose}
  Consider the WTT $\trn_{\hadtensor}$ from \cref{ex:image} and
  $\trn_{\rotx}^{\mathit{even}}$ be the WTT implementing the $\rotx$ rotation
  (with the unitary $\invsqrttwo\left[\begin{smallmatrix}1 & -i \\ -i & 1\end{smallmatrix}\right]$)
  on every other qubit defined as follows:
  \begin{equation}
  \begin{aligned}
  \initmark & \trntransof u a {\invsqrttwo v(\leftT) - i \invsqrttwo v(\rightT)}
                              {-i \invsqrttwo v(\leftT)+ \invsqrttwo v(\rightT)}
            \qquad
            &
              \trntransof v a u u
  \end{aligned}
  \end{equation}
  where only~$u$ is a~leaf state (this restricts the input
  of~$\trn_{\rotx}^{\mathit{even}}$ to only trees of even height).
  To compute the composition $\trn_{\rotx}^{\mathit{even}} \compose
  \trn_{\hadtensor}$ (i.e., first, we apply the~$\had$ gate on all qubits and
  then the~$\rotx$ gate on every other qubit), we start with composing the root
  transitions to obtain the transition
  \begin{equation}
  \begin{aligned}
  \initmark&
    \trntransof{\tuple{s,u}} a {\smallfrac {1-i} 2 \tuple{s,v}(\leftT) +
    \smallfrac{1+i} 2 \tuple{s,v}(\rightT)}{\smallfrac {1-i} 2
    \tuple{s,v}(\leftT) - \smallfrac{1+i} 2 \tuple{s,v}(\rightT)}
  \end{aligned}
  \end{equation}
  (we invite the reader to perform the full computation with pen and paper).
  Computing the transition from~$\tuple{s,v}$ is easier:
  \begin{equation}
    \trntransof{\tuple{s,v}} a
    {\invsqrttwo \tuple{s,u}(\leftT) + \invsqrttwo \tuple{s,u}(\rightT)}
    {\invsqrttwo \tuple{s,u}(\leftT) - \invsqrttwo \tuple{s,u}(\rightT)} .
  \end{equation}
  The set of root states is~$\{\tuple{s,u}\}$.
\qed
\end{example}

\vspace{-0.0mm}
\section{Example of BV}\label{appendix:bv}
\vspace{-0.0mm}

We illustrate modeling of the size-parametrized BV circuit (\cref{fig:BV})
with a secret key of the form $(\texttt{10})^*
(\texttt{1} + \varepsilon) = \{\varepsilon, \texttt{1}, \texttt{10},
\texttt{101}, \dots \}$.
We start with $\aut_\pre$ representing the precondition
and the transducer $\transfH$ as given in \cref{sec:BV},
and show how $\transfH(\aut_\pre)$ is constructed
using the composition described in~\cref{sec:trn_compose}.

Looking at the initial states of $\aut_\pre$
and $\transfH$, we start with the (initial) state $\tuple{s_1, u}$ and
explore its successors by taking transitions with the same symbols. First,
both $s_1$ and $u$ can make a transition labeled with $w$. Taking
the transition over $w$ takes $s_1$ to $(s_1, 0s_1)$. Simultaneously,
$\transf_\had$ takes the transition 
$\trntransof u w {\invsqrttwo u (\leftT) + \invsqrttwo u (\rightT)}{\invsqrttwo u (\leftT) - \invsqrttwo u (\rightT)}$.
Substituting $s_1$ and $0s_1$ in the places of the placeholders $(\leftT)$ and $(\rightT)$ respectively,
we obtain the following transition in $\transfH(\aut_\pre)$:
\begin{equation*}
\begin{aligned}
    \initmark & \transcof{\tuple{s_1, u}} w 1
    {\invsqrttwo\tuple{s_1, u}}
    {\invsqrttwo\tuple{s_1, u}}.
\end{aligned}
\end{equation*}
Note that to simplify the presentation, we removed elements having a zero
coefficient. In general, such simplifications must be carried out with care, since
although such elements cannot change the leaf values, they might affect the
heights of trees accepted from such a state.

Similarly, both $s_1$ and $u$ can take a transition over~$a$.
Composing the transition $\transcof{s_1} a 1 {0s_2}{s_2}$ with $\transfH$'s
transition in \cref{eq:transH} in the same way as above, we obtain the following
transition of $\transfH(\aut_\pre)$:
\begin{equation*}
\begin{aligned}
    \initmark & \transcof{\tuple{s_1, u}} a 1
                {\invsqrttwo \tuple{s_2, u}} {-\invsqrttwo \tuple{s_2, u}}.
\end{aligned}
\end{equation*}
Since~$s_2$ and~$u$ are both leaf states in $\aut_\pre$ and $\transfH$ respectively, the state $\tuple{s_2, u}$ is also a~leaf state in~$\transfH(\aut_\pre)$.
To keep our exposition managable, we rename the states of the resulting product by applying
the following map $\{\tuple{s_1, u} \mapsto \alpha, \tuple{s_2, u} \mapsto \beta \}$. 
\begin{equation*}
\begin{aligned}
  &\initmark
    \transcof \alpha w 1 {\invsqrttwo \alpha}{\invsqrttwo \alpha} \quad \initmark
    \transcof \alpha a 1 {\invsqrttwo \beta}{-\invsqrttwo \beta}
\end{aligned}
\end{equation*}

We continue by applying $\transfCX$ to $\transfH(\aut_\pre)$ using the same product construction as above.
The resulting automaton $\transfCX \big( \transfH (\aut_\pre))\big)$ has the following transitions.
We again rename the states of the product using the map
$ \big\{
    \tuple{\alpha, r^0} \mapsto \gamma,
    \tuple{\alpha, s^0} \mapsto \delta,
    \tuple{\alpha, r^1} \mapsto \mu,
    \tuple{\alpha, s^1} \mapsto \varepsilon,
    \tuple{\beta, l} \mapsto \sigma
 \big\}$.
\begin{equation*}
\begin{aligned}
    \initmark & \transcof \gamma w 1 {\invsqrttwo \delta}{\invsqrttwo \varepsilon} \quad  & \initmark &
                \transcof \gamma a 1 {\invsqrttwo \sigma}{- \invsqrttwo \sigma} \\
              & \transcof \delta w 1 {\invsqrttwo \gamma}{\invsqrttwo \gamma}  &&
                \transcof \delta a 1 {\invsqrttwo \sigma}{- \invsqrttwo \sigma} \\
              & \transcof \varepsilon w 1 {\invsqrttwo \mu}{\invsqrttwo \mu}  &&
                \transcof \varepsilon a 1 {- \invsqrttwo \sigma}{\invsqrttwo \sigma} \\
              & \transcof \mu w 1 {\invsqrttwo \varepsilon}{\invsqrttwo \delta }  &&
                \transcof \mu a 1 {- \invsqrttwo \sigma}{ \invsqrttwo \sigma}
\end{aligned}
\end{equation*}

Finally, we compute $\autRes \triangleq \transfH \Big( \transfCX \big( \transfH(\aut_\pre) \big)\Big)$
with the states of the product automaton being renamed using
$\{
    \tuple{\gamma, u}       \mapsto \gamma_{u},
    \tuple{\delta, u}       \mapsto \delta_{u},
    \tuple{\varepsilon, u}  \mapsto \varepsilon_{u},
    \tuple{\mu, u}          \mapsto \mu_{u},
    \tuple{\sigma, u}       \mapsto \sigma_{u}
\}$ ($\sigma_u$ is the only leaf state).
\begin{equation*}
\begin{aligned}
    \initmark & \transcof{\gamma_u} w 1 {\invtwo \delta_u + \invtwo \varepsilon_u}{
                                               \invtwo \delta_u - \invtwo \varepsilon_u} \quad  & \initmark &
              \transcof{\gamma_u} a 1 {0\sigma_u}{\sigma_u} \\
              & \transcof{\delta_u} w 1 {\gamma_u} {0 \gamma_u} &&
                \transcof{\delta_u} a 1 {0 \sigma_u}{\sigma_u} \\
              & \transcof{\varepsilon_u} w 1 {\mu_u}{0\mu_u} &&
                \transcof{\varepsilon_u} a 1 {0 \sigma_u}{- \sigma_u} \\
              & \transcof{\mu_u} w 1 {\invtwo \varepsilon_u + \invtwo \delta_u}{\invtwo \varepsilon_u - \invtwo \delta_u} &&
                \transcof{\mu_u} a 1 {0\sigma_u}{- \sigma_u}
\end{aligned}
\end{equation*}

Having $\autRes$, we check the correctness of our BV circuit against the postcondition SWTA $\aut_\post$ defined in \cref{sec:BV}
using the algorithm in~\cref{sec:colored_equivalence_checking}.


\vspace{-0.0mm}
\section{Complex Number Representation}
\label{sec:complex}
\vspace{-0.0mm}

To ensure accuracy, we restrict our attention to a subset of complex numbers that admit the following algebraic representation, which generalizes the one proposed in~\cite{burgholzer2020advanced}:
\begin{equation}
  \label{eq:encoding}
  z \;=\;
  \smallfrac{1}{\sqrt{2^{k}}}
  \left(
    a_0 + a_1 \omega_m + a_2 \omega_m^2 + \ldots + a_{m-1}\omega_m^{m-1}
  \right),
  \qquad
  a_0,a_1,\ldots,a_{m-1} \in \integers, k \in \naturals
\end{equation}
where $\omega_m = e^{\frac{i\pi} m}$ is the $m$-th root of unity, forming an angle of $\frac{180^\circ}{m}$ with the positive real axis in the complex plane.
The parameter $m$ allows us to control the granularity of complex numbers used
in quantum circuit representations.
Each complex number is thus encoded as a tuple $(a_0, a_1, \ldots, a_{m-1}, k)$.

Although the set defined by~\cref{eq:encoding} constitutes only a countable subset of~$\complex$, it is sufficient to support a universal gate set already when $m = 4$.
Moreover, this algebraic structure enables efficient manipulation of complex numbers.
For instance, multiplication by~$\omega_m$ corresponds to a right circular shift of the first $m$ coefficients, followed by negation of the new leading term:
\begin{equation}
  \omega_m \cdot (a_0, a_1, \ldots, a_{m-1}, k) \;=\; (-a_{m-1}, a_0, \ldots, a_{m-2}, k).
\end{equation}

\vspace{-0.0mm}